\documentclass[11pt]{article}

\usepackage[margin=1in]{geometry}
\usepackage{amsmath,amssymb,amsthm,mathtools}
\usepackage{microtype}
\usepackage{booktabs,tabularx,array}
\usepackage{enumitem}
\usepackage[dvipsnames]{xcolor}
\usepackage{tikz}
\usepackage{quantikz}
\usepackage[ruled,vlined,linesnumbered]{algorithm2e}
\usepackage{hyperref}
\usepackage[capitalize,nameinlink]{cleveref}
\usepackage{aliascnt}
\usepackage{dsfont}
\usepackage[normalem]{ulem}
\usepackage{authblk}

\definecolor{papergray}{RGB}{245,247,249}

\theoremstyle{plain}
\newtheorem{theorem}{Theorem}
\newtheorem{lemma}{Lemma}

\newtheorem{corollary}{Corollary}
\theoremstyle{definition}
\newtheorem{definition}{Definition}

\theoremstyle{remark}

\newcommand{\F}{\mathrm F}
\newcommand{\op}{\mathrm{op}}

\newcommand{\Tr}{\operatorname{Tr}}
\newcommand{\E}{\mathbb E}
\newcommand{\Prb}{\mathbb P}
\newcommand{\eps}{\varepsilon}

\newcommand{\supp}{\operatorname{supp}}
\newcommand{\wt}{\operatorname{wt}}

\providecommand{\ket}[1]{|#1\rangle}
\providecommand{\bra}[1]{\langle #1|}
\newcommand{\ketbra}[2]{\ket{#1}\!\bra{#2}}
\newcommand{\abs}[1]{\left|#1\right|}
\newcommand{\norm}[1]{\left\|#1\right\|}
\newcommand{\cP}{\mathcal P}

\newcommand{\Pnk}{\mathcal P_{n,k}}

\SetKwInput{KwInput}{Input}
\SetKwInput{KwOutput}{Output}
\SetKwInput{KwAccess}{Quantum access}

\title{Hamiltonian Learning and Certification via Eigenphase Engineering}

\author[1,2]{Myeongjin Shin}
\author[3,4,5]{Yu Tong}

\affil[1]{School of Computing, Korea Advanced Institute of Science and Technology (KAIST), Daejeon 34141, Republic of Korea}
\affil[2]{Institute for Quantum Information and Matter, California Institute of Technology, Pasadena, CA 91125, USA}
\affil[3]{Department of Electrical and Computer Engineering, Duke University, Durham, NC 27708, USA}
\affil[4]{Department of Mathematics, Duke University, Durham, NC 27708, USA}
\affil[5]{Duke Quantum Center, Duke University, Durham, NC 27701, USA}

\date{\today}

\begin{document}
\maketitle

\begin{abstract}
We study Hamiltonian learning for \(k\)-local Hamiltonians at large step size, using only forward time evolution at integer multiples of a fixed time interval. For this task, our new technique, which we call eigenphase engineering, achieves Heisenberg-limited learning with a step size that is optimal up to logarithmic factors. For an unknown \(s\)-sparse Hamiltonian with fixed locality \(k\) and \(\|H\|_{\mathrm{op}}\leq\Lambda\), our algorithm recovers its Pauli coefficients to \(\ell_2\) error \(\varepsilon\) in total evolution time \(\widetilde{\mathcal O}(s/\varepsilon)\), matching the best known scaling while allowing a near-maximal step size \(\widetilde{\Theta}(1/\Lambda)\). As an independent contribution, we strengthen the evolution-time lower bound to \(\Omega_k(s^{1-1/(2k)}/\varepsilon)\), placing our algorithm within a factor \(\widetilde{\mathcal O}(s^{1/(2k)})\) of being optimal. Eigenphase engineering departs from approaches based on dynamical approximation or iteratively canceling the unknown Hamiltonian. It encodes energy expectation differences in the perturbative response of eigenphases of the actual controlled evolution, then extracts them through phase estimation and high-order extrapolation. This avoids the need for polynomially shorter evolution steps as the target precision improves. A version using only single-qubit operations retains the same step size with a \(\sqrt n\) overhead in \(\ell_2\) learning time for \(n\) qubits. The method also enables tolerant Hamiltonian certification and estimation of a specified Pauli coefficient, the latter without locality or sparsity assumptions. These results establish eigenphase engineering as a versatile approach to extracting Hamiltonian information with infrequent control.
\end{abstract}

\tableofcontents

\section{Introduction}\label{sec:introduction}

The Hamiltonian of a quantum system encodes its interactions and determines its dynamics. Recovering this Hamiltonian from experimental observations is therefore a fundamental problem in quantum science, with applications ranging from sensing magnetic fields to calibrating quantum processors and validating analog quantum simulators~\cite{giovannetti2011advances,degen2017quantum,GranadeFerrieWiebeCory2012robust}. Over more than a decade, Hamiltonian learning has developed into a broad subject, with protocols that infer interactions from dynamical measurements, stationary states, and thermal states~\cite{GranadeFerrieWiebeCory2012robust,BaireyAradEtAl2019learning,HaahKothariTang2022optimal}. Here we study learning from real-time evolution. A central objective in this setting is Heisenberg-limited learning: estimating the unknown parameters to precision \(\eps\) using total evolution time proportional to \(1/\eps\), up to logarithmic factors. Recent algorithms attain this optimal precision dependence for increasingly general classes of many-body Hamiltonians~\cite{huang2023learning,dutkiewicz2024advantage,bakshi2024structure,MaEtAl2024,HuEtAl2025}.

Evolution time, however, is only one of the resources that determine whether a learning protocol is useful. Quantum control plays an essential role: for broad classes of Hamiltonians, including thermalizing systems under suitable assumptions, the absence of control obstructs Heisenberg-limited learning~\cite{dutkiewicz2024advantage}. Single-qubit operations are especially attractive because they can often be implemented with very high fidelity. The frequency of control is equally important. Even a protocol using only single-qubit gates can become difficult to implement if its pulse intervals must decrease rapidly as the target precision improves. Finite gate durations and imperfect pulses motivate balancing the amount and complexity of control against the statistical efficiency of learning.

We formalize the timing constraint through a \emph{step size} \(\tau\): every interval of unknown Hamiltonian evolution must be a nonnegative integer multiple of \(\tau\). A larger step size permits less frequent control but also restricts the information accessible to the learner. Zhou and Gong~\cite{ZhouGong2026} established a fundamental obstruction in this fixed-grid model: for Hamiltonians satisfying \(\norm{H}_{\op}\leq\Lambda\), uniform learning to arbitrarily small error requires \(\tau<\pi/\Lambda\). Beyond this scale, distinct Hamiltonians can generate identical unitaries at every allowed time, making them indistinguishable even with arbitrary intervening controls. Thus \(1/\Lambda\) is the largest possible order of the step size, or equivalently \(\Omega(\Lambda)\) is a lower bound on the inverse step size. This motivates the question addressed in this work:
\begin{quote}
\emph{How efficiently can we learn a local Hamiltonian when the step size is as large as this fundamental limit allows, up to logarithmic factors?}
\end{quote}
To the best of our knowledge, the algorithms in this work, through a technique we call \emph{eigenphase engineering}, are the first for learning and certifying local Hamiltonians that achieve both Heisenberg-limited total evolution time and near-maximal step size.
Remarkably, even with this near-maximal step size, we match the best known total evolution time for learning sparse local Hamiltonians in the \(\ell_2\)-norm, up to logarithmic factors. 
Below we will discuss how this is achieved in detail.

\subsection{Problem setup}

We consider
\begin{equation}\label{eq:intro-H}
H=\sum_{P_a\in\cP_{n,k}}\lambda_aP_a,\quad \cP_{n,k}:=\{P_a\in\{I,X,Y,Z\}^{\otimes n}:1\le\wt(P_a)\le k\},
\end{equation}
with at most $s$ nonzero coefficients and $\norm{H}_{\op}\le\Lambda$.
We consider the physically relevant constant locality regime $k=O(1)$, although the formal statements retain their explicit dependence on $k$. 
We do not impose geometric locality or bounded interaction degree. Therefore, in principle, there can exist a qubit that all $s$ Hamiltonian terms act on.
Importantly, we \emph{do not} assume knowledge of which $P_a$ has a nonzero contribution, similar to the setting considered in \cite{MaEtAl2024}. Our goal is to recover the coefficients $\lambda_a$. 
We will write all these coefficients as a vector $\lambda=(\lambda_a)_{a}$, and our estimate will be a vector $\hat{\lambda}=(\hat{\lambda}_a)_a$.
To quantify the accuracy of the coefficient recovery, we consider both the $\ell^2$-error and the $\ell^\infty$-error.
\begin{equation}
    \|\hat{\lambda}-\lambda\|_2=\left(\sum_a |\hat{\lambda}_a-\lambda_a|^2\right)^{1/2},\quad \|\hat{\lambda}-\lambda\|_\infty=\max_{a}|\hat{\lambda}_a-\lambda_a|.
\end{equation}
Our focus will be on the former, which is an operationally meaningful error metric that bounds how well the learned Hamiltonian can be used to predict actual Hamiltonian dynamics for random initial states \cite[Sec.~7.2, Eq.~(108)]{MaEtAl2024}.

\medskip
\noindent\textbf{Access model.}
We are given oracle access to the forward time-evolution unitary $U(t) = e^{-iHt}$. We can make queries to $\{U(t_i)\}_{i=1}^{\mathcal N}$ for nonnegative times $t_i$. We do not assume access to backward time-evolution, nor do we assume access to controlled-$U(t)$.

\medskip
\noindent\textbf{Metrics.}
We evaluate the performance of a learning algorithm using the following metrics:
\begin{itemize}
    \item the \emph{total evolution time} $T_{\rm tot} = \sum_{i=1}^\mathcal{N} t_i$,
    
    \item the \emph{step size} $\tau$: every $t_i$ is a nonnegative integer multiple of $\tau$.
    
    \item the \emph{number of queries} $\mathcal{N}$.

    \item the \emph{classical postprocessing time}.
\end{itemize}
Here $t_1,\dots,t_{\mathcal N}$ denote the evolution times used by the algorithm.
We also consider two related tasks under the same access model. In \emph{tolerant Hamiltonian certification}, we are given a classical description of a traceless \(k\)-local reference Hamiltonian \(H_0\) and seek to distinguish \(\|H-H_0\|_{\mathrm F}\leq\eps/12^k\) from \(\|H-H_0\|_{\mathrm F}\geq\eps\), promised that one of these conditions holds. Here \(\|\cdot\|_{\mathrm F}\) denotes the normalized Frobenius norm, which equals the \(\ell_2\)-norm of the Pauli coefficient vector. In \emph{single-coefficient learning}, we are given a nonidentity Pauli string \(P\) and seek an estimate \(\widehat\lambda_P\) of \(\lambda_P=2^{-n}\operatorname{Tr}(PH)\) satisfying \(|\widehat\lambda_P-\lambda_P|\leq\eps\). For this task, we require only that \(H\) be traceless and satisfy \(\|H\|_{\mathrm{op}}\leq\Lambda\), without imposing locality or sparsity assumptions. Both tasks must succeed with high probability.

\subsection{Main results}

Throughout this subsection, \(k=O(1)\), and the notation \(\widetilde O\), \(\widetilde\Omega\), and \(\widetilde\Theta\) suppresses polylogarithmic factors, including the dependence on the inverse failure probability. All algorithmic guarantees hold with high probability, use only forward evolution under the unknown Hamiltonian, and achieve a common step size \(\tau=\widetilde\Theta(1/\Lambda)\). This step size matches the optimal limit in \cite{ZhouGong2026} up to logarithmic factors. Here, optimality is with respect to algorithms that query \(e^{-iHt}\) only at times \(t\) that are nonnegative integer multiples of \(\tau\).

\subsubsection{Hamiltonian learning}

\begin{table}[!t]
    \centering
    \footnotesize
    \setlength{\tabcolsep}{2.2pt}
    \renewcommand{\arraystretch}{1.35}
    \begin{tabular}{
        @{}
        >{\raggedright\arraybackslash}p{0.168\textwidth}
        >{\centering\arraybackslash}p{0.2\textwidth}
        >{\centering\arraybackslash}p{0.17\textwidth}
        >{\centering\arraybackslash}p{0.22\textwidth}
        >{\centering\arraybackslash}p{0.12\textwidth}
        @{}
    }
        \toprule
        Reference
        &
        \(T_2\)
        &
        \(T_\infty\)
        &
        \shortstack{step size / \\ time resolution}
        &
        \shortstack{Only \\ single-qubit}
        \\
        \midrule

        Bakshi et al.~\cite{bakshi2024structure}
        &
        \(\widetilde{\mathcal O}(s^{3/2}/\eps)^{\dagger}\)
        &
        \textcolor{blue}{\(\widetilde{\mathcal O}(s/\eps)\)}
        &
        \shortstack{\(\tau_{\rm step}^{-1}=\mathcal{O}(s\Lambda)\)}
        &
        No
        \\
        \addlinespace[2pt]

        Ma et al.~\cite{MaEtAl2024}
        &
        \textcolor{blue}{\(\widetilde{\mathcal O}(s/\eps)\)}
        &
        \textcolor{blue}{\(\widetilde{\mathcal O}(s/\eps)^{\dagger}\)}
        &
        \shortstack{\(\tau_{\rm step}^{-1}=\widetilde{\mathcal O}(\Lambda^2/\eps)\)}
        &
        No
        \\
        \addlinespace[2pt]

        Abbas et al.~\cite{AbbasEtAl2025}
        &
        \(\widetilde{\mathcal O}(s^{3/2}/\eps)^{\dagger}\)
        &
        \textcolor{blue}{\(\widetilde{\mathcal O}(s/\eps)\)}
        &
        \shortstack{\(t_{\min}^{-1}={\mathcal O}(s\Lambda^{3/2}/\eps^{1/2})\)}
        &
        No
        \\
        \addlinespace[3pt]

        Shin et al.~\cite{shin2026heisenberg}
        &
        \(\widetilde{\mathcal O}(s^{\mathcal{O}(K)}/\eps)^{\dagger}\)
        &
        \(\widetilde{\mathcal O}(s^{\mathcal{O}(K)}/\eps)\)
        &
        \shortstack{\(t_{\min}^{-1}=\Theta(s^{1/K}\Lambda)\)}
        &
        No
        \\
        \addlinespace[3pt]

        Zhou and Gong~\cite{ZhouGong2026}
        &
        \(\widetilde{\mathcal O}(s\Lambda/\eps^2)^{\dagger}\)
        &
        \(\widetilde{\mathcal O}(\Lambda/\eps^2)\)
        &
        \shortstack{\textcolor{blue}{\(\tau_{\rm step}^{-1}=\Theta(\Lambda)\)}}
        &
        \textcolor{blue}{Yes}
        \\

        \midrule

        \textbf{\Cref{thm:near-upper}}
        &
        \textcolor{blue}{\(\widetilde{\mathcal O}(s/\eps)\)}
        &
        \textcolor{blue}{\(\widetilde{\mathcal O}(s/\eps)^{\dagger}\)}
        &
        \shortstack{\textcolor{blue}{\(\tau_{\rm step}^{-1}=\widetilde{\Theta}(\Lambda)\)}}
        &
        No
        \\
        \addlinespace[2pt]

        \textbf{\Cref{thm:single-qubit-learning}}
        &
        \(\widetilde{\mathcal O}(s\sqrt n/\eps)\)
        &
        \textcolor{blue}{\(\widetilde{\mathcal O}(s/\eps)\)}
        &
        \shortstack{\textcolor{blue}{\(\tau_{\rm step}^{-1}=\widetilde{\Theta}(\Lambda)\)}}
        &
        \textcolor{blue}{Yes}
        \\

        \bottomrule
    \end{tabular}

    \caption{
    Comparison of Hamiltonian-learning algorithms. We assume that the Hamiltonian is $s$-sparse, $k$-local and $k=\mathcal{O}(1)$, with \(\norm{H}_{\op}\leq\Lambda\). Entries marked \(\dagger\) follow by norm conversion; converting $\ell_\infty$ error to $\ell_2$ error uses accuracy $\Theta(\eps/\sqrt{s})$ and retains at most $s$ estimated coefficients. \(T_2\) and \(T_\infty\) denote the total evolution times required to learn the Pauli coefficient vector of the Hamiltonian to additive error \(\eps\) in the \(\ell_2\)- and \(\ell_\infty\)-norms, respectively. $\tau_{\rm step}$ denotes step size, which is also known as common time-grid spacing, whereas $t_{\rm min}$ denotes the minimum time resolution for protocols that may not use a uniform time-grid. Bounds are stated for constant failure probability. For Abbas et al., the displayed resolution is for \(T_\infty\); for \(T_2\), multiply \(t_{\min}^{-1}\) by \(s^{1/4}\). In Shin et al., \(K\geq1\) is a fixed tradeoff parameter, distinct from locality \(k\). The Ma et al. step size follows by retaining \(\norm{H}_{\op}\) in \cite[Appendix~A, Eq.~(113)]{MaEtAl2024}, rather than directly rescaling their stated theorem. The final column is marked ``Yes'' only when all state preparations, known controls, and measurements are tensor products of single-qubit operations, without ancillary qubits or initially entangled states. The state-of-the-art results up to polylogarithmic factors are colored blue.
    }
    \label{tab:learning-comparison}
\end{table}

The common building block of our algorithms is \emph{eigenphase engineering}, which uses discrete controls to encode energy expectation differences into measurable eigenphases. Its operational guarantee is the following.

\begin{theorem}[Eigenphase engineering, informal]\label{thm:intro-eigenphase}
Let \(H\) satisfy \(\norm{H}_{\op}\leq\Lambda\). Given a known unitary \(W\) and its inverse, define \(\ket{x}=W\ket{0^n}\) and \(\ket{y}=W\ket{10^{n-1}}\). The energy expectation difference
\begin{equation}
    \Delta_H(x,y)=\bra{x}H\ket{x}-\bra{y}H\ket{y}
\end{equation}
can be estimated to additive error \(\eps\) using
\begin{equation}
T_{\rm tot}=\widetilde{\mathcal O}(1/\eps),
\qquad
\tau=\widetilde\Theta(1/\Lambda).
\end{equation}
If \(H\) is \(k\)-local and \(W\) is a product of single-qubit unitaries, the same asymptotic guarantees hold using only single-qubit preparations, controls, and measurements, without ancillary qubits. See \Cref{thm:energy-gap-subroutine,thm:energy-single-qubit} for the formal statements.
\end{theorem}

These energy differences provide linear information about the Hamiltonian coefficients. Combining their estimation with sparse recovery through compressed sensing \cite{donoho2006compressed,candes2006robust} gives a learning algorithm that does not require prior knowledge of which Pauli terms have nonzero coefficients.

\begin{theorem}[Near-optimal Hamiltonian learning, informal]\label{thm:intro-learning}
Let \(H=\sum_{P\in\Pnk}\lambda_P P\) be traceless, \(s\)-sparse, and \(k\)-local, with \(\norm{H}_{\op}\leq\Lambda\) and $k=\mathcal{O}(1)$. There is an algorithm that outputs an estimate \(\widehat\lambda\) satisfying
\begin{equation}
    \norm{\widehat\lambda-\lambda}_2\leq\eps
\end{equation}
using total evolution time \(\widetilde{\mathcal O}(s/\eps)\) and step size \(\widetilde\Theta(1/\Lambda)\). See \Cref{thm:near-upper} for the formal statement.
\end{theorem}

This matches the best known total evolution time for coefficient \(\ell_2\) recovery, achieved by Ma et al.~\cite{MaEtAl2024}, at a near-maximal step size.  By comparison, Ref.~\cite{MaEtAl2024} requires a step size of $\tau=\widetilde{\Omega}(\eps/\Lambda^2)$ for their algorithm.\footnote{This step size can be derived from Appendix~A of \cite{MaEtAl2024}.} Bakshi et al.~\cite{bakshi2024structure} already obtain Heisenberg scaling with a precision-independent step size, scaling as $\tau=\Omega(1/(s\Lambda))$ under the norm bound used here.  Compared to it, our step size removes the $1/s$ factor.  The total evolution time not only almost matches the best known result, it is in fact very close to the optimal total evolution time one can achieve without any restriction on the step size, as proved in the lower bound result below:

\begin{theorem}[Learning lower bound, informal]\label{thm:intro-lower-bound}
Suppose \(3^k\binom{n}{k}\geq s\) and \(0<\eps\leq c_k\Lambda s^{-1/(2k)}\), for a sufficiently small constant \(c_k>0\). Any algorithm that learns every traceless, \(s\)-sparse, \(k\)-local Hamiltonian with \(\norm{H}_{\op}\leq\Lambda\) to coefficient \(\ell_2\) error \(\eps\), with success probability at least \(2/3\), requires
\begin{equation}
    T_{\rm tot}=\Omega_k\!\left(\frac{s^{1-1/(2k)}}{\eps}\right).
\end{equation}
This bound holds even with arbitrary controls, ancillary systems, adaptive measurements, and unrestricted evolution times. See \Cref{thm:lower-bound} for the formal statement.
\end{theorem}

In this parameter regime, \Cref{thm:intro-learning} is within a factor \(\widetilde{\mathcal O}(s^{1/(2k)})\) of the evolution-time lower bound.  
For comparison, the $\ell^2$ learning lower bound in Ref.~\cite{MaEtAl2024} is $\Omega(\sqrt{s}/\eps)$.\footnote{The $\ell^1$ learning lower bound in \cite[Theorem~4.5]{AbbasEtAl2025} can be converted to a similar $\Omega(\sqrt{s}/\eps)$ $\ell^2$ learning lower bound through norm conversion.} Our lower bound therefore represents a factor of $s^{1/2-1/(2k)}$ improvement over the state of the art.

The single-qubit version of eigenphase engineering further allows us to restrict all trusted quantum operations to single qubits.
A key challenge in eigenphase engineering with single-qubit controls is keeping the relevant eigenphases nondegenerate and sufficiently separated, properties ensured by multiqubit controls in the construction we just introduced.
To address this challenge, for the single-qubit version of eigenphase engineering, we use Diophantine approximation and a low-degree truncation of the perturbation expansion to handle small eigenphase gaps, allowing us to restrict all trusted quantum operations to single qubits.

\begin{theorem}[Learning with single-qubit operations, informal]\label{thm:intro-single-qubit}
Under the assumptions of \Cref{thm:intro-learning}, there are algorithms using only product-state preparation, single-qubit controls and measurements, and no ancillary qubits that achieve
\begin{equation}
\begin{aligned}
\norm{\widehat\lambda-\lambda}_\infty\leq\eps
&\quad\text{with}\quad
 T_{\rm tot}=\widetilde{\mathcal O}(s/\eps),\\
\norm{\widehat\lambda-\lambda}_2\leq\eps
&\quad\text{with}\quad
 T_{\rm tot}=\widetilde{\mathcal O}(s\sqrt n/\eps).
\end{aligned}
\end{equation}
Both algorithms use step size \(\widetilde\Theta(1/\Lambda)\). See \Cref{thm:single-qubit-learning} for the formal statement.
\end{theorem}

Thus, restricting the available operations to single qubits incurs only a factor of \(\sqrt n\) in our \(\ell_2\) learning cost, while retaining Heisenberg scaling and near-optimal step size. \Cref{tab:learning-comparison} summarizes the learning guarantees in this work and compares them with a selection of existing results in the literature.

\subsubsection{Certification and learning one coefficient}

\begin{table}[!t]
    \centering
    \footnotesize
    \setlength{\tabcolsep}{2.2pt}
    \renewcommand{\arraystretch}{1.35}
    \begin{tabular}{
        @{}
        >{\raggedright\arraybackslash}p{0.168\textwidth}
        >{\centering\arraybackslash}p{0.2\textwidth}
        >{\centering\arraybackslash}p{0.17\textwidth}
        >{\centering\arraybackslash}p{0.22\textwidth}
        >{\centering\arraybackslash}p{0.12\textwidth}
        @{}
    }
        \toprule
        Reference
        &
        \(T_{\rm tot}\)
        &
        \shortstack{step size / \\ time resolution}
        &
        \shortstack{HL limit}
        &
        \shortstack{only \\ single-qubit}
        \\
        \midrule

        Using learning~\cite{MaEtAl2024}
        &
        \(\widetilde{\mathcal O}(n^k/\eps)\)
        &
        \shortstack{\(\tau_{\rm step}^{-1}=\widetilde{\mathcal O}(\Lambda^2/\eps)\)}
        &
        \textcolor{blue}{Yes}
        &
        No
        \\
        \addlinespace[2pt]

        Gao et al.~\cite[Theorem~4.4]{GaoEtAl2025}
        &
        \textcolor{blue}{\({\Theta}(1/\eps)\)}
        &
        \(t_{\min}^{-1}=\mathcal O(n^{3k/2}\Lambda)\)
        &
        \textcolor{blue}{Yes}
        &
        No
        \\
        \addlinespace[2pt]

        Gao et al.~\cite[Theorem~5.5]{GaoEtAl2025}
        &
        \(\widetilde{\mathcal O}(n^{3k/2}/\eps)\)
        &
        \(t_{\min}^{-1}=\mathcal O(n^{9k/4}\Lambda)\)
        &
        \textcolor{blue}{Yes}
        &
        No
        \\
        \addlinespace[2pt]

        Bluhm et al.~\cite{BluhmEtAl2026}
        &
        \textcolor{blue}{\({\Theta}(1/\eps)\)}
        &
        \shortstack{\(t_{\min}^{-1}=\mathcal{O}(\Lambda^{3/2}/\sqrt{\eps})^{\ddagger}\)}
        &
        \textcolor{blue}{Yes}
        &
        No
        \\
        \addlinespace[2pt]

        Flammia et al.~\cite{flammia2026autonomous}
        &
        \(\mathcal O(n\Lambda/\eps^2)\)
        &
        \textcolor{blue}{\shortstack{\(\tau_{\rm step}^{-1}=\Theta(\Lambda)\)}}
        &
        No
        &
        \textcolor{blue}{Yes}
        \\

        \midrule

        \textbf{\Cref{thm:cert}}
        &
        \textcolor{blue}{\(\widetilde{\Theta}(1/\eps)\)}
        &
        \shortstack{\textcolor{blue}{\(\tau_{\rm step}^{-1}=\widetilde{\Theta}(\Lambda)\)}}
        &
        \textcolor{blue}{Yes}
        &
        No
        \\
        \addlinespace[2pt]

        \textbf{\Cref{thm:cert-single}}
        &
        \(\widetilde{\mathcal O}(n^{3/2}/\eps)\)
        &
        \shortstack{\textcolor{blue}{\(\tau_{\rm step}^{-1}=\widetilde{\Theta}(\Lambda)\)}}
        &
        \textcolor{blue}{Yes}
        &
        \textcolor{blue}{Yes}
        \\
        \bottomrule
    \end{tabular}

    \caption{
    Comparison of Hamiltonian-certification algorithms. Both \(H\) and the known reference \(H_0\) are traceless and \(k\)-local, with \(k=\mathcal O(1)\) and \(\norm{H}_{\op},\norm{H_0}_{\op}\leq\Lambda\). Distances are in normalized Frobenius norm, \(\eps\) is the far-case threshold, and bounds are stated for constant failure probability in the nontrivial regime $0<\eps\leq\Lambda$. $\tau_{\rm step}$ denotes step size, which is also known as common time-grid spacing, whereas \(t_{\min}\) is the minimum absolute duration of an unknown-evolution query. 
    The learning baseline applies Ma et al. followed by comparison with \(H_0\), using the step-size refinement described in \Cref{tab:learning-comparison}. The $\ddagger$ entry uses a modification of Bluhm et al. that aborts if a sampled evolution time is below $a/\eps$, for a sufficiently small constant $a>0$; this preserves constant success probability and the stated evolution-time scaling. Without this modification Bluhm et al. may require arbitrarily small time resolution. Gao et al.'s bounds specialize their sparsity parameter to \(O(n^k)\) and coefficient bound to \(\Lambda\); the timing entries follow from the query durations in Algorithms~1 and~4, respectively. Their coherent algorithm (Algorithm~1) requires controlled and inverse unknown evolution; their ancilla-free algorithm (Algorithm~4) uses potentially entangling known controls. The HL limit column is marked ``Yes'' when the total evolution time achieves the Heisenberg limit up to logarithm factors. The final column is marked ``Yes'' only when all state preparations, known controls, and measurements are tensor products of single-qubit operations, without ancillary qubits or initially entangled states. The state-of-the-art bounds up to polylogarithmic factors are colored blue.
    }
    \label{tab:cert-comparison}
\end{table}

The same energy-gap primitive also supports tasks that require less information than full reconstruction. First, we can certify whether an unknown local Hamiltonian is close to a known reference, with a tolerance gap between the close and far cases.

\begin{theorem}[Tolerant Hamiltonian certification, informal]\label{thm:intro-certification}
Let \(H\) and \(H_0\) be traceless, \(k\)-local Hamiltonians, where \(H_0\) is known and \(\norm{H}_{\op}\leq\Lambda\). Promised that one of the following conditions holds, we can distinguish
\begin{equation}
\norm{H-H_0}_{\F}\leq\frac{\eps}{12^k}
\qquad\text{from}\qquad
\norm{H-H_0}_{\F}\geq\eps,
\end{equation}
where \(\norm{\cdot}_{\F}\) is the normalized Frobenius norm, using total evolution time \(\widetilde{\mathcal O}(1/\eps)\) and step size \(\widetilde\Theta(1/\Lambda)\). A version using only product-state preparation, single-qubit controls and measurements, and no ancillary qubits has total evolution time \(\widetilde{\mathcal O}(n^{3/2}/\eps)\) with the same asymptotic step size. See \Cref{thm:cert,thm:cert-single} for the formal statements.
\end{theorem}

Our certification algorithms are tolerant in the sense that they are guaranteed to accept Hamiltonians that may not be the target Hamiltonian exactly but are within \(\varepsilon/12^k\) distance, with high probability.
We note that no sparsity assumption is needed for the actual Hamiltonian in certification. With multiqubit controls, the evolution time has no polynomial dependence on system size, so certification can be substantially less costly than learning the entire Hamiltonian. We can also estimate a specified coefficient without imposing locality or sparsity on the unknown Hamiltonian. In Table~\ref{tab:cert-comparison} we compare our certification result with a selection of existing results in the literature.

The same energy-gap estimation subroutine also allows us to learn a specified Pauli coefficient without reconstructing the full Hamiltonian, and this task requires neither locality nor sparsity assumptions.
\begin{theorem}[Learning one Hamiltonian coefficient, informal]\label{thm:intro-single-coefficient}
Let \(H\) be an arbitrary traceless \(n\)-qubit Hamiltonian satisfying \(\norm{H}_{\op}\leq\Lambda\), and let \(P\) be a specified nonidentity Pauli string. There is an algorithm that estimates
\begin{equation}
    \lambda_P=2^{-n}\Tr(PH)
\end{equation}
to additive error \(\eps\), using total evolution time \(\widetilde{\mathcal O}(1/\eps)\) and step size \(\widetilde\Theta(1/\Lambda)\). The algorithm allows ancillary qubits and multiqubit operations. See \Cref{thm:single-coeff} for the formal statement.
\end{theorem}

Applying \Cref{thm:intro-single-coefficient} to all Pauli strings of weight at most \(k\) also yields agnostic learning of the Frobenius-nearest \(k\)-local approximation to an arbitrary Hamiltonian, with coefficient \(\ell_\infty\) error \(\eps\) and total evolution time \(\widetilde{\mathcal O}(n^k/\eps)\). These applications are developed in \Cref{sec:applications}.

\subsection{Comparison with previous works.}

Our main contribution is to combine a near-optimal common step size with the best previous total-evolution-time scalings for sparse Hamiltonian learning and local Hamiltonian certification. As summarized in \Cref{tab:learning-comparison,tab:cert-comparison}, our algorithms use \(\tau_{\rm step}=\widetilde\Theta(1/\Lambda)\), within polylogarithmic factors of the largest uniformly identifiable fixed-grid scale \(O(1/\Lambda)\)~\cite[Appendix~F]{ZhouGong2026}. This is stronger than a lower bound on the minimum query duration: every unknown-evolution query uses an integer multiple of the same step size, whereas a guarantee \(t_i\geq t_{\min}\) permits incommensurate query times. Throughout this comparison, we fix \(k=O(1)\) and suppress polylogarithmic factors.

For learning, \Cref{thm:near-upper} achieves coefficient \(\ell_2\) error \(\eps\) in total evolution time \(\widetilde O(s/\eps)\), matching the bound of Ma et al.~\cite{MaEtAl2024}. For coefficient \(\ell_\infty\) error, \Cref{thm:single-qubit-learning} matches the \(\widetilde O(s/\eps)\) bounds of Bakshi et al.~\cite{bakshi2024structure} and Abbas et al.~\cite{AbbasEtAl2025}, while using only product-state preparation, single-qubit controls and measurements, and no ancillary qubits. These efficiencies are attained without the polynomial dependence on precision in the pulse intervals of Refs.~\cite{MaEtAl2024} and \cite{AbbasEtAl2025}. While Ref.~\cite{bakshi2024structure} already obtains precision-independent timing, our result further removes a $1/s$ factor from the required step size while retaining their total-evolution-time scaling.

Prior algorithms that allow similarly coarse timing exhibit a different efficiency tradeoff. Shin et al.~\cite{shin2026heisenberg} obtain $t_{\min}=\Theta(s^{-1/K}/\Lambda)$ with total evolution time $\widetilde{\mathcal O}(s^{\mathcal O(K)}/\eps)$ for a chosen integer $K\geq1$. Taking $K=\Theta(\log s)$ gives constant $t_{\min}=\Theta(1/\Lambda)$ at quasipolynomial cost in $s$. Our results retain linear dependence on \(s\) at a near-optimal common step size. Zhou and Gong~\cite{ZhouGong2026} already allow a common grid of spacing \(\Theta(1/\Lambda)\), with \(T_\infty=\widetilde O(\Lambda/\eps^2)\) and, by norm conversion, \(T_2=\widetilde O(s\Lambda/\eps^2)\). Our general-control \(\ell_2\) learner improves the latter bound by a factor \(\Lambda/\eps\), up to logarithms. Our single-qubit \(\ell_\infty\) learner also improves their bound in the high-accuracy regime \(\eps\ll\Lambda/s\). 
Their protocol requires no interleaved controls, whereas ours uses discrete controls to obtain Heisenberg scaling. The improvement is therefore in evolution time, not a reduction in every experimental resource.

For tolerant certification, \Cref{thm:cert} attains \(\widetilde O(1/\eps)\) total evolution time, matching the optimal precision dependence of Gao et al.~\cite{GaoEtAl2025} and Bluhm et al.~\cite{BluhmEtAl2026} for a constant relative tolerance gap, now with a common step size \(\widetilde\Theta(1/\Lambda)\). Our algorithm uses only forward unknown evolution, unlike the coherent algorithm of Gao et al., which requires controlled and inverse unknown evolution. With only single-qubit operations, \Cref{thm:cert-single} gives \(\widetilde O(n^{3/2}/\eps)\) total evolution time. By comparison, the autonomous protocol of Flammia et al.~\cite{flammia2026autonomous} uses \(O(n\Lambda/\eps^2)\) time at step size \(\Theta(1/\Lambda)\), without interleaved controls. Our single-qubit protocol improves the evolution-time bound when \(\eps\ll\Lambda/\sqrt n\), up to logarithms, and accepts a constant-fraction \(\eps\)-neighborhood of the reference at fixed \(k\), rather than a neighborhood of radius \(O(\eps/\sqrt n)\).

The method in this work to attain the Heisenberg limit differs significantly from previous works. We extract Hamiltonian information from the eigenphases of the actual controlled evolution, rather than first simulating a simpler Hamiltonian by reshaping or learned cancellation~\cite{MaEtAl2024,bakshi2024structure,shin2026heisenberg}. Our eigenphase-engineering subroutine combines coherent phase amplification with high-order extrapolation on a common grid, so increasing the estimation accuracy does not require polynomially shorter evolution steps (\Cref{thm:energy-gap-subroutine}). Its single-qubit implementation requires a new perturbative argument: locality restricts which states can contribute at each order, and suitable irrational rotation angles control the relevant small eigenphase gaps without requiring a system-size-independent global gap (\Cref{thm:energy-single-qubit}). This subroutine supplies linear energy-gap measurements for sparse recovery and discrepancy tests for certification, providing one framework for both tasks. On the lower-bound side, \Cref{thm:lower-bound} uses a many-hypothesis packing and a truncated-Dyson dimension bound, rather than reducing the problem to pairwise discrimination of a single coefficient, as done in \cite{MaEtAl2024,AbbasEtAl2025}.

\subsection{Overview of the techniques.}

Many existing Heisenberg-limited protocols use dynamical decoupling or Hamiltonian reshaping to approximate evolution under a simpler effective Hamiltonian~\cite{huang2023learning,MaEtAl2024}. Controlling the accumulated approximation error over increasingly long experiments then forces the pulse interval to shrink with the desired precision. Methods that do not suffer from this limitation adopt a recursive approach to gradually refine Hamiltonian estimates and partially cancel its effect, as is done in \cite{bakshi2024structure,dutkiewicz2024advantage}. 
Our methods take a different route. Rather than approximating a target Hamiltonian, we instead engineer isolated eigenphases of the actual controlled evolution and read the desired information from their perturbative response. We call this method \emph{eigenphase engineering}.

The central primitive is a subroutine for the energy difference \(\bra{x}H\ket{x}-\bra{y}H\ket{y}\) between two known orthogonal states. This subroutine is first explicitly constructed and then employed for local Hamiltonian learning and certification. 
In this subroutine, we construct a known unitary $V$ whose relevant eigenvalues are separated from the remainder of the spectrum by constant gaps (see \eqref{eq:V}), and use it together with the unknown Hamiltonian evolution to build unitaries \(F_+(t)\) and \(F_-(t)\) in \eqref{eq:Fpm}. The gaps result in analytically varying eigenphases for \(F_+(t)\) and \(F_-(t)\), whose first derivatives at \(t=0\) are precisely the two energy expectations.
Moreover, \Cref{lem:analytic} shows that  combining the two signs gives the even expansion \eqref{eq:def-delta-a}. Consequently, repeated applications of \(F_+(t)\) and \(F_-(t)\) coherently amplify the desired phase while the eigenvector approximation error remains bounded independently of the repetition number, as described in \Cref{lem:estimate-f}. Robust phase estimation \cite{KimmelLowYoder2015} then estimates the finite-\(t\) quantity \(f(t)\) defined in \eqref{eq:def-f} with Heisenberg-limited scaling. Finally, a logarithmic-order extrapolation from several commensurate values of \(t\) removes the perturbative bias. This last step is essential to ensuring that all evolution times lie on one grid with spacing \(\widetilde\Theta(1/\Lambda)\), while retaining total evolution time \(\widetilde{\mathcal O}(1/\eps)\). The resulting subroutine is stated in \Cref{thm:energy-gap-subroutine}.

To reconstruct an unknown sparse Hamiltonian, we employ this subroutine on random Pauli product states. The gap function \(G(\sigma,x,r)\) in \eqref{eq:gap-def} is a linear function of the unknown Pauli coefficients, and the normalized functions appearing in \eqref{eq:linear-measurement} form a bounded orthonormal system. Thus the quantum part of the algorithm supplies noisy random linear measurements, while the classical part recovers the Hamiltonian coefficients through compressed sensing. The restricted-isometry estimate in \Cref{lem:RIP} implies that \(\widetilde{\mathcal O}(s)\) such measurements suffice even though the support of the Hamiltonian is unknown. Combining this sample bound with the energy-gap subroutine gives the total evolution time \(\widetilde{\mathcal O}(s/\eps)\) in \Cref{thm:near-upper}. 

Restricting every known operation to single-qubit gates creates a qualitatively new obstacle. Unlike the multiqubit control in \eqref{eq:V}, a product control cannot isolate the probe eigenphases from the entire orthogonal complement by a system-size-independent spectral gap. To see the obstruction, write $e^{i\alpha_j}$ for the ratio of the two eigenvalues of the $j$-th single-qubit control. Eigenphase differences between product eigenstates are subset sums of these angles modulo $2\pi$. For $n\geq2$, consider only qubits $2,\ldots,n$ and the $n$ partial sums $0,\alpha_2,\alpha_2+\alpha_3,\ldots,\sum_{j=2}^n\alpha_j$ modulo $2\pi$. Two of these points have circular distance at most $2\pi/n$. Their difference is the sum over a nonempty subset of $\{2,\ldots,n\}$ whose phase is within $\mathcal O(1/n)$ of $2\pi\mathbb Z$. Flipping these qubits in $\ket{e_0}$ gives a state outside $\operatorname{span}\{\ket{e_0},\ket{e_1}\}$, with an eigenphase this close to that of $\ket{e_0}$. Thus, in this example, the eigenphase gap between the probe subspace \(\operatorname{span}\{|e_0\rangle,|e_1\rangle\}\) and its orthogonal complement is at most \(O(1/n)\).
As a result, ordinary perturbation theory cannot be applied using a global constant gap.

To address this obstacle, we choose irrational single-qubit rotation angles satisfying a Diophantine lower bound on how closely integer multiples of the rotation angle can approach an integer, which separates the relevant eigenphase from states in a Hamming ball around the probe state. 
This restricted separation is useful due to locality, because an order-\(r\) perturbative term of a \(k\)-local Hamiltonian can change at most \(kr\) bits, as formalized in \Cref{lem:single-taylor-locality}. For the analysis, we introduce an auxiliary unitary that agrees with the physical product control inside this ball but moves the exterior spectrum away. Perturbation theory for the auxiliary unitary produces an approximate eigenpair for the physical circuit, with the approximation error appearing only at the $(L+1)$-th order, as shown in \Cref{lem:single-physical-defect}. Choosing the order logarithmically in the target precision keeps this approximation error small even after the long repetitions required by phase estimation, as summarized in \Cref{lem:estimate-fL}. The same extrapolation argument as in the multi-qubit control setting then yields \Cref{thm:energy-single-qubit}, with only polylogarithmic losses in the step size and no polynomial dependence on the system size.

Information-theoretic ideas underlie the improved lower bound in \Cref{thm:lower-bound}. Using a standard Hamming-packing estimate from coding theory~\cite[Theorem~4.2.1 and Section~4.2.1]{guruswami2012essential} (adapted to our setting in \Cref{lem:greedy-Hamming-packing}) together with the matrix Khintchine inequality, we construct an exponentially large family of sparse local Hamiltonians whose coefficient vectors are well separated while their operator norms remain small, as shown in \eqref{eq:lower-bound-hard-family}. Unlike previous lower bounds, which ultimately reduce learning to pairwise discrimination in a single coefficient using Assouad's lemma or a single-coefficient reduction~\cite[Theorem~11]{MaEtAl2024}\cite[Theorem~4.5]{AbbasEtAl2025}, our proof analyzes the entire hard family simultaneously. To accommodate fully adaptive protocols, we use the tree representation developed in \cite{chen2022exponential}. Truncating the Dyson expansion of a linear map representing one purified experiment at degree \(R\) confines all possible output states to the common low-dimensional subspace in \eqref{eq:lower-bound-rank}, while \eqref{eq:lower-bound-Dyson} bounds the discarded tail solely in terms of the total evolution time. Comparing the dimension of this subspace with the size of the chosen family of Hamiltonians shows that a short protocol cannot identify the unknown Hamiltonian, giving us the total evolution time lower bound in \Cref{thm:lower-bound}.

The same gap-estimation primitive also gives two applications that do not require full reconstruction. For tolerant certification, set \(A=H-H_0\) and sample energy differences on random Pauli product states. Hypercontractive moment bounds imply that when \(\norm{A}_{\F}\geq\eps\), the empirical squared gaps are detectably large, whereas when \(\norm{A}_{\F}\leq\eps/12^k\), they remain below a separated threshold. The two cases are quantified in \Cref{lem:far-gaps,lem:close-gaps}; in particular, the global and single-qubit-control tests are governed by \eqref{eq:far-global} and \eqref{eq:far-single}, respectively. We estimate each unknown gap using the appropriate eigenphase engineering subroutine, subtract the exactly known contribution of \(H_0\), and compare the root-mean-square residual with a fixed threshold. This proves \Cref{thm:cert,thm:cert-single}. Unlike full learning, certification needs no sparsity assumption: it detects the aggregate Frobenius discrepancy without identifying the offending coefficients.

Finally, a single specified Pauli coefficient can be isolated without assuming that the Hamiltonian is local or sparse. On two copies of the Hilbert space, let \(\ket{\Omega}\) be maximally entangled and compare the symmetric and antisymmetric superpositions of \(\ket{\Omega}\) and \((P\otimes I)\ket{\Omega}\). Their energy difference under \(H\otimes I\) is exactly twice the desired coefficient \(2^{-n}\Tr(PH)\). Applying the general energy-gap subroutine therefore estimates this coefficient with total evolution time \(\widetilde{\mathcal O}(1/\eps)\) and step size \(\widetilde\Theta(1/\Lambda)\), proving \Cref{thm:single-coeff}. Applying the construction coefficient by coefficient also gives the agnostic local approximation result in \Cref{thm:agnostic-learning}.

\section{Preliminaries and notation}
\label{sec:preliminaries}

In this section, we introduce the notation and technical tools used throughout the paper. We first define the matrix and coefficient norms for Hamiltonians. We then recall a sparse-recovery result for bounded orthonormal systems, which will be used to recover sparse Pauli coefficients from random energy-gap measurements. Finally, we review the extrapolation and eigenvalue perturbation tools used in the construction of the energy-gap estimation subroutine in \Cref{sec:energy-gap}.

\subsection{Hamiltonian norms}\label{subsec:notation}

Let $d=2^n$ denote the Hilbert-space dimension of an $n$-qubit system.

\begin{definition}[Matrix norms]
For an operator $A$, we define the normalized Frobenius norm and the operator norm by
\begin{equation}
    \norm{A}_{\F}^2
    :=
    d^{-1}\Tr(A^\dagger A),
    \qquad
    \norm{A}_{\op}
    :=
    \sup_{\norm{v}=1}\norm{Av}.
\end{equation}
\end{definition}

\begin{definition}[Coefficient norms]
Let $A=\sum_P a_P P$ be the Pauli expansion of $A$. We define the coefficient $\ell_p$-norm by
\begin{equation}
    \norm{A}_{\ell_p}
    :=
    \norm{a}_p.
\end{equation}
\end{definition}

Since the Pauli operators are orthonormal with respect to the normalized Hilbert--Schmidt inner product, the normalized Frobenius norm coincides with the coefficient $\ell_2$-norm. We will repeatedly use
\begin{equation}
    \norm{A}_{\F}
    =
    \norm{A}_{\ell_2},
    \qquad
    \norm{A}_{\F}
    \leq
    \norm{A}_{\op},
    \qquad
    \norm{A}_{\op}
    \leq
    \norm{A}_{\ell_1}.
    \label{eq:norm-relations}
\end{equation}

\paragraph{Notation.}
Let $\mathcal P_n:=\{I,X,Y,Z\}^{\otimes n}$ denote the set of
$n$-qubit Pauli operators, and define
$\Pnk:=\{P\in\mathcal P_n:1\leq\wt(P)\leq k\}$.
For $P\in\Pnk$, we denote its support and weight by
$S_P:=\supp(P)$ and $w_P:=\wt(P)$, respectively. We use $e_j\in\{0,1\}^n$ for the string with a single $1$ at coordinate $j$, write $x\oplus y$ for bitwise addition modulo two, and define $x\cdot S:=\sum_{i\in S}x_i\pmod2$. The notation $\operatorname{Ham}(x,y)$ denotes Hamming distance. For $\sigma\in\{X,Y,Z\}^n$, we write $P\preceq\sigma$ if $P_i=\sigma_i$ for every $i\in S_P$. For $x\in\{0,1\}^n$, let $\ket{x_i}_{\sigma_i}$ be the eigenstate of $\sigma_i$ with eigenvalue $(-1)^{x_i}$, and define the product state
\begin{equation}
    \ket{x}_\sigma:=
    \bigotimes_{i=1}^n \ket{x_i}_{\sigma_i}.
\end{equation}

\subsection{Sparse recovery from random function values}

Our Hamiltonian-learning algorithms reduce coefficient recovery to a sparse-recovery problem. We therefore recall a standard result for bounded orthonormal systems.

Let $(\Omega,\mu)$ be a probability space, and let
$\{\phi_j\}_{j=1}^M\subset L_2(\Omega,\mu)$ satisfy
\begin{equation}
    \mathbb E_{\omega\sim\mu}
    \left[
        \overline{\phi_j(\omega)}\phi_\ell(\omega)
    \right]
    =
    \delta_{j,\ell},
    \qquad
    \norm{\phi_j}_{L_\infty(\mu)}
    \leq
    K.
    \label{eq:orthonormal}
\end{equation}
Draw $\omega_1,\ldots,\omega_m$ independently from $\mu$ and define the normalized sampling matrix $A\in\mathbb C^{m\times M}$ by
\begin{equation}
    A_{r,j}
    :=
    \frac{1}{\sqrt m}\phi_j(\omega_r).
    \label{eq:normal-matrix}
\end{equation}
Suppose the observed data satisfy
\begin{equation}
    y=Az+e,
    \qquad
    \norm{e}_2\leq\eta.
    \label{eq:error-bound}
\end{equation}
When $z$ is sparse, it can be recovered from a number of samples essentially linear in its sparsity. The sufficient sample bound below follows from \cite[Theorem~2.3]{BrugiapagliaDirksenJungRauhut2021}; see also \cite[Theorem~4]{RauhutWard2011}.
\begin{lemma}[Sparse recovery from bounded orthonormal systems]
\label{lem:RIP}
Assume that \eqref{eq:orthonormal}--\eqref{eq:error-bound} hold, that $1\leq s\leq M$, and that $z$ is $s$-sparse. For $0<\delta<1/3$, if
\begin{equation}\label{eq:ortho-sample}
    m
    \geq
    C K^2 s
    (1+\log(sK^2))^2
    \log\left(\frac{M}{\delta}\right)
\end{equation}
for a universal constant $C$, then every solution of
\begin{equation}
    \widehat z
    \in
    \operatorname*{arg\,min}_{v\in\mathbb C^M}
    \norm{v}_1
    \quad
    \text{subject to}
    \quad
    \norm{Av-y}_2\leq\eta
    \label{eq:BPDN}
\end{equation}
satisfies
\begin{equation}
    \norm{\widehat z-z}_2
    \leq
    C_{\rm cs}\eta
\end{equation}
with probability at least $1-\delta$, where $C_{\rm cs}$ is a universal constant.
\end{lemma}

\begin{proof}
The proof is given in \Cref{app:RIP}.
\end{proof}

\subsection{Extrapolation techniques}
\label{subsec:extrapolation}

The energy-gap estimation developed in \Cref{sec:energy-gap} first estimates a quantity at several nonzero evolution times and then extrapolates these values to $t=0$. We explain the extrapolation estimate used later.

Let $f$ be an even analytic function with expansion
\begin{equation}
    f(x)
    =
    f(0)
    +
    \sum_{j=1}^{\infty}a_jx^{2j},
    \qquad
    |a_j|
    \leq
    C_a\rho^{-2j}.
    \label{eq:even-expansion}
\end{equation}
Here $\rho$ denotes the radius on which the expansion is controlled.

Fix an integer $p\geq1$ and a step size $\tau>0$. We use the values
$f(\tau),f(2\tau),\ldots,f(p\tau)$ to estimate $f(0)$.
Suppose that $\widehat f_q$ estimates $f(q\tau)$ with error
$|\widehat f_q-f(q\tau)|\leq\eta_q$.
Define
\begin{equation}
    \widehat{\mathcal E}_p[f]
    :=
    \sum_{q=1}^p c_q\widehat f_q,
    \qquad
    c_q
    :=
    (-1)^{q-1}
    \frac{2(p!)^2}{(p-q)!(p+q)!}.
    \label{eq:extrapolated-value}
\end{equation}
The coefficients satisfy, for a universal constant $C$,
\begin{equation}
    A_p
    :=
    \sum_{q=1}^p|c_q|
    \leq
    C\sqrt p,
    \qquad
    S_p
    :=
    \sum_{q=1}^p\sqrt{|c_q|}
    \leq
    C\sqrt p.
    \label{eq:weight-bounds}
\end{equation}
Indeed, $|c_q|=2\prod_{j=0}^{q-1}(p-j)/(p+j+1)\leq2e^{-q^2/(2p)}$, and summing this bound and its square root gives \eqref{eq:weight-bounds}. The following lemma separates the extrapolation error from the error in estimating the individual values $f(q\tau)$.

\begin{lemma}[Noisy extrapolation]
\label{lem:noisy-extrapolation}
Suppose
$r:=p^2\tau^2/\rho^2<1$.
Then, for some constant $C_f=CC_a$,
\begin{equation}
    \left|
        \widehat{\mathcal E}_p[f]-f(0)
    \right|
    \leq
    \frac{C_f\sqrt p}{1-r}r^p
    +
    \sum_{q=1}^p|c_q|\eta_q.
    \label{eq:noisy-extrapolation}
\end{equation}
\end{lemma}

\begin{proof}
Let $\mathcal E_p[f]:=\sum_{q=1}^pc_qf(q\tau)$. The weights $c_q$
are the Lagrange interpolation weights at zero for the nodes
$1^2,\ldots,p^2$, so
\begin{equation}\label{eq:moment-cancellation}
    \sum_{q=1}^pc_q=1,\qquad
    \sum_{q=1}^pc_qq^{2j}=0\quad(1\leq j\leq p-1).
\end{equation}
Substituting \eqref{eq:even-expansion} cancels all nonconstant terms
of degree below $2p$. Using \eqref{eq:weight-bounds} and
$r=(p\tau/\rho)^2<1$, we obtain
\begin{equation}
    \begin{aligned}
    |\mathcal E_p[f]-f(0)|
    &\leq \sum_{j\geq p}|a_j|\sum_{q=1}^p|c_q|(q\tau)^{2j}\\
    &\leq C_aA_p\sum_{j\geq p}r^j
    \leq \frac{C_f\sqrt p}{1-r}r^p.
\end{aligned}
\end{equation}
Finally, $|\widehat{\mathcal E}_p[f]-\mathcal E_p[f]|
\leq\sum_{q=1}^p|c_q|\eta_q$, which proves the claim.
\end{proof}

\subsection{Eigenvalue perturbation theory}

The construction of the energy-gap estimation requires controlling how selected eigenvalues and eigenvectors of a known control unitary change under a small perturbation. We employ analytic eigenvalue perturbation theory. Standard results guarantee analytic continuations of the chosen eigenvalue and eigenvector, together with the corresponding first-order derivative formula; see \cite[Theorems~1--3]{GreenbaumLiOverton2020}. For completeness, we establish below the explicit norm bounds needed in our analysis.

\begin{lemma}[Analytic perturbation]
\label{lem:quantitative-analytic-perturbation}
Let $A(z)$ be a matrix-valued analytic function for $|z|<r$, and suppose that $A(0)$ is normal. Let $\nu$ be a simple eigenvalue of $A(0)$ with normalized eigenvector $\ket{u}$, and define its spectral separation by
\begin{equation}
    g
    :=
    \min_{\substack{
        \mu\in\operatorname{spec}(A(0))\\
        \mu\neq\nu
    }}
    |\nu-\mu|.
    \label{eq:abstract-spectral-gap}
\end{equation}
Assume that
\begin{equation}
    \norm{A(z)-A(0)}_{\op}
    \leq
    \frac{g}{12}
    \label{eq:abstract-perturbation-smallness}
\end{equation}
for every $|z|<r$.

Then there exist an analytic eigenvalue $\lambda(z)$ and
a corresponding analytic eigenvector $\ket{\phi(z)}$ satisfying
$\lambda(0)=\nu$, $\ket{\phi(0)}=\ket{u}$, and
$\braket{u}{\phi(z)}=1$. Moreover,
\begin{equation}
    \norm{\ket{\phi(z)}-\ket{u}}
    \leq
    \frac{C}{g}
    \norm{A(z)-A(0)}_{\op},
    \qquad
    |\lambda(z)-\nu|
    \leq
    C\norm{A(z)-A(0)}_{\op},
    \label{eq:abstract-eigenpair-bounds}
\end{equation}
for a universal constant $C$. Also, the first derivative of the eigenvalue is
\begin{equation}
    \lambda'(0)
    =
    \bra{u}A'(0)\ket{u}.
    \label{eq:abstract-eigenvalue-derivative}
\end{equation}
\end{lemma}
\begin{proof}
Set $A_0:=A(0)$, $E(z):=A(z)-A_0$, and
$\Gamma:=\{\zeta\in\mathbb C:|\zeta-\nu|=g/2\}$, traversed counterclockwise.
Since $A_0$ is normal, \eqref{eq:abstract-spectral-gap} gives
$\norm{(\zeta I-A_0)^{-1}}_{\op}\leq2/g$ for $\zeta\in\Gamma$.
The factorization $\zeta I-A(z)=(\zeta I-A_0)[I-(\zeta I-A_0)^{-1}E(z)]$
and $\norm{(\zeta I-A_0)^{-1}E(z)}_{\op}\leq1/6$ imply
\begin{equation}
    \norm{(\zeta I-A(z))^{-1}}_{\op}
    \leq\frac{1}{1-1/6}\frac2g=\frac{12}{5g}.
\end{equation}
We may therefore define
\begin{equation}
    P(z):=\frac{1}{2\pi i}\oint_\Gamma(\zeta I-A(z))^{-1}\,d\zeta.
\end{equation}
This spectral projector is analytic throughout $|z|<r$ and has constant rank one, since its rank is continuous and $P(0)=\ketbra{u}{u}$. The enclosed eigenvalue $\lambda(z)$ is therefore simple and analytic; see \cite[Section~3.3]{GreenbaumLiOverton2020}.

Using $(\zeta I-A(z))^{-1}-(\zeta I-A_0)^{-1}=(\zeta I-A(z))^{-1}E(z)(\zeta I-A_0)^{-1}$, we obtain
\begin{equation}
\begin{aligned}
    \norm{P(z)-P(0)}_{\op}
    &\leq\frac1{2\pi}\oint_\Gamma
    \norm{(\zeta I-A(z))^{-1}}_{\op}\norm{E(z)}_{\op}
    \norm{(\zeta I-A_0)^{-1}}_{\op}|d\zeta|\\
    &\leq\frac{\pi g}{2\pi}\frac{12}{5g}\norm{E(z)}_{\op}\frac2g
    =\frac{12}{5g}\norm{E(z)}_{\op}\leq\frac15,
\end{aligned}
\end{equation}
where $\Gamma$ has length $\pi g$.
Since $P(0)=\ketbra{u}{u}$, we have $|\bra uP(z)\ket u|\geq4/5$.
Thus $\ket{\phi(z)}:=P(z)\ket u/\bra uP(z)\ket u$ is an analytic eigenvector with the required normalization, and
\begin{equation}
    \norm{\ket{\phi(z)}-\ket u}
    =\frac{\norm{(I-P(0))(P(z)-P(0))\ket u}}{|\bra uP(z)\ket u|}
    \leq\frac54\norm{P(z)-P(0)}_{\op}
    \leq\frac3g\norm{E(z)}_{\op}.
\end{equation}

It remains to bound $\lambda(z)$. If
$\operatorname{dist}(\xi,\operatorname{spec}(A_0))>\norm{E(z)}_{\op}$,
then $\norm{(\xi I-A_0)^{-1}E(z)}_{\op}<1$ by normality.
Hence $\xi I-A(z)=(\xi I-A_0)[I-(\xi I-A_0)^{-1}E(z)]$ is invertible.
Thus some $\mu\in\operatorname{spec}(A_0)$ satisfies
$|\lambda(z)-\mu|\leq\norm{E(z)}_{\op}$.
Since $|\lambda(z)-\nu|<g/2$, every $\mu\neq\nu$ instead satisfies $|\lambda(z)-\mu|>g/2>\norm{E(z)}_{\op}$.
Therefore $\mu=\nu$ and $|\lambda(z)-\nu|\leq\norm{E(z)}_{\op}$.
Finally, differentiating the eigenvalue equation at zero and using
$\bra uA_0=\nu\bra u$ gives \eqref{eq:abstract-eigenvalue-derivative}.
\end{proof}

The lemma says that an isolated eigenvalue and its eigenvector remain close to their unperturbed values under a sufficiently small perturbation. In particular, the first derivative of the eigenvalue is determined by the expectation value of the perturbation in the original eigenvector. This relation will allow us to encode the desired Hamiltonian energy gap into the phase of an eigenvalue of the unitary that we design in the next section.

\section{Eigenphase engineering for energy-gap estimation}\label{sec:energy-gap}
We estimate Hamiltonian energy gaps using a technique that we call \emph{eigenphase engineering}, which we introduce in this section. We first describe its setting and objective. We assume access to $e^{-iHt}$ for chosen positive times $t$, together with known gates $W,W^\dagger$ and standard computational-basis gates and measurements. Let $\ket{e_0}=\ket{0^n}$, $\ket{e_1}:=\ket{10^{n-1}}$ and set $\ket{x}=W\ket{e_0}, \ket{y}=W\ket{e_1}$ and $K:=W^\dagger H W$. The goal of eigenphase engineering is to estimate the energy gap of $H$ between $\ket{x}$ and $\ket{y}$,
\begin{equation}\label{eq:Delta}
\Delta_H(W)=\bra{x}H\ket{x}-\bra{y}H\ket{y}
=\bra{e_0}K\ket{e_0}-\bra{e_1}K\ket{e_1}.
\end{equation}

We construct a subroutine for estimating $\Delta_H(W)$ using the known gates $W,W^\dagger$ and forward evolution under $H$, while maintaining a near-optimal step size. In \Cref{subsec:energy-general}, we allow $W$ to be an arbitrary unitary. In \Cref{subsec:energy-single-qubit}, we show that when $H$ is local and $W$ is a product unitary, the same task can be performed using only single-qubit controls and measurements. Both constructions achieve Heisenberg-limited total evolution time and a near-optimal step size.

\subsection{Estimating the energy gap of a Hamiltonian}\label{subsec:energy-general}
In this section, we develop eigenphase engineering for an arbitrary unitary $W$. The subroutine estimates the energy gap using the circuit in \Cref{fig:energy-gap-subroutine}. We then prove the correctness of the construction and analyze its total evolution time and step size, as summarized in \Cref{thm:energy-gap-subroutine}.
\begin{theorem}[Estimating the energy gap]\label{thm:energy-gap-subroutine}
Suppose $\ket{x}=W\ket{e_0}$ and $\ket{y}=W\ket{e_1}$, and assume $\|H\|_{\op}\leq\Lambda$.  For every $\eps\in(0,\Lambda)$ and $\delta\in(0,1/3)$, one can estimate $\bra{x}H\ket{x}-\bra{y}H\ket{y}$ to additive error $\eps$ with probability at least $1-\delta$, using step size
\begin{equation}\label{eq:energy-gap-step-size}
\tau=\Theta\left(\frac{1}{\Lambda(1+\log(\Lambda/\eps))}\right)
=\widetilde\Theta(1/\Lambda)
\end{equation}
and total evolution time $T_{\rm tot}=\widetilde O(\frac{1}{\eps}\log\frac{1}{\delta})$.
\end{theorem}

To prove \Cref{thm:energy-gap-subroutine}, we first analyze the output of the circuit in \Cref{fig:energy-gap-subroutine}. The circuit is built from the known unitaries $W,W^\dagger$, the forward evolution $e^{-itH}$, and a control unitary $V$. We choose
\begin{equation}\label{eq:V}
    V
    :=
    \ketbra{e_0}{e_0}
    +
    e^{-i\gamma}\ketbra{e_1}{e_1}
    -
    \left(
        I-\ketbra{e_0}{e_0}-\ketbra{e_1}{e_1}
    \right),
    \qquad
    \gamma:=\pi/2.
\end{equation}
Thus, $V$ has eigenvalues $1$ and $-i$ on $\ket{e_0}$ and $\ket{e_1}$, respectively, and eigenvalue $-1$ on their orthogonal complement. In particular, each of the two distinguished eigenvalues is separated from the rest of the spectrum by a constant gap of $\sqrt2$, which allows us to apply eigenvalue perturbation theory for sufficiently small $t$. We then define
\begin{equation}\label{eq:Fpm}
    F_+(t)
    :=
    Ve^{-itK}
    =
    VW^\dagger e^{-itH}W,
    \qquad
    F_-(t)
    :=
    V^\dagger e^{-itK}
    =
    V^\dagger W^\dagger e^{-itH}W.
\end{equation}
The resulting circuit is shown in \Cref{fig:energy-gap-subroutine}.

\begin{figure}[ht]
\centering
\begin{quantikz}[row sep=0.45cm,column sep=0.45cm]
\lstick{$\ket{+}_1$} & \gate[wires=2]{F_+(t)^m} & \gate[wires=2]{F_-(t)^m} & \meter{} & \rstick{$X_1$ or $Y_1$}\\
\lstick{$\ket{0^{n-1}}$} & \ghost{F_+(t)^m} & \ghost{F_-(t)^m} & \qw & \qw
\end{quantikz}

\vspace{0.35cm}

\begin{quantikz}[row sep=0.45cm,column sep=0.4cm]
\lstick{} & \gate[wires=2]{W} & \gate[wires=2]{e^{-itH}} & \gate[wires=2]{W^\dagger} & \gate[wires=2]{V} & \rstick{$F_+(t)$}\\
\lstick{} & \ghost{W} & \ghost{e^{-itH}} & \ghost{W^\dagger} & \ghost{V} & \qw
\end{quantikz}
\qquad
\begin{quantikz}[row sep=0.45cm,column sep=0.4cm]
\lstick{} & \gate[wires=2]{W} & \gate[wires=2]{e^{-itH}} & \gate[wires=2]{W^\dagger} & \gate[wires=2]{V^\dagger} & \rstick{$F_-(t)$}\\
\lstick{} & \ghost{W} & \ghost{e^{-itH}} & \ghost{W^\dagger} & \ghost{V^\dagger} & \qw
\end{quantikz}
\caption{The quantum circuit of the energy gap estimation subroutine.}
\label{fig:energy-gap-subroutine}
\end{figure}
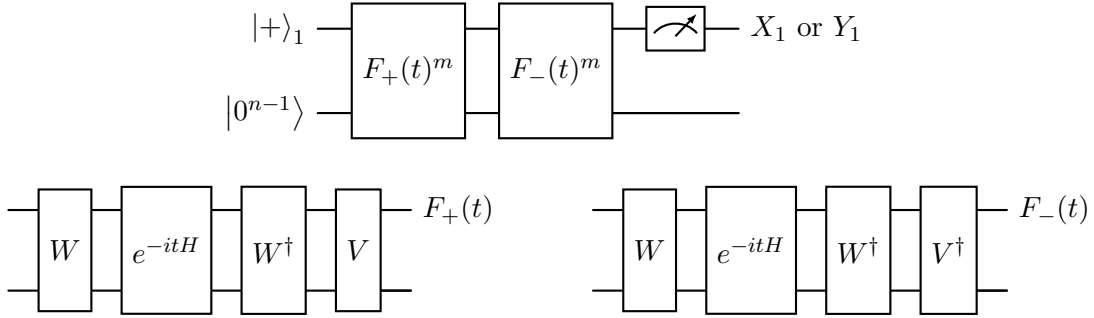

To analyze the eigenphases of $F_\pm(t)$ for small $t$, we extend these matrices to a complex parameter $z$ by defining $F_+(z):=Ve^{-izK}$ and $F_-(z):=V^\dagger e^{-izK}$. This allows us to apply the analytic perturbation result in \Cref{lem:quantitative-analytic-perturbation} to the isolated eigenvalues associated with $\ket{e_0}$ and $\ket{e_1}$. The following lemma characterizes the resulting eigenphases and shows that the linear term of their Taylor expansions is given by the corresponding energy expectations $\bra{x}H\ket{x}$ and $\bra{y}H\ket{y}$, respectively.

\begin{lemma}\label{lem:analytic}
The matrix $F_+(z)$ has an analytic eigenvector $\ket{\phi_{a,+}(z)}$ and corresponding analytic eigenvalue $\lambda_{a,+}(z)=e^{-i\theta_{a,+}(z)}$ for $a=0,1$. For $|z|<c/\Lambda$, the eigenvectors and eigenvalues satisfy
\begin{equation}\label{eq:eigen-perturbation}
    \norm{\ket{\phi_{a,+}(z)}-\ket{e_a}} \le C\Lambda |z|,\quad |\lambda_{a,+}(z)-\lambda_{a,+}(0)| \le C\Lambda |z|,
\end{equation}
where $c,C$ are universal constants. The same argument holds for $F_{-}(z)$, with eigenvector $\ket{\phi_{a,-}(z)}$ corresponding to eigenvalue $\lambda_{a,-}(z)=e^{-i\theta_{a,-}(z)}$.
Choose the phases analytically with $\theta_{a,\pm}(0)=\pm a\pi/2$.

With the value at $t=0$ defined by continuity, we have
\begin{equation}\label{eq:def-delta-a}
     \delta_a(t):=\frac{\theta_{a,+}(t)+\theta_{a,-}(t)}{2t} = \bra{e_a}K\ket{e_a}+\sum_{j=1}^\infty b_{a,j}t^{2j},\quad |b_{a,j}| \le C\Lambda(C\Lambda)^{2j}.
\end{equation}
\end{lemma}
\begin{proof}
    The proof is given in \Cref{app:analytic}.
\end{proof}

\Cref{lem:analytic} identifies the eigenphases whose first-order behavior contains the desired energy expectations. We next show that these eigenphases can be accessed directly by the circuit in \Cref{fig:energy-gap-subroutine}. Since the corresponding eigenvectors remain $\mathcal O(\Lambda t)$-close to $\ket{e_0}$ and $\ket{e_1}$, repeated applications of $F_+(t)$ and $F_-(t)$ approximately preserve these two basis states while accumulating the phases $\theta_{a,+}(t)$ and $\theta_{a,-}(t)$. The following lemma makes this phase accumulation explicit and bounds the resulting error.

\begin{lemma}\label{lem:estimate-f}
    For $0<t\le c/\Lambda$ and every integer $m\geq0$, we have
    \begin{equation}\label{eq:circuit-approx}
        \norm{F_{-}(t)^mF_{+}(t)^m\ket{+}\ket{0^{n-1}} - \ket{\psi}} \le \mathcal{O}(\Lambda t),
    \end{equation}
    where $\ket{\psi}$ is
    \begin{equation}
        \ket{\psi}=\frac{1}{\sqrt2}(e^{-im(\theta_{0,+}(t)+\theta_{0,-}(t))}\ket{e_0}+e^{-im(\theta_{1,+}(t)+\theta_{1,-}(t))}\ket{e_1}).
    \end{equation}
\end{lemma}
\begin{proof}
     We proved $\norm{\ket{\phi_{a,+}(t)}-\ket{e_a}}\le C\Lambda t$ in \eqref{eq:eigen-perturbation}. Using that $\ket{\phi_{a,+}(t)}$ is an eigenvector of $F_+(t)$, we have
    \begin{align}
        F_-(t)^mF_+(t)^m\ket{e_a}&=F_-(t)^mF_+(t)^m\ket{\phi_{a,+}(t)}+\mathcal{O}( \Lambda t) \nonumber \\
        &= e^{-im\theta_{a,+}(t)}F_-(t)^m\ket{\phi_{a,+}(t)}+\mathcal{O}( \Lambda t) \nonumber \\
        &=e^{-im\theta_{a,+}(t)}F_-(t)^m\ket{e_a}+\mathcal{O}(\Lambda t).
    \end{align}
    Applying the same logic to $\ket{\phi_{a,-}(t)}$ and $F_-(t)$, we have
    \begin{align}
        F_-(t)^mF_+(t)^m\ket{e_a}&=e^{-im\theta_{a,+}(t)}F_-(t)^m\ket{\phi_{a,-}(t)}+\mathcal{O}(\Lambda t) \nonumber \\
        &= e^{-im(\theta_{a,+}(t)+\theta_{a,-}(t))}\ket{\phi_{a,-}(t)}+\mathcal{O}(\Lambda t) \nonumber \\
        &=e^{-im(\theta_{a,+}(t)+\theta_{a,-}(t))}\ket{e_a}+\mathcal{O}(\Lambda t).
    \end{align}
    The claim follows by applying these identities to $\ket{+}\ket{0^{n-1}}=(\ket{e_0}+\ket{e_1})/\sqrt2$.
\end{proof}

Hence, the circuit accumulates the desired phase difference with only $\mathcal O(\Lambda t)$ error. Since this phase difference is $2mt(\delta_0(t)-\delta_1(t))$, we can estimate $\delta_0(t)-\delta_1(t)$ by applying robust phase estimation to measurements of the first qubit.

\begin{corollary}\label{col:estimate-f}
    Define
    \begin{equation}\label{eq:def-f}
        f(t):=\delta_0(t)-\delta_1(t).
    \end{equation}
    For $\eta\in(0,\Lambda)$, $\delta\in(0,1/3)$, and $0 < t\leq c/\Lambda$, one can estimate $f(t)$ within
    $\eta$-additive error with probability at least $1-\delta$ using
    $\mathcal{O}(\eta^{-1}\log(1/\delta))$ total evolution time.
\end{corollary}

\begin{proof}
Using \eqref{eq:def-f} and \eqref{eq:def-delta-a}, up to a global
phase, the ideal state in \Cref{lem:estimate-f} is
\begin{equation}
    \frac{1}{\sqrt2}
    \left(
        \ket{0}
        +
        e^{2imtf(t)}\ket{1}
    \right)
    \ket{0^{n-1}}.
\end{equation}
By \Cref{lem:analytic}, there is a universal constant $C_f$ such that $|f(t)|\leq C_f\Lambda$ for $t\leq c/\Lambda$, and by
\Cref{lem:estimate-f}, the actual output state differs from the ideal
state by at most $C_s\Lambda t$ for some universal constant $C_s$. Choose
\begin{equation}
    c\leq
    \min\left\{
        \frac{\pi}{8C_f},
        \frac{1}{64C_s}
    \right\}.
\end{equation}
Then $|2tf(t)|\leq\pi/4$, so the initial phase is determined without
a $2\pi$ ambiguity, while the circuit error changes each $X$- or
$Y$-measurement probability by at most $1/64$.

We now apply robust phase estimation~\cite{KimmelLowYoder2015} to the first qubit using repetition numbers $m_j=2^j$, $j=0,\ldots,J$.
At each stage, estimate both measurement probabilities to error $1/64$. Including the circuit error, the estimated sine and cosine have Euclidean error at most $\sqrt2/16$, giving phase error below $\pi/8$ modulo $2\pi$. The initial bound $|2tf(t)|\leq\pi/4$ fixes the first phase. At later stages, choose the $2\pi$-modulo nearest twice the preceding estimate; the correct modulo is within $3\pi/8<\pi$ of that value.
Hence, after the final stage, the estimate of $2tf(t)$ has error $\mathcal O(2^{-J})$.
Taking $J:=\lceil\log_2(1/(t\eta))\rceil$ gives error at most $\pi/(16t2^J)\leq\eta$ in $f(t)$.

We now calculate the required total evolution time. Assign failure probability
$\delta_j:=\delta\,2^{j-J-1}$ to stage $j$. Then
$\sum_{j=0}^J\delta_j<\delta$, and stage $j$ requires $\mathcal O(J-j+\log(1/\delta))$ repetitions. Therefore, the total
evolution time is
\begin{equation}
    T_{\rm tot}=
    \mathcal O\left(
        t\sum_{j=0}^J
        2^j\left(J-j+\log\frac1\delta\right)
    \right) =
    \mathcal O\left(
        t2^J\log\frac1\delta
    \right) =
    \mathcal O\left(
        \frac1\eta\log\frac1\delta
    \right),
\end{equation}
which concludes the proof.
\end{proof}

With \Cref{col:estimate-f}, we can efficiently estimate $f(t)$ at any sufficiently small nonzero time $t$. On the other hand, \Cref{lem:analytic} shows that $f(t)$ is an even analytic function whose value at $t=0$ is exactly the desired energy gap. Therefore, it remains to recover $f(0)$ from estimates of $f(t)$ at nonzero times. We accomplish this using the extrapolation procedure in \Cref{lem:noisy-extrapolation}, choosing all sampling times on a common time grid with near-optimal step size.

Now we prove \Cref{thm:energy-gap-subroutine}. Check that $f(0)$ is the desired energy gap $\bra{e_0}K\ket{e_0}-\bra{e_1}K\ket{e_1}$. By \eqref{eq:def-delta-a} in \Cref{lem:analytic} and \eqref{eq:def-f}, we have
\begin{equation}
    f(t)=
    f(0)+\sum_{j=1}^{\infty}a_jt^{2j},
    \qquad
    |a_j|
    \leq
    C_a\rho^{-2j},
    \label{eq:energy-gap-expansion}
\end{equation}
where $C_a\leq C\Lambda$. Note that $f$ is even and analytic for $|t|<\rho=c/\Lambda$. Hence, we estimate $f(q\tau)$ for various $q\in\mathbb{N}$ and employ the extrapolation in \Cref{lem:noisy-extrapolation} to estimate $f(0)$. 

We extrapolate with the nodes
\begin{equation}
    \tau:=\frac{\rho}{2p},
    \qquad
    t_q:=q\tau,
    \qquad
    1\leq q\leq p.
    \label{eq:energy-subroutine-nodes}
\end{equation}
Every $t_q$ lies in $(0,\rho)$ and $r:=\frac{p^2\tau^2}{\rho^2}=\frac14$. Hence, it satisfies the assumption of \Cref{lem:noisy-extrapolation} and therefore we can employ extrapolation with $f(t_q)$. Let $c_1,\ldots,c_p$ be the extrapolation weights in
\eqref{eq:extrapolated-value}, and recall from \eqref{eq:weight-bounds} that
\begin{equation}
    S_p
    :=
    \sum_{q=1}^p\sqrt{|c_q|}
    \leq
    C\sqrt p.
    \label{eq:energy-subroutine-Sp}
\end{equation}
For every $q$, choose the desired accuracy
\begin{equation}\label{eq:energy-subroutine-node-accuracy}
    \eta_q:=
    \min\left\{
        \frac{\eps}{4S_p\sqrt{|c_q|}},
        \frac{\Lambda}{2}
    \right\}.
\end{equation}
By \Cref{col:estimate-f}, estimating $f(t_q)$ to this accuracy with a failure probability of $\delta/p$ requires evolution time $\mathcal O\left(\frac1{\eta_q}\log\frac p\delta\right)$. Let $\widehat f_q$ be these estimates. Their extrapolated value and statistical error satisfy
\begin{align}
    \widehat{\mathcal E}_p[f]&:=\sum_{q=1}^pc_q\widehat f_q,
        \label{eq:energy-subroutine-estimator}\\
    \sum_{q=1}^p|c_q|\,|\widehat f_q-f(t_q)|
    &\leq\sum_{q=1}^p|c_q|\eta_q
    \leq\frac{\eps}{4S_p}\sum_{q=1}^p\sqrt{|c_q|}
    =\frac\eps4.
        \label{eq:energy-subroutine-stat-error}
\end{align}
The above succeeds with probability at least $1-\delta$ by a union bound. Consequently,
\Cref{lem:noisy-extrapolation} gives
\begin{equation}\label{eq:energy-subroutine-total-error}
    |\widehat{\mathcal E}_p[f]-f(0)|
    \leq C\Lambda\sqrt p\,4^{-p}+\frac\eps4.
\end{equation}
Choose 
\begin{equation}\label{eq:chosen-p}
    p:=\left\lceil C_p\left(1+\log\frac\Lambda\eps\right)\right\rceil
\end{equation} 
with $C_p$ sufficiently large that
$C\Lambda\sqrt p\,4^{-p}\leq\eps/2$. This proves the required error bound. For the required total evolution time, \eqref{eq:energy-subroutine-node-accuracy}, $S_p^2\leq Cp$, and $\eps<\Lambda$ imply
\begin{equation}\label{eq:energy-subroutine-summed-cost}
    \sum_{q=1}^p\frac1{\eta_q}
    \leq\frac{2p}{\Lambda}+\frac{4S_p^2}{\eps}
    \leq\frac{(4C+2)p}{\eps}.
\end{equation}
Therefore the total evolution time is $ T_{\rm tot}=\mathcal O\left(\frac p\eps\log\frac p\delta\right)=\widetilde{\mathcal O}\left(\frac1\eps\log\frac1\delta\right)$. Finally, every evolution time is an integer multiple of $\tau=\rho/(2p)=\Theta(1/[\Lambda(1+\log(\Lambda/\eps))]) =\widetilde\Theta(1/\Lambda)$, which proves the step size.

\subsection{Estimating the energy gap of a Hamiltonian with single-qubit gates}\label{subsec:energy-single-qubit}

We now construct the energy-gap subroutine using only single-qubit controls and measurements, assuming that $H$ is $k$-local and $W$ is a product of single-qubit unitaries. We therefore replace $V$ in \eqref{eq:Fpm} by a product of single-qubit gates.

Define the product control
\begin{equation}\label{eq:single-product-control}
    V:=\exp\left[-i\gamma\left(N_1+2\sum_{\ell=2}^nN_\ell\right)\right]
    =e^{-i\gamma N_1}\bigotimes_{\ell=2}^ne^{-i2\gamma N_\ell},
\end{equation}
where $ N_\ell:=\ketbra{1}{1}_\ell$ and $\gamma:=2\pi(\sqrt2-1)$. Both $V$ and $V^\dagger$ are implemented by depth-one single-qubit circuits. We use this control in the circuit of \Cref{fig:energy-gap-subroutine}. The following theorem establishes its accuracy and required total evolution time and step size.

\begin{theorem}[Estimating the energy gap with single-qubit gates]\label{thm:energy-single-qubit}
Suppose $W$ is a product of single-qubit unitaries and $\ket{x}=W\ket{e_0}$, $\ket{y}=W\ket{e_1}$. Assume that $H$ is $k$-local and $\|H\|_{\op}\leq\Lambda$.  For every $\eps\in(0,\Lambda)$ and $\delta\in(0,1/3)$, one can estimate $\bra{x}H\ket{x}-\bra{y}H\ket{y}$ to additive error $\eps$ with probability at least $1-\delta$, using step size
\begin{equation}
\tau=\Theta\left(\frac{1}{k\Lambda(1+\log(k\Lambda/\eps))^2}\right),
\end{equation}
total evolution time $T_{\rm tot}=\widetilde O(\eps^{-1}\log(1/\delta))$, and only single-qubit gates and measurements. For fixed $k$, the step size is $\widetilde\Theta(1/\Lambda)$.
\end{theorem}

To prove \Cref{thm:energy-single-qubit}, we analyze the spectrum of the product control $V$ and identify the eigenvalues associated with the two states $\ket{e_0}$ and $\ket{e_1}$. Since $V$ is diagonal in the computational basis, it is convenient to label its eigenvectors by bit strings $u\in\{0,1\}^n$. For $u\in\{0,1\}^n$, write
\begin{equation}\label{eq:def-u}
    \ket{u}:=\ket{u_1\cdots u_n},
    \qquad
    E(u):=u_1+2\sum_{\ell=2}^nu_\ell.
\end{equation}
Then, $V\ket{u}=\nu_u\ket{u}, \nu_u:=e^{-i\gamma E(u)}$, so the eigenvector $\ket{u}$ has the corresponding eigenvalue $\nu_u$. 

Unlike the control used in \Cref{subsec:energy-general}, these eigenvalues are not separated by a constant gap. In particular, the eigenvalues corresponding to $\ket{e_0}$ and $\ket{e_1}$ can be arbitrarily close to other eigenvalues of $V$. To apply \Cref{lem:quantitative-analytic-perturbation}, we therefore introduce an auxiliary unitary used only in the proof. For $a\in\{0,1\}$ and an integer $L\geq1$ to be specified later, define
\begin{equation}\label{eq:single-auxiliary-V}
    \widehat V_{a,+}\ket{u}
    :=
    \begin{cases}
        V\ket{u},
        &
        \operatorname{Ham}(u,e_a)\leq kL,
        \\[1mm]
        -\nu_{e_a}\ket{u},
        &
        \operatorname{Ham}(u,e_a)>kL,
    \end{cases}
\end{equation}
and set $\widehat V_{a,-}:=\widehat V_{a,+}^\dagger$. The auxiliary unitary agrees with $V$ on the computational-basis states within Hamming distance $kL$ of $e_a$, while separating the selected eigenvalue from the rest of the spectrum. Then, the eigenvalue $\nu_{e_a}$ is separated from other eigenvalues of $\widehat V_{a,+}$ as shown in the following lemma.

\begin{lemma}[Local spectral gap]
\label{lem:single-local-gap}
For each \(a\in\{0,1\}\), the eigenvalue \(\nu_{e_a}\) of
\(\widehat V_{a,+}\) is nondegenerate and satisfies
\begin{equation}
    \min_{\substack{
        \nu\in\operatorname{spec}(\widehat V_{a,+})\\
        \nu\neq\nu_{e_a}
    }}
    |\nu-\nu_{e_a}|
    \geq
    \frac{1}{2kL}.
    \label{eq:single-local-gap}
\end{equation}
Similarly, the eigenvalue \(\nu_{e_a}^{-1}\) of
\(\widehat V_{a,-}\) is nondegenerate and satisfies
\begin{equation}
    \min_{\substack{
        \nu\in\operatorname{spec}(\widehat V_{a,-})\\
        \nu\neq\nu_{e_a}^{-1}
    }}
    |\nu-\nu_{e_a}^{-1}|
    \geq
    \frac{1}{2kL}.
\end{equation}
\end{lemma}

\begin{proof}
Set $\alpha:=\sqrt2-1$, so that $\gamma=2\pi\alpha$.
Every eigenvalue of $\widehat V_{a,+}$ other than $\nu_{e_a}$
is either $-\nu_{e_a}$ or $\nu_u$ for some $u\neq e_a$ with
$\operatorname{Ham}(u,e_a)\leq kL$. The first case satisfies
\begin{equation}\label{eq:auxiliary-outside-gap}
    |-\nu_{e_a}-\nu_{e_a}|=2\geq\frac1{2kL}.
\end{equation}
In the second case, define
\begin{equation}\label{eq:single-integer-r}
    r:=E(u)-E(e_a).
\end{equation}
The equality $E(u)=0$ holds only for $u=e_0$. Likewise, $E(u)=1$
requires $u_1=1$ and all other bits to be zero, so it holds only
for $u=e_1$. Thus $r\neq0$. Each changed bit alters $E$ by at
most two, giving
\begin{equation}\label{eq:r-upper-bound}
    0<|r|\leq2\operatorname{Ham}(u,e_a)\leq2kL.
\end{equation}
We use the Diophantine bound
\begin{equation}\label{eq:single-diophantine-bound}
    \operatorname{dist}(r\alpha,\mathbb Z)\geq\frac1{4|r|}.
\end{equation}
Indeed, if $m$ is the nearest integer to $\sqrt2|r|$, then
$|2r^2-m^2|\geq1$ and $m+\sqrt2|r|<4|r|$.
Dividing these inequalities gives
$|\sqrt2|r|-m|\geq1/(4|r|)$, as claimed.
For $y:=\operatorname{dist}(r\alpha,\mathbb Z)\in[0,1/2]$,
$\sin(\pi y)\geq2y$. Hence
\begin{equation}\label{eq:inside-ball-spectral-gap}
    |\nu_u-\nu_{e_a}|=|1-e^{-i2\pi\alpha r}|
    =2\sin(\pi y)\geq4y\geq\frac1{|r|}\geq\frac1{2kL}.
\end{equation}
Together with \eqref{eq:auxiliary-outside-gap}, this proves
\eqref{eq:single-local-gap} and nondegeneracy.
Finally, $\widehat V_{a,-}=\widehat V_{a,+}^\dagger$, and
$|\nu^{-1}-\nu_{e_a}^{-1}|=|\nu-\nu_{e_a}|$ for unit-modulus
eigenvalues. The same conclusion therefore holds for
$\widehat V_{a,-}$.
\end{proof}

The spectral separation in \Cref{lem:single-local-gap} allows us to apply \Cref{lem:quantitative-analytic-perturbation} to the analytic matrices
\begin{equation}
    \widehat{\mathcal F}_{a,\pm}(z)
    :=
    \widehat V_{a,\pm}e^{-izK},
    \qquad
    z\in\mathbb C.
\end{equation}
The proof follows the argument of \Cref{lem:analytic}, with spectral separation of order $1/(kL)$.

\begin{lemma}\label{lem:analytic-single}
Let $\Delta_L=kL\Lambda$. For $a\in\{0,1\}$, the matrix $\widehat{\mathcal F}_{a,+}(z)$ has an analytic eigenvector $\ket{\widehat\phi_{a,+}(z)}$ and corresponding analytic eigenvalue $\widehat\lambda_{a,+}(z)=e^{-i\widehat\theta_{a,+}(z)}$. For $|z|<c/\Delta_L$, the eigenvectors and eigenvalues satisfy
\begin{equation}\label{eq:perturbation}
    \norm{\ket{\widehat\phi_{a,+}(z)}-\ket{e_a}} \le C\Delta_L |z|,\quad |\widehat\lambda_{a,+}(z)-\widehat\lambda_{a,+}(0)| \le C\Lambda |z|,
\end{equation}
where $c,C$ are universal constants. The same statement holds for $\widehat{\mathcal F}_{a,-}(z)$, with eigenvector $\ket{\widehat\phi_{a,-}(z)}$ corresponding to eigenvalue $\widehat\lambda_{a,-}(z)=e^{-i\widehat\theta_{a,-}(z)}$. Choose the phases analytically with $\widehat\theta_{a,\pm}(0)=\pm\gamma E(e_a)$.

With the value at $t=0$ defined by continuity, we have
\begin{equation}\label{eq:delta-a}
     \widehat\delta_a(t):=\frac{\widehat\theta_{a,+}(t)+\widehat\theta_{a,-}(t)}{2t} = \bra{e_a}K\ket{e_a}+\sum_{j=1}^\infty b_{a,j}t^{2j},\quad |b_{a,j}| \le C\Lambda(C\Delta_L)^{2j}.
\end{equation}
\end{lemma}

\begin{proof}
The perturbation $\widehat{\mathcal F}_{a,\pm}(z)-\widehat V_{a,\pm}$ has norm at most $2\Lambda|z|$, while \Cref{lem:single-local-gap} gives a separation of at least $1/(2kL)$. Applying \Cref{lem:quantitative-analytic-perturbation} by setting $A(z)=\widehat{\mathcal F}_{a,+}(z), A(0)=\widehat V_{a,+}$ proves \eqref{eq:perturbation}. The eigenvalue bound gives $|\widehat\theta_{a,+}(z)-\widehat\theta_{a,+}(0)|\leq C\Lambda|z|$. Using $\widehat\theta_{a,-}(z)=-\widehat\theta_{a,+}(-z)$ with the derivative formula and the Cauchy estimate on a circle of radius $\Theta(1/\Delta_L)$ then gives the even expansion and its coefficient bound, as in the proof of \Cref{lem:analytic}.
\end{proof}

\Cref{lem:analytic-single} provides the analogue of \Cref{lem:analytic} for the auxiliary matrices $\widehat{\mathcal F}_{a,\pm}(z)$, with the smaller range $|z|<c/\Delta_L$. These auxiliary matrices are used only in the analysis; the actual circuit still uses the operators $F_\pm(t)$. The key point is that the auxiliary and physical controls act identically on all computational-basis states that appear in the first $L$ orders of the expansion, because of the definition \eqref{eq:single-auxiliary-V}. As a result, repeated applications of the physical circuit accumulate the auxiliary eigenphases, with the error quantified below. The following lemma makes this phase accumulation and the resulting error precise.

\begin{lemma}\label{lem:estimate-fL}
    Let $\Delta_L=kL\Lambda$. There exist universal constants $c,C$ such that, for $0<t\le c/\Delta_L$ and every integer $m\geq0$, we have
    \begin{equation}\label{eq:circuit-approx-L}
        \norm{F_{-}(t)^mF_{+}(t)^m\ket{+}\ket{0^{n-1}} - \ket{\psi_L}} \le Ct\Delta_L + Cm(Ct\Delta_L)^{L+1},
    \end{equation}
    where $\ket{\psi_L}$ is
    \begin{equation}
        \ket{\psi_L}=\frac{1}{\sqrt2}(e^{-im(\widehat\theta_{0,+}(t)+\widehat\theta_{0,-}(t))}\ket{e_0}+e^{-im(\widehat\theta_{1,+}(t)+\widehat\theta_{1,-}(t))}\ket{e_1}).
    \end{equation}
\end{lemma}
\begin{proof}
    The proof is given in \Cref{app:estimate-fL}.
\end{proof}

\Cref{lem:estimate-fL} shows that the physical circuit accumulates the phase difference $2mtf_L(t)$, while its output differs from the corresponding ideal phase-evolution state $\ket{\psi_L}$ by at most $Ct\Delta_L+Cm(Ct\Delta_L)^{L+1}$.
We can therefore apply the same robust phase-estimation procedure as in \Cref{col:estimate-f}. By choosing $L$ logarithmically large in $k\Lambda/\eta$, this error remains a small constant even for the largest value of $m$ used in phase estimation. This gives the following corollary.

\begin{corollary}\label{col:estimate-fL}
    Let $\eta\in(0,\Lambda)$ and $\delta\in(0,1/3)$, and choose an integer
    $L\geq\lceil C_L(1+\log(k\Lambda/\eta))\rceil$ for a sufficiently large
    universal constant $C_L$. Define
    \begin{equation}\label{eq:def-fL}
        f_L(t):=\widehat\delta_0(t)-\widehat\delta_1(t).
    \end{equation}
    For $0<t\leq c/(kL\Lambda)$, one can estimate $f_L(t)$ to additive
    error $\eta$ with probability at least $1-\delta$, using
    $\mathcal O(\eta^{-1}\log(1/\delta))$ total evolution time.
\end{corollary}

\begin{proof}
Using \eqref{eq:def-fL}, up to a global phase, the ideal state in
\Cref{lem:estimate-fL} is
\begin{equation}
    \frac{1}{\sqrt2}
    \left(\ket{0}+e^{2imtf_L(t)}\ket{1}\right)
    \ket{0^{n-1}}.
\end{equation}
By \Cref{lem:analytic-single}, there is a universal constant $C_f$
such that $|f_L(t)|\leq C_f\Lambda$. Also, by \Cref{lem:estimate-fL}, the difference between the actual and ideal states is bounded by $C_st\Delta_L+C_sm(C_st\Delta_L)^{L+1}$ for a universal constant $C_s\geq1$, where $\Delta_L=kL\Lambda$. Choose
\begin{equation}
    c\leq
    \min\left\{
        \frac{\pi}{8C_f},
        \frac{1}{128C_s}
    \right\}.
\end{equation}
Then $|2tf_L(t)|\leq\pi/4$, so the initial phase has no $2\pi$ ambiguity, and the first error term is at most $1/128$.

We apply the same robust phase-estimation procedure as in
\Cref{col:estimate-f}, using $m_j=2^j$, $j=0,\ldots,J$, with
$J$ chosen as in \Cref{col:estimate-f}, so $2^J\leq2/(t\eta)$. Since $C_st\Delta_L\leq1/4$, for every $m_j\leq2^J$ the second error term satisfies
\begin{equation}
    C_sm_j(C_st\Delta_L)^{L+1}
    \leq2C_s^2\frac{kL\Lambda}{\eta}4^{-L}
    \leq2C_s^2\frac{k\Lambda}{\eta}2^{-L}.
\end{equation}
Taking $C_L\geq\log(256C_s^2)/\log2$ makes this at most $1/128$ for every allowed $L$. Thus the circuit error is uniformly small at
every stage, and the robust phase-estimation argument of
\Cref{col:estimate-f} gives an $\eta$-additive estimate of $f_L(t)$.

Using the same failure-probability allocation
$\delta_j=\delta\,2^{j-J-1}$ as before, the total evolution time is
\begin{equation}
    T_{\rm tot}
    =
    \mathcal O\left(
        t\sum_{j=0}^J
        2^j\left(J-j+\log\frac1\delta\right)
    \right)
    =
    \mathcal O\left(
        \frac1\eta\log\frac1\delta
    \right).
\end{equation}
This proves the claim.
\end{proof}

With \Cref{col:estimate-fL}, we can estimate $f_L(t)$ efficiently at sufficiently small nonzero times. On the other hand, \Cref{lem:analytic-single} shows that $f_L(t)$ is even and analytic,
and that its value at $t=0$ is exactly the desired energy gap. Therefore, as in \Cref{thm:energy-gap-subroutine}, we recover $f_L(0)$ by extrapolating estimates of $f_L(t)$ obtained at several nonzero
times. The only additional requirement is to choose a common $L$ that is large enough for all of these estimates.

We now prove \Cref{thm:energy-single-qubit}. By
\Cref{lem:analytic-single} and \eqref{eq:def-fL},
\begin{equation}\label{eq:energy-gap-L-expansion}
    f_L(t)
    =
    f_L(0)+\sum_{j=1}^{\infty}a_jt^{2j},
    \qquad
    |a_j|\leq C_a\rho^{-2j},
\end{equation}
where $f_L(0)=\bra{e_0}K\ket{e_0}-\bra{e_1}K\ket{e_1}$,
$C_a\leq C\Lambda$, and $\rho:=\frac{c}{kL\Lambda}$.

Choose $p$ and the node accuracies $\eta_q$ exactly as in \eqref{eq:chosen-p} and \eqref{eq:energy-subroutine-node-accuracy}, respectively, as in the proof of \Cref{thm:energy-gap-subroutine}. Let
$\eta_*:=\min_{1\leq q\leq p}\eta_q$ and choose
\begin{equation}\label{eq:L-choice}
    L:=\left\lceil
        C_L\left(1+\log\frac{k\Lambda}{\eta_*}\right)
    \right\rceil.
\end{equation}
Since $S_p\leq C\sqrt p$ and $|c_q|\leq2$, substituting into \eqref{eq:energy-subroutine-node-accuracy} gives $\eta_*=\Omega(\eps/\sqrt p)$. Hence, with $p=\Theta(1+\log(\Lambda/\eps))$ by \eqref{eq:chosen-p}, we have $L=\mathcal O\left(1+\log\frac{k\Lambda}{\eps}\right)$. Choose the common step size as
\begin{equation}
    \tau
    :=
    \frac{c_\tau}{
        k\Lambda
        \left(1+\log(k\Lambda/\eps)\right)^2
    },
\end{equation}
where $c_\tau>0$ is a sufficiently small universal constant. Since both $p$ and $L$ are $\mathcal O(1+\log(k\Lambda/\eps))$, this choice ensures $p\tau\leq\rho/2$. Therefore, all nodes $t_q:=q\tau$ satisfy $t_q\leq\rho/2<c/(kL\Lambda)$, so \Cref{col:estimate-fL} applies. Moreover, $\left(\frac{p\tau}{\rho}\right)^2\leq\frac14$, and hence \Cref{lem:noisy-extrapolation} also applies.

We estimate each $f_L(t_q)$ to accuracy $\eta_q$ with failure
probability $\delta/p$. Since the extrapolation weights, node
accuracies, and their error bounds are the same as in the proof of
\Cref{thm:energy-gap-subroutine}, the same calculation gives an
$\eps$-additive estimate of $f_L(0)$ with probability at least
$1-\delta$. Likewise, the total evolution time is
\begin{equation}
    T_{\rm tot}
    =
    \mathcal O\left(
        \frac{p}{\eps}\log\frac p\delta
    \right)
    =
    \widetilde{\mathcal O}\left(
        \frac1\eps\log\frac1\delta
    \right).
\end{equation}
Finally, the step size is
\begin{equation}
    \tau
    =
    \Theta\left(
        \frac{1}{
            k\Lambda(1+\log(k\Lambda/\eps))^2
        }
    \right)
    =
    \widetilde\Theta(1/\Lambda).
\end{equation}
All known operations $W,W^\dagger,V,V^\dagger$ are products of single-qubit gates, the measurements are single-qubit measurements, and the unknown Hamiltonian is used only through forward evolution.
This proves \Cref{thm:energy-single-qubit}.

We now apply the energy-gap estimation procedures in
\Cref{thm:energy-gap-subroutine,thm:energy-single-qubit} to several Hamiltonian learning and certification tasks. These applications retain the near-optimal step size while achieving state-of-the-art total evolution time.

\section{Near-optimal local Hamiltonian learning}
We apply the energy-gap estimation procedure in \Cref{thm:energy-gap-subroutine} to learn sparse local Hamiltonians. For $k=\mathcal O(1)$, \Cref{thm:near-upper} gives an upper bound of $\widetilde{\mathcal O}(s/\eps)$ on the total evolution time, while \Cref{thm:lower-bound} proves a lower bound of $\Omega_k(s^{1-1/(2k)}/\eps)$ under its stated conditions on $n,s,\eps$. In this regime, the two bounds differ by a factor of $\widetilde{\mathcal O}(s^{1/(2k)})$.

\subsection{Upper bound}\label{subsec:near-upper}
\begin{theorem}[Sparse local Hamiltonian learning]\label{thm:near-upper}
Let $H=\sum_{P\in\Pnk}h_PP$ be traceless, $s$-sparse, $k$-local, and satisfy $\norm{H}_{\op}\leq\Lambda$, $k=\mathcal{O}(1)$. For every $\eps>0$, one can output $\widehat h$ such that
\begin{equation}
\norm{\widehat h-h}_2\leq \eps
\end{equation}
with high probability. The total evolution time is $T_{\rm tot}=\widetilde{\mathcal{O}}(s/\eps)$, and the step size is $\tau=\widetilde{\Theta}(1/\Lambda)$.
\end{theorem}

The formal statement and proof are given in \Cref{app:near-upper}. We first describe the main idea. Choose independently and uniformly $\sigma\in\{X,Y,Z\}^n$ and $x,r\in\{0,1\}^n$. Define the gap function as 
\begin{equation}\label{eq:gap-def}
G(\sigma,x,r)
:=\frac12\left(
\bra{x}_{\sigma}H\ket{x}_{\sigma}-\bra{x\oplus r}_{\sigma}H\ket{x\oplus r}_{\sigma}
\right),
\end{equation}
where the product state $\ket{x}_\sigma$ has eigenvalue $(-1)^{x_i}$ for $\sigma_i$ on qubit $i$. We derive that, for every $\sigma,x,r$,
\begin{equation}\label{eq:linear-measurement}
G(\sigma,x,r)=\sum_{P\in\Pnk}z_P\phi_P(\sigma,x,r),\qquad
z_P:=\frac{h_P}{\sqrt{2}\,3^{w_P/2}},
\end{equation}
where
\begin{equation}
    \phi_P(\sigma,x,r)
:=
\sqrt{2}\,3^{w_P/2}
\mathbf 1\{P\preceq\sigma\}
(-1)^{x\cdot S_P}
\mathbf 1\{r\cdot S_P=1\}.
\end{equation}
The functions $\{\phi_P:P\in\Pnk\}$ satisfy 
\begin{equation}
    \E[\phi_P\phi_Q]=\delta_{P,Q},\quad \norm{\phi_P}_{\infty}\leq\sqrt{2}\,3^{k/2}.
\end{equation}
Recall the definitions of $P\preceq\sigma$, $S_P$ in \Cref{subsec:notation}.

Therefore, we can apply \Cref{lem:RIP} and retrieve the sparse coefficients $\{z_P\}$ within $\eps/(\sqrt2 3^{k/2})$-additive error in the $\ell_2$-norm, hence recovering the Hamiltonian coefficients $\{h_P\}$ within $\eps$-additive error in the $\ell_2$-norm. Since $\{z_P\}$ are $s$-sparse, the sufficient number of samples is $\mathcal{\widetilde{O}}(s)$. If $\eps\geq\Lambda$, the zero estimate already has $\ell_2$ error at most $\eps$; thus it suffices to consider $0<\eps<\Lambda$. Therefore, we can apply \Cref{thm:energy-gap-subroutine}. Using evolution time $T_{\rm tot}=\mathcal{\widetilde{O}}(1/\eps)$ and step size $\widetilde{\Theta}(1/\Lambda)$, we can estimate one sample with desired accuracy. Hence, the total evolution time of the algorithm becomes $\mathcal{\widetilde{O}}(s/\eps)$.

The number of queries $\mathcal N$ is bounded by $\mathcal N\leq T_{\rm tot}/\tau=\widetilde{\mathcal O}\left(s\Lambda/\eps\right)$, since every nonzero evolution time is at least the step size $\tau$. For the classical postprocessing, the dominant cost comes from the
sparse-recovery step, which is performed by solving the convex
optimization problem in \eqref{eq:BPDN}. This optimization can be
solved in polynomial time in its input size. The corresponding
sampling matrix has dimension $m\times M$, where
$m=\widetilde{\mathcal O}(s)$ and
$M=|\Pnk|=\mathcal O(n^k)$ for constant $k$. Therefore, the input
size of the sparse-recovery problem is polynomial in $n$ and $s$,
and the total classical postprocessing time is also polynomial in
$n$ and $s$ for fixed $k$.

\subsection{Lower bound}\label{subsec:lower-bound}

We will use the following lemma, which follows from Gilbert's greedy construction of a binary code whose distinct codewords are separated by Hamming distance at least \(d\) \cite[Theorem~4.2.1 and Section~4.2.1]{guruswami2012essential}.
\begin{lemma}[Greedy Hamming packing]
\label{lem:greedy-Hamming-packing}
Let $\mathcal G\subseteq\{-1,1\}^s$ and let $1\leq d\leq s$.  There is
a subset $\mathcal C\subseteq\mathcal G$ such that
\begin{equation}
    d_{\mathrm H}(z,z')\geq d
    \quad\text{for all distinct }z,z'\in\mathcal C,
    \qquad
    \abs{\mathcal C}
    \geq
    \frac{\abs{\mathcal G}}
    {\displaystyle\sum_{j=0}^{d-1}\binom{s}{j}}.
    \label{eq:greedy-Hamming-packing}
\end{equation}
In particular, if $\abs{\mathcal G}\geq2^{s-1}$, then for every fixed
$0<\delta<1/2$ there is, for all sufficiently large $s$, a set
$\mathcal C\subseteq\mathcal G$ and a constant $c_\delta>0$ such that
\begin{equation}
    d_{\mathrm H}(z,z')\geq\delta s
    \quad\text{for all distinct }z,z'\in\mathcal C,
    \qquad
    \abs{\mathcal C}\geq\exp(c_\delta s).
    \label{eq:asymptotic-Hamming-packing}
\end{equation}
\end{lemma}

\begin{proof}
Initialize $\mathcal C=\varnothing$ and $\mathcal R=\mathcal G$.  While
$\mathcal R$ is nonempty, choose any $z\in\mathcal R$, add it to
$\mathcal C$, and remove from $\mathcal R$ every point at Hamming
distance at most $d-1$ from $z$.
Distinct selected points have distance at least $d$.  Each iteration
removes at most
\begin{equation}
    V_s(d-1):=\sum_{j=0}^{d-1}\binom{s}{j}
\end{equation}
points, the volume of a radius-$(d-1)$ ball in $\{-1,1\}^s$.  Since the
iterations remove every point of the original set $\mathcal G$, we have
$\abs{\mathcal G}\leq\abs{\mathcal C}V_s(d-1)$, proving
\eqref{eq:greedy-Hamming-packing}.

For the final claim, take $d=\lceil\delta s\rceil$ and use the 
entropy bound on binomial coefficients \cite[Proposition 3.3.3]{guruswami2012essential}
\begin{equation}
    \sum_{j=0}^{d-1}\binom{s}{j}
    \leq 2^{sH_2(\delta)},
    \qquad
    H_2(\delta):=-\delta\log_2\delta-(1-\delta)\log_2(1-\delta)<1.
\end{equation}
Then
\begin{equation}
    \abs{\mathcal C}
    \geq 2^{(1-H_2(\delta))s-1},
\end{equation}
which implies \eqref{eq:asymptotic-Hamming-packing} with any
$c_\delta<(1-H_2(\delta))\log 2$.
\end{proof}

We now state and prove the lower bound.
\begin{theorem}
\label{thm:lower-bound}
Let $k=\mathcal O(1)$, and let $q$ be the smallest integer such that
$3^k\binom{q}{k}\geq s$.  Suppose $n\geq q$.  There is a constant
$c_k>0$ such that, whenever
\begin{equation}
\label{eq:upper-bound-eps}
    0<\eps\leq c_k\Lambda s^{-1/(2k)},
\end{equation}
any algorithm which almost surely uses finitely many experiments, and, for every traceless, $s$-sparse, $k$-local
Hamiltonian $H=\sum_{P\in\Pnk}h_PP$ with $\norm{H}_{\op}\leq\Lambda$,
outputs $\widehat h$ satisfying
\begin{equation}
    \Prb\left[\norm{h-\widehat h}_2\leq\eps\right]\geq\frac23
\end{equation}
requires worst-case expected total evolution time \(T_{\rm tot}=\Omega_k(s^{1-1/(2k)}/\varepsilon)\).
\end{theorem}
\begin{proof}
We first construct a packing of Hamiltonians.  By the choice of $q$, we
may fix distinct weight-$k$ Paulis $P_1,\ldots,P_s$ on the first $q$
qubits.  For $z\in\{-1,1\}^s$, let
\begin{equation}
    A_z:=\sum_{j=1}^s z_jP_j.
\end{equation}
Since \(P_j^2=I\), the matrix Khintchine inequality~\cite[Theorem~4.1.1]{Tropp2015} gives
\begin{equation}
    \E_z\norm{A_z}_{\op}
    \leq C\sqrt{\log(2^{q+1})}
       \norm{\sum_{j=1}^sP_j^2}_{\op}^{1/2}
    \leq C_0\sqrt{sq}.
\end{equation}
Markov's inequality therefore gives a set
\(\mathcal G\subseteq\{-1,1\}^s\) of cardinality at least \(2^{s-1}\) on
which \(\norm{A_z}_{\op}\leq 2C_0\sqrt{sq}\).  Fix any constant
$0<c_0<1/2$.  Applying \Cref{lem:greedy-Hamming-packing} with
$d=\lceil c_0s\rceil$ gives a constant $c_1>0$ and a set
$\mathcal C\subseteq\mathcal G$ such that, for all sufficiently large
$s$,
\begin{equation}
    \abs{\mathcal C}\geq \exp(c_1s),
    \qquad
    d_{\mathrm H}(z,z')\geq c_0s
    \quad\text{for all distinct }z,z'\in\mathcal C.
    \label{eq:lower-bound-code}
\end{equation}

Define
\begin{equation}
    K_z:=\frac{A_z}{2C_0\sqrt{sq}},
    \qquad z\in\mathcal C.
\end{equation}
Then $\norm{K_z}_{\op}\leq1$, and, writing $\kappa_z$ for its Pauli
coefficient vector, 
\begin{equation}
    \norm{\kappa_z-\kappa_{z'}}_2
    =\frac{\sqrt{d_{\mathrm H}(z,z')}}{C_0\sqrt{sq}}
    \geq\frac{c_2}{\sqrt q}.
    \label{eq:lower-bound-separation-K}
\end{equation}
Set
\begin{equation}
    H_z:=\lambda K_z,
    \qquad
    \lambda:=\frac{4\eps\sqrt q}{c_2}.
    \label{eq:lower-bound-hard-family}
\end{equation}
Because $q=\Theta_k(s^{1/k})$, the assumed upper bound \eqref{eq:upper-bound-eps} on $\eps$,
with $c_k$ sufficiently small, ensures $\lambda\leq\Lambda$.
Consequently every $H_z$ belongs to the promised class, while
\eqref{eq:lower-bound-separation-K} gives
\begin{equation}
    \norm{h_z-h_{z'}}_2\geq4\eps
    \quad(z\neq z').
    \label{eq:lower-bound-separation-H}
\end{equation}

We next bound how many labels a protocol of short total evolution time
can distinguish.  We use the rooted-tree representation of
Huang et al.~\cite[Appendix~G.3.3]{huang2023learning}:
each node \(v\) records the past outcomes and specifies the next
experiment.  Write \(t_v\) for its evolution time and
\(T_v=t_v+\max_c T_c\) for the maximum remaining path cost, where
\(c\) ranges over its children and leaves have cost zero.

We first consider a protocol that uses at most \(D\) experiments, for some finite \(D\), and total evolution time at most \(T\) on every branch. 
Its rooted-tree representation therefore has depth at most \(D\). We will later weaken these two assumptions to using finitely many experiments almost surely and having a bound on expected total evolution time.

To simplify the description of adaptive experiments,
we purify the initial state and represent each known quantum operation by an isometry on a larger system, retaining any ancillary systems that would otherwise be discarded. In particular, implement each measurement coherently by storing its outcome in mutually orthogonal states of a record register, and condition subsequent operations on these records while preserving the protocol’s quantum memory. Reading the outcome registers and ignoring the additional systems recovers the original protocol, so retaining these systems can only strengthen the learner. This gives a linear description of the protocol that preserves the unnormalized amplitudes of all branches.

We define the \emph{coherent subtree map} of a node $v$ to be the linear map taking the quantum state entering node \(v\) to the direct sum of the unnormalized states at its descendant leaves, with distinct branches labeled by orthogonal measurement records.
Let \(\mathsf A_{v,z}(w)\) denote the coherent subtree map obtained by replacing every unknown evolution \(e^{-itH_z}\) in the subtree rooted at \(v\) with \(e^{-iwtH_z}\), while keeping all known operations unchanged. 
We define its Dyson coefficients through the expansion
\begin{equation}
    \mathsf A_{v,z}(w)=\sum_{r=0}^{\infty}w^r\mathsf A_{v,z}^{(r)}.
\end{equation}
Thus, \(\mathsf A_{v,z}^{(r)}\) collects all terms with a total of \(r\) Hamiltonian insertions across the evolution intervals in the subtree.
Similarly, expanding the coherent linear map for the local experiment at node \(v\) in powers of \(w\) gives an order-\(r\) coefficient with operator norm at most \((\lambda t_v)^r/r!\). This follows by expanding the local evolution segments, whose durations sum to $t_v$, using $\|H_z\|_{\op}\leq\lambda$ and the fact that the known controls and purified measurement are isometries.
Since the measurement is purified to be a \(z\)-independent isometry and the norm of a direct sum of child maps is their maximum norm, induction from the leaves gives
\begin{equation}
    \begin{aligned}
    \norm{\mathsf A_{v,z}^{(r)}}_{\op}
    &\leq\sum_{j=0}^r
      \max_c\norm{\mathsf A_{c,z}^{(j)}}_{\op}
      \frac{(\lambda t_v)^{r-j}}{(r-j)!}\\
    &\leq\sum_{j=0}^r
      \frac{(\lambda\max_c T_c)^j}{j!}
      \frac{(\lambda t_v)^{r-j}}{(r-j)!}
    =\frac{(\lambda T_v)^r}{r!}.
    \end{aligned}
\end{equation}
At the root, \(T_v\leq T\), so the final pure state, including the
leaf record, has the expansion
\begin{equation}
    \ket{\psi_z}=\sum_{r=0}^{\infty}\ket{\psi_z^{(r)}},
    \quad
    \norm{\ket{\psi_z^{(r)}}}\leq\frac{(\lambda T)^r}{r!}.
    \label{eq:lower-bound-Dyson}
\end{equation}

Since $K_z$ depends linearly on $z$, the truncation through order $R$ is
a vector-valued multilinear polynomial of degree at most $R$, using $z_j^2=1$:
\begin{equation}
    \sum_{r=0}^R\ket{\psi_z^{(r)}}
    =\sum_{\substack{S\subseteq[s]\\ |S|\leq R}}
      \left(\prod_{j\in S}z_j\right)\ket{v_S},
    \label{eq:lower-bound-multilinear}
\end{equation}
where the vectors $\ket{v_S}$ do not depend on $z$.  Thus all truncated
states lie in a common subspace of dimension at most
\begin{equation}
    L_R:=\sum_{j=0}^R\binom{s}{j}.
    \label{eq:lower-bound-rank}
\end{equation}
Choose a sufficiently small constant $\alpha>0$ and set
\(R=\lfloor\alpha s\rfloor\).  The standard Hamming-ball bound~\cite[Proposition~3.3.3]{guruswami2012essential} permits
us to choose $\alpha$ so that $L_R\leq\exp(c_1s/2)$.  

If we assume
$\lambda T<(R+1)/(4e)$, then \eqref{eq:lower-bound-Dyson} and the
factorial bound give, uniformly in $z$,
\begin{equation}
    \norm{\sum_{r>R}\ket{\psi_z^{(r)}}}
    \leq
    2\left(\frac{e\lambda T}{R+1}\right)^{R+1}
    \leq2\cdot4^{-(R+1)}
    =:\eta_R.
    \label{eq:lower-bound-tail}
\end{equation}
Let $\Pi_R$ project onto the subspace in
\eqref{eq:lower-bound-multilinear}, and let $\{M_z\}_{z\in\mathcal C}$
be any measurement intended to decode a uniform label $z$.  Since
$\norm{(I-\Pi_R)\ket{\psi_z}}\leq\eta_R$ and
$\Pi_R\ketbra{\psi_z}{\psi_z}\Pi_R\preceq\Pi_R$, its average success
probability is at most
\begin{equation}
    \frac{1}{\abs{\mathcal C}}
    \sum_{z\in\mathcal C}\bra{\psi_z}M_z\ket{\psi_z}
    \leq
    \frac{L_R}{\abs{\mathcal C}}+3\eta_R,
    \label{eq:lower-bound-decoding}
\end{equation}
The omitted components contribute at most $2\eta_R+\eta_R^2\leq3\eta_R$. For the remaining contribution, we use

\begin{equation}
\begin{aligned}
\frac{1}{\abs{\mathcal C}}
    \sum_{z\in\mathcal C}\bra{\psi_z}\Pi_R M_z\Pi_R\ket{\psi_z}
&=\frac{1}{\abs{\mathcal C}}
    \sum_{z\in\mathcal C}\Tr[M_z\Pi_R\ketbra{\psi_z}{\psi_z}\Pi_R]\\
&\leq\frac{1}{\abs{\mathcal C}}
    \sum_{z\in\mathcal C}\Tr[M_z\Pi_R]
=\frac{\Tr\Pi_R}{\abs{\mathcal C}}
\leq\frac{L_R}{\abs{\mathcal C}}.
\end{aligned}
\end{equation}

For all sufficiently large $s$, the right-hand side of \eqref{eq:lower-bound-decoding} is smaller than
$2/3$.  On the other hand, by \eqref{eq:lower-bound-separation-H}, an
$\eps$-accurate coefficient estimate identifies $z$ by nearest-neighbor
decoding and hence has success probability at least $2/3$.  
This contradiction means our previous hypothesis of \(\lambda T<(R+1)/(4e)\) must be false, which gives us
\begin{equation}
    \lambda T\geq \frac{R+1}{4e}\geq \frac{\alpha s}{4e}=\Omega(s).
\end{equation}
Using \eqref{eq:lower-bound-hard-family} and
$q=\Theta_k(s^{1/k})$ now gives
\begin{equation}
    T=\Omega\left(\frac{s}{\lambda}\right)
    =\Omega_k\left(\frac{s^{1-1/(2k)}}{\eps}\right).
\end{equation}
The argument also applies to any fixed positive decoding success
probability, since the right-hand side of
\eqref{eq:lower-bound-decoding} tends to zero.  

The preceding proof assumes sufficiently large $s$ beyond a constant threshold.
For $s$ below the threshold, take $c_k\leq1/2$ and distinguish $H_\pm=\pm2\eps P$ for 
a fixed nonidentity Pauli \(P\) of weight at most \(k\)  gives \(T=\Omega(1/\eps)\) by the usual
hybrid bound, which also follows by the same subtree induction.

If a protocol uses finitely many experiments almost surely but has no fixed bound on their number, truncate it after $D$ experiments. Since the hard family is finite, taking $D$ sufficiently large makes the loss of success probability arbitrarily small for every Hamiltonian in the family. The bounded-depth argument then applies.

The expected-time extension follows from the truncation argument of
Brahmachari et al.~\cite[Lemma~S9]{brahmachari2026static}.
If the expected evolution time is at most \(B\) for every input,
stop the protocol before its accumulated evolution time exceeds
\(12B\).  Markov's inequality bounds the additional failure probability
by \(1/12\), so the truncated protocol still succeeds with probability
at least \(7/12\).  Applying the preceding bound yields
\(B=\Omega_k(s^{1-1/(2k)}/\eps)\).

\end{proof}

\section{Learning local Hamiltonians with single-qubit gates}\label{sec:single-qubit-learning}
Single-qubit gates and measurements are among the most experimentally accessible quantum operations, which motivates studying Hamiltonian learning under this restricted control model. Using \Cref{thm:energy-single-qubit}, we show that the coefficient vector can be learned in the $\ell_\infty$-norm with state-of-the-art total evolution time $\widetilde{\mathcal O}(s/\eps)$. For the coefficient $\ell_2$-norm, restricting to single-qubit operations incurs only an additional factor of $\sqrt n$, giving total evolution time $\widetilde{\mathcal O}(s\sqrt n/\eps)$. In both cases, the step size remains near optimal.

\begin{theorem}\label{thm:single-qubit-learning}
Let $H=\sum_{P\in\Pnk}h_PP$ be traceless, $s$-sparse, $k$-local, and satisfy $\norm{H}_{\op}\leq\Lambda$, $k=\mathcal{O}(1)$. For every $\eps>0$, one can output $\widehat h$ such that
\begin{equation}
\norm{\widehat h-h}_\infty\leq \eps
\end{equation}
using $\widetilde{\mathcal{O}}(s/\eps)$ total evolution time with high probability. The algorithm can also output 
$\widehat h$ such that
\begin{equation}
\norm{\widehat h-h}_2\leq \eps
\end{equation}
using $\widetilde{\mathcal{O}}(\sqrt{n}s/\eps)$ total evolution time with high probability. The step size is $\widetilde{\Theta}(1/\Lambda)$, and the algorithm only uses single-qubit gates and measurements with forward evolution $e^{-iH\tau}$.
\end{theorem}

We state the formal version of the theorem in \Cref{app:single-qubit-learning} with a complete proof. Here, we explain the high-level idea. The main difference from \Cref{thm:near-upper} is that we replace the
multi-qubit gap in \eqref{eq:gap-def} by one-qubit gaps. Choose independently and uniformly
$\sigma\in\{X,Y,Z\}^n$ and $x\in\{0,1\}^n$.
For every $j\in[n]$, define
\begin{equation}\label{eq:single-qubit-gap}
G_j(\sigma,x)
:=
\frac12\left(
\bra{x}_{\sigma}H\ket{x}_{\sigma}
-
\bra{x\oplus e_j}_{\sigma}
H
\ket{x\oplus e_j}_{\sigma}
\right).
\end{equation}
The two states in \eqref{eq:single-qubit-gap} differ only on qubit $j$.
Hence each gap is compatible with the single-qubit energy-difference
procedure, and \Cref{thm:energy-single-qubit} allows us to estimate the values $G_j(\sigma,x)$ using only single-qubit gates and measurements, with forward evolution.

A direct calculation gives
\begin{equation}\label{eq:single-qubit-linear-measurement}
G_j(\sigma,x)
=
\sum_{\substack{P\in\Pnk\\j\in\operatorname{supp}(P)}}
z_P\phi_{j,P}(\sigma,x),
\qquad
z_P:=3^{-w_P/2}h_P,
\end{equation}
where
\begin{equation}\label{eq:single-qubit-functions}
\phi_{j,P}(\sigma,x)
:=
3^{w_P/2}
\mathbf 1\{P\preceq\sigma\}
(-1)^{\sum_{i\in\operatorname{supp}(P)}x_i}.
\end{equation}
For every fixed $j$, the functions $\left\{\phi_{j,P}:P\in\Pnk,\ j\in\operatorname{supp}(P)\right\}$ satisfy
\begin{equation}\label{eq:single-qubit-orthonormality}
\E_{\sigma,x}
\left[
\phi_{j,P}(\sigma,x)
\phi_{j,Q}(\sigma,x)
\right]
=
\delta_{P,Q},
\qquad
\norm{\phi_{j,P}}_\infty
\leq
3^{k/2}.
\end{equation}
Define
\begin{equation}
    h_P^{(j)}
    :=
    \begin{cases}
        h_P,
        & j\in\operatorname{supp}(P),\\
        0,
        & j\notin\operatorname{supp}(P),
    \end{cases}
    \qquad
    d_j:=\norm{h^{(j)}}_0,
\end{equation}
The rescaled coefficient vector in \eqref{eq:single-qubit-linear-measurement} has the same support as $h^{(j)}$ and is therefore $d_j$-sparse. By \Cref{lem:RIP}, we can recover the coefficients $h^{(j)}$ within $\eps$-additive error in the $\ell_2$-norm using $\mathcal{\widetilde{O}}(\max\{1,d_j\})$ samples. 

If $\eps\geq\Lambda$, the zero estimate already has $\ell_2$ and $\ell_\infty$ error at most $\eps$; thus it suffices to consider $0<\eps<\Lambda$. Therefore, we apply \Cref{thm:energy-single-qubit}. Using evolution time $\mathcal{\widetilde{O}}(1/\eps)$ and step size $\widetilde{\Theta}(1/\Lambda)$, we can estimate one sample with desired accuracy. Hence, the total evolution time of recovering $h^{(j)}$ becomes $\mathcal{\widetilde{O}}(\max\{1,d_j\}/\eps)$. Therefore, repeating the procedure for every $j$ recovers the entire coefficients $h$ within $\eps$-additive error in the $\ell_\infty$-norm using
\begin{equation}\label{eq:sum-dj}
\sum_{j=1}^n\mathcal{\widetilde{O}}(\max\{1,d_j\}/\eps)=\mathcal{\widetilde{O}}((n+ks)/\eps)     
\end{equation}
total evolution time. We can recover $h$ within $\eps$-additive error in the $\ell_2$-norm using $\mathcal{\widetilde{O}}((n+ks)\sqrt{n}/\eps)$ total evolution time as explained in \Cref{app:single-qubit-learning}.

Since $d_j$ is unknown, we double its estimate as $1,2,4,\ldots$ and verify each candidate using measurement data independent of the recovery data. Recovery is performed at the smaller error required for guaranteed acceptance by the verifier, while accepted candidates have local $\ell_2$ error below $\eps$. For fixed $k$, this preserves \eqref{eq:sum-dj}. Parallel coordinate discovery removes its additive $n$ term by identifying which coordinates require recovery before applying the preceding procedure. The details are given in \Cref{app:single-qubit-learning}.

The number of queries $\mathcal N$ is bounded by $\mathcal N\leq T_{\rm tot}/\tau$, since every nonzero evolution time is at least the step size $\tau$. Hence, $\mathcal N=\widetilde{\mathcal O}(s\Lambda/\eps)$ and $\mathcal N=\widetilde{\mathcal O}(s\Lambda\sqrt n/\eps)$ for the $\ell_\infty$ and $\ell_2$ guarantees, respectively. For the classical postprocessing, the dominant cost comes from the
sparse-recovery step, which is performed by solving the convex
optimization problem in \eqref{eq:BPDN}. This optimization can be
solved in polynomial time in its input size. The corresponding
sampling matrix has dimension $m\times M$, where
$m=\widetilde{\mathcal O}(s)$ and
$M=|\Pnk|=\mathcal O(n^k)$ for constant $k$. Therefore, the input
size of the sparse-recovery problem is polynomial in $n$ and $s$,
and the total classical postprocessing time is also polynomial in
$n$ and $s$ for fixed $k$.

\section{More applications}\label{sec:applications}
Beyond full Hamiltonian learning, the energy-gap estimation procedures also lead to efficient algorithms for several related tasks. In this section, we show applications to learning a single Hamiltonian coefficient, agnostic local Hamiltonian learning, and tolerant Hamiltonian certification, while retaining the near-optimal step size.

\subsection{Learning a single coefficient of the Hamiltonian}
We consider the single-coefficient learning question raised in \cite[Section~1.4]{bakshi2024structure}. Under the operator-norm bound used here, a specified Pauli coefficient can be estimated with Heisenberg-limited total evolution time and near-optimal step size, up to logarithmic factors. For constant $\Lambda$, the evolution-time bound is independent of $n$. This addresses the system-size dependence in that question, but does not remove the precision-dependent logarithms or establish a bound in terms of the local interaction norm used in that work.

\begin{theorem}\label{thm:single-coeff}
Let $H=\sum_P h_P P$ be a traceless Hamiltonian satisfying $\|H\|_{\op}\leq\Lambda$, and let $P$ be a specified nonidentity Pauli
operator. For every $\eps>0$, we can estimate $h_P$ to additive error $\eps$ with high probability using total evolution time $\mathcal{\widetilde{O}}(1/\eps)$ and step size $\widetilde\Theta(1/\Lambda)$.
\end{theorem}
\begin{proof}
If $\eps\geq\Lambda$, we output zero, so assume $0<\eps<\Lambda$.
We use $n$ system qubits and $n$ ancillary qubits. Evolution under
$H$ on the system is equivalent to evolution under $H\otimes I$ on
the $2n$ qubits, and $\norm{H\otimes I}_{\op}=\norm{H}_{\op}\leq\Lambda$.
Hence, we may apply \Cref{thm:energy-gap-subroutine} to $H\otimes I$.

Let $\ket{x}:=\ket{\Omega}$ and
$\ket{y}:=(P\otimes I)\ket{\Omega}$, where
$\ket{\Omega}:=2^{-n/2}\sum_{v\in\{0,1\}^n}\ket{v}\ket{v}$.
Since $P\neq I$, we have
$\braket{x}{y}=2^{-n}\Tr(P)=0$. We now construct the unitary $W$ required by
\Cref{thm:energy-gap-subroutine}. Let
$B:=(\prod_{j=1}^n\mathrm{CNOT}_{j,n+j})(\prod_{j=1}^n H_j)$,
so that $B\ket{0^{2n}}=\ket{\Omega}$. Since $B$ is a Clifford
circuit, $B^\dagger(P\otimes I)B$ is a Pauli operator. Therefore,
$B^\dagger(P\otimes I)\ket{\Omega}$ is a computational-basis state
up to a known phase. Since $P\neq I$, this state is not
$\ket{0^{2n}}$. Write $B^\dagger(P\otimes I)\ket{\Omega}=\omega\ket{b}$, where $b\neq0^{2n}$ and $\omega\in\{1,-1,i,-i\}$ are known. Construct $C_P$ by first applying $\operatorname{diag}(1,\omega)$ to qubit $1$, moving its bit to any position where $b$ equals one using a SWAP if necessary, and using that position as the control of CNOT gates to the other positions where $b$ equals one. Then $C_P\ket{0^{2n}}=\ket{0^{2n}}$ and $C_P\ket{10^{2n-1}}=\omega\ket b$. Thus
\begin{equation}
    W:=BC_PH_1
\end{equation}
satisfies $W\ket{0^{2n}}=(\ket{x}+\ket{y})/\sqrt2$ and $W\ket{10^{2n-1}}=(\ket{x}-\ket{y})/\sqrt2$. Here $H_j$ denotes a Hadamard gate, and both $B$ and $C_P$ use $\mathcal O(n)$ known one- and two-qubit gates.

Applying \Cref{thm:energy-gap-subroutine} to these two states estimates
\begin{equation}
\begin{aligned}
    &\left(\frac{\bra{x}+\bra{y}}{\sqrt2}\right)
    (H\otimes I)
    \left(\frac{\ket{x}+\ket{y}}{\sqrt2}\right)
    -
    \left(\frac{\bra{x}-\bra{y}}{\sqrt2}\right)
    (H\otimes I)
    \left(\frac{\ket{x}-\ket{y}}{\sqrt2}\right) \\
    &\qquad=
    \bra{x}(H\otimes I)\ket{y}
    +
    \bra{y}(H\otimes I)\ket{x}
    =
    2h_P.
\end{aligned}
\end{equation}
Estimate this gap to additive error $\eps$ with failure probability at most $\delta$. By the proof of \Cref{thm:energy-gap-subroutine}, its cost is $\mathcal O((p/\eps)\log(p/\delta))$, with $p=\Theta(1+\log(\Lambda/\eps))$, which gives the total evolution time $\mathcal{\widetilde{O}}(1/\eps)$. Dividing by two produces an estimate of $h_P$ with error at most $\eps/2\leq\eps$, at the stated step size.
\end{proof}

We can further use \Cref{thm:single-coeff} to learn the best $k$-local approximation of a general Hamiltonian in the Frobenius norm. Since the Pauli operators form an orthonormal basis under the normalized Frobenius inner product, this approximation is obtained by retaining exactly the Pauli coefficients of weight at most $k$. Thus, applying \Cref{thm:single-coeff} to all such Pauli operators gives the following result.

\begin{theorem}[Agnostic local Hamiltonian learning]\label{thm:agnostic-learning}
Assume that the traceless Hamiltonian $H=\sum_P h_P P$ satisfies $\|H\|_\op\le \Lambda$. Let the Frobenius-nearest $k$-local Hamiltonian be represented as $H_k=\sum_{P\in\Pnk}h_P^{(k)} P$, with coefficient vector $h^{(k)}$. For $k=\mathcal O(1)$, we can obtain $\hat{h}^{(k)}$ with high probability such that
\begin{equation}
    \norm{h^{(k)}-\hat{h}^{(k)}}_{\infty} \le \eps,
\end{equation}
using total evolution time $\mathcal{\widetilde O}(n^k/\eps)$ and step size $\widetilde\Theta(1/\Lambda)$.
\end{theorem}
\begin{proof}
Let $M:=|\Pnk|=\mathcal O(n^k)$ and fix $\delta\in(0,1/3)$.
Apply \Cref{thm:single-coeff} to each $P\in\Pnk$ with accuracy
$\eps$ and failure probability $\delta/M$. Since
$h_P^{(k)}=h_P$ for $P\in\Pnk$, a union bound gives
$\norm{\widehat h^{(k)}-h^{(k)}}_\infty\leq\eps$
with probability at least $1-\delta$.
The total evolution time is
$\widetilde{\mathcal O}(M/\eps)
=\widetilde{\mathcal O}(n^k/\eps)$.
All estimates use the same accuracy and step size.
\end{proof}

\subsection{Certifying Hamiltonians with single-qubit gates and measurements}
We next apply the energy-gap estimation procedure to Hamiltonian certification. With general control, we achieve near-optimal scaling in both total evolution time and step size, using $\widetilde{\Theta}(1/\eps)$ total evolution time and step size $\widetilde{\Theta}(1/\Lambda)$. We then extend the certification procedure to the single-qubit control setting. This restriction preserves the near-optimal step size and incurs only an additional factor of $n^{3/2}$ in the total evolution time.

\begin{theorem}\label{thm:cert}
Let $H=\sum_{P\in\Pnk}h_PP$ be traceless, $k$-local, and satisfy $\norm{H}_{\op}\leq\Lambda$, $k=\mathcal{O}(1)$. Given a known traceless $k$-local Hamiltonian $H_0$ and $\eps>0$, promised that one of the following conditions holds, one determines whether
\begin{equation}
    \|H-H_0\|_F \ge \eps \quad \text{or} \quad \|H-H_0\|_F \le \eps/12^k
\end{equation}
with high probability. The total evolution time is $\widetilde{\mathcal O}(1/\eps)$, and the step size is $\widetilde{\Theta}(1/\Lambda)$.
\end{theorem}

\begin{theorem}\label{thm:cert-single}
Let $H=\sum_{P\in\Pnk}h_PP$ be traceless, $k$-local, and satisfy $\norm{H}_{\op}\leq\Lambda$, $k=\mathcal{O}(1)$. Given a known traceless $k$-local Hamiltonian $H_0$ and $\eps>0$, promised that one of the following conditions holds, one determines whether
\begin{equation}
    \|H-H_0\|_F \ge \eps \quad \text{or} \quad \|H-H_0\|_F \le \eps/12^k
\end{equation}
with high probability. The total evolution time is $\widetilde{\mathcal{O}}(n^{3/2}/\eps)$, and the step size is $\widetilde{\Theta}(1/\Lambda)$. The algorithm uses only single-qubit gates and measurements with forward evolution $e^{-iH\tau}$. 
\end{theorem}

We prove both theorems by employing two lemmas. We utilize the method used in \cite{BluhmEtAl2026}, which employs hypercontractivity of $k$-local Hamiltonians to obtain optimal certification for Hamiltonians of constant locality.

\begin{lemma}[Far case]
\label{lem:far-gaps}
Let $A$ be a traceless $k$-local Hamiltonian with
$\|A\|_F\geq\eps>0$, where $k\geq1$ and
$\|A\|_F^2=2^{-n}\operatorname{Tr}(A^2)$.
Independently sample product states $x_\ell,y_\ell$, each
single-qubit factor being a uniformly random Pauli eigenstate,
and uniform indices $j_\ell\in[n]$.
Then
\begin{align}
    \Pr\left[
        \frac1N\sum_{\ell=1}^N
        \left|\bra{y_\ell}A\ket{y_\ell}
              -\bra{x_\ell}A\ket{x_\ell}\right|^2
        \geq\frac{\eps^2}{3^k}
    \right]
    &\geq\frac78,
    \label{eq:far-global}\\
    \Pr\left[
        \frac1N\sum_{\ell=1}^N
        \left|\bra{x_\ell\oplus e_{j_\ell}}A
                    \ket{x_\ell\oplus e_{j_\ell}}
              -\bra{x_\ell}A\ket{x_\ell}\right|^2
        \geq\frac{\eps^2}{n3^k}
    \right]
    &\geq\frac78,
    \label{eq:far-single}
\end{align}
provided $N\geq64k9^k$ and $N\geq64kn9^k$, respectively. Here $x_\ell\oplus e_{j_\ell}$ replaces the $j_\ell$-th factor by its orthogonal state.
\end{lemma}

\begin{proof}
It suffices to prove the claim for $\|A\|_F=\eps$, since all the energy gaps scale linearly with $A$.

We first establish a fourth-moment bound that will be used below. For any Hermitian operator $B$ and a uniformly random product Pauli
eigenstate $\ket{x}$,
\begin{equation}\label{eq:product-fourth-moment}
    \mathbb E_x\!\left[\bra{x}B\ket{x}^4\right]
    \leq
    \|B\|_F^4.
\end{equation}
For one qubit, write $B=aI+b\cdot\sigma$. A uniformly random Pauli
eigenstate has Bloch vector $r$ uniformly distributed over
$\{\pm e_1,\pm e_2,\pm e_3\}$, and therefore $\mathbb E_r[(a+b\cdot r)^4]=a^4+2a^2\norm b_2^2+\frac13\sum_{\nu=1}^3b_\nu^4\leq(a^2+\norm b_2^2)^2=\norm B_{\F}^4$. For the induction step, write $B=\sum_{\nu=0}^3\sigma_\nu\otimes B_\nu$, where $(\sigma_0,\sigma_1,\sigma_2,\sigma_3)=(I,X,Y,Z)$, and condition on the product state $\ket{x'}$ of the remaining qubits. Setting $b_\nu(x'):=\bra{x'}B_\nu\ket{x'}$, the one-qubit bound gives
\begin{equation}
    \mathbb E_{x_1}\!\left[
        \bra{x_1,x'}B\ket{x_1,x'}^4
        \,\middle|\,x'
    \right]
    \leq
    \left(\sum_{\nu=0}^3 b_\nu(x')^2\right)^2.
\end{equation}
After averaging over $x'$, apply the triangle inequality to the functions $b_\nu(x')^2$: $(\mathbb E_{x'}[(\sum_\nu b_\nu^2)^2])^{1/2}\leq\sum_\nu(\mathbb E_{x'}b_\nu^4)^{1/2}$. This gives
\begin{equation}
\begin{aligned}
    \left(
        \mathbb E_x[\bra{x}B\ket{x}^4]
    \right)^{1/2}
    &\leq
    \sum_{\nu=0}^3
    \left(
        \mathbb E_{x'}[b_\nu(x')^4]
    \right)^{1/2} \\
    &\leq
    \sum_{\nu=0}^3\|B_\nu\|_F^2
    =
    \|B\|_F^2,
\end{aligned}
\end{equation}
where the second inequality follows by induction. This proves
\eqref{eq:product-fourth-moment}.

Now write $A=\sum_{P\neq I}a_PP$. For a random product Pauli
eigenstate $\ket{x}$, different Pauli strings are orthogonal under
the averaging over the local basis and eigenvalue signs. More precisely, $\mathbb E_x[\bra{x}P\ket{x}\bra{x}Q\ket{x}]=\delta_{P,Q}3^{-\wt(P)}$. Since $A$ is traceless, $\mathbb E_x\bra{x}A\ket{x}=0$, and hence
\begin{equation}\label{eq:far-second-moments}
\begin{aligned}
    \mathbb E_{x,y}
    \left|
        \bra{y}A\ket{y}-\bra{x}A\ket{x}
    \right|^2
    &=
    2\sum_P3^{-\wt(P)}a_P^2
    \geq
    \frac{2\eps^2}{3^k},\\
    \mathbb E_{x,j}
    \left|
        \bra{x\oplus e_j}A\ket{x\oplus e_j}
        -\bra{x}A\ket{x}
    \right|^2
    &=
    \frac4n\sum_P
        \wt(P)3^{-\wt(P)}a_P^2
    \geq
    \frac{4\eps^2}{n3^k}.
\end{aligned}
\end{equation}
For the second identity, a Pauli term contributes only when
$j\in\supp(P)$; flipping the $j$-th eigenstate then changes its expectation value by a factor of two.

We next bound the corresponding fourth moments. Let
$X:=\bra{x}A\ket{x}$ and $Y:=\bra{y}A\ket{y}$. Since $X,Y$ are independent, identically distributed, and have mean zero, $\mathbb E[(Y-X)^4]=2\mathbb E[X^4]+6(\mathbb E[X^2])^2\leq8\eps^4$, where we used \eqref{eq:product-fourth-moment} and
$\mathbb E[X^2]\leq\|A\|_F^2=\eps^2$.

For the single-qubit gap, define $A^{(j)}:=\sum_{P:j\in\supp(P)}a_PP$. Terms not containing $j$ cancel, while each remaining expectation changes sign. Thus $\bra{x\oplus e_j}A\ket{x\oplus e_j}-\bra{x}A\ket{x}=-2\bra{x}A^{(j)}\ket{x}$. Applying \eqref{eq:product-fourth-moment} gives
\begin{equation}
    \begin{aligned}
    \mathbb E_x
    \left|
        \bra{x\oplus e_j}A\ket{x\oplus e_j}
        -\bra{x}A\ket{x}
    \right|^4
    &\leq
    16\|A^{(j)}\|_F^4.
\end{aligned}
\end{equation}
Averaging over $j$ therefore gives
\begin{equation}\label{eq:far-fourth-moments}
\begin{aligned}
    \mathbb E_{x,j}
    \left|
        \bra{x\oplus e_j}A\ket{x\oplus e_j}
        -\bra{x}A\ket{x}
    \right|^4
    &\leq
    \frac{16}{n}
    \sum_j
    \left(
        \sum_{P:\,j\in\supp(P)}a_P^2
    \right)^2 \leq
    \frac{16k\eps^4}{n}.
\end{aligned}
\end{equation}
Indeed,
$\sum_j\sum_{P:\,j\in\supp(P)}a_P^2
=\sum_P\wt(P)a_P^2\leq k\eps^2$,
while each inner sum is at most $\eps^2$.

Finally, apply Chebyshev's inequality to the squared gaps. For the
first experiment, let
$Z_\ell:=|\bra{y_\ell}A\ket{y_\ell}
-\bra{x_\ell}A\ket{x_\ell}|^2$.
By \eqref{eq:far-second-moments},
$\mathbb E Z_\ell\geq2\eps^2/3^k$, while the fourth-moment bound gives
$\operatorname{Var}(Z_\ell)\leq\mathbb E Z_\ell^2\leq8\eps^4$.
Independence gives $\operatorname{Var}(N^{-1}\sum_\ell Z_\ell)=\operatorname{Var}(Z_1)/N$. Chebyshev's inequality then gives
\begin{equation}
    \Pr\left[
        \frac1N\sum_{\ell=1}^NZ_\ell
        <
        \frac{\eps^2}{3^k}
    \right]
    \leq
    \frac{8\eps^4/N}{(\eps^2/3^k)^2}
    =
    \frac{8\cdot9^k}{N}.
\end{equation}
Similarly, for the single-qubit gaps, \eqref{eq:far-second-moments} and \eqref{eq:far-fourth-moments} give
\begin{equation}
    \Pr\left[
        \frac1N\sum_{\ell=1}^N Z_\ell
        <
        \frac{\eps^2}{n3^k}
    \right]
    \leq\frac{16k\eps^4/(nN)}{(3\eps^2/(n3^k))^2}
    \leq\frac{4kn9^k}{N},
\end{equation}
where $Z_\ell$ now denotes the squared single-qubit gap. The assumptions $N\geq64k9^k$ and $N\geq64kn9^k$ make the two failure probabilities at most $1/8$, respectively, proving the lemma.
\end{proof}

\begin{lemma}[Close case]
\label{lem:close-gaps}
Let $A$ be traceless and $k$-local, with
$\|A\|_F\leq\eps/12^k$ and $k\geq1$.
With the sampling of \Cref{lem:far-gaps}, for every $N\geq1$,
\begin{align}
    \Pr\left[
        \frac1N\sum_{\ell=1}^N
        \left|\bra{y_\ell}A\ket{y_\ell}
              -\bra{x_\ell}A\ket{x_\ell}\right|^2
        \leq\frac{\eps^2}{4\cdot3^k}
    \right]
    &\geq\frac78,\\
    \Pr\left[
        \frac1N\sum_{\ell=1}^N
        \left|\bra{x_\ell\oplus e_{j_\ell}}A
                    \ket{x_\ell\oplus e_{j_\ell}}
              -\bra{x_\ell}A\ket{x_\ell}\right|^2
        \leq\frac{\eps^2}{4n3^k}
    \right]
    &\geq\frac78.
\end{align}
\end{lemma}

\begin{proof}
By the second-moment identities in \eqref{eq:far-second-moments},
\begin{equation}
    \mathbb E_{x,y}
    \left|\bra{y}A\ket{y}-\bra{x}A\ket{x}\right|^2
    \leq \frac23\|A\|_F^2,
    \qquad
    \mathbb E_{x,j}
    \left|\bra{x\oplus e_j}A\ket{x\oplus e_j}
          -\bra{x}A\ket{x}\right|^2
    \leq \frac{4}{3n}\|A\|_F^2,
\end{equation}
where we used $3^{-w}\leq1/3$ and $w3^{-w}\leq1/3$ for $w\geq1$.

The sample averages have the same expectations as the individual squared gaps. Markov's inequality therefore gives
\begin{equation}
    \Pr\left[
        \frac1N\sum_{\ell=1}^N
        \left|\bra{y_\ell}A\ket{y_\ell}
              -\bra{x_\ell}A\ket{x_\ell}\right|^2
        >
        \frac{\eps^2}{4\cdot3^k}
    \right]
    \leq\frac{(2/3)\norm A_{\F}^2}{\eps^2/(4\cdot3^k)}
    \leq\frac83\left(\frac3{144}\right)^k
    \leq\frac18,
\end{equation}
where we used $\|A\|_F\leq\eps/12^k$. Similarly,
\begin{equation}
    \Pr\left[
        \frac1N\sum_{\ell=1}^N
        \left|\bra{x_\ell\oplus e_{j_\ell}}A
                    \ket{x_\ell\oplus e_{j_\ell}}
              -\bra{x_\ell}A\ket{x_\ell}\right|^2
        >
        \frac{\eps^2}{4n3^k}
    \right]
    \leq\frac{(4/(3n))\norm A_{\F}^2}{\eps^2/(4n3^k)}
    \leq\frac{16}{3}\left(\frac3{144}\right)^k
    \leq\frac18.
\end{equation}
Taking complements proves the claim.
\end{proof}

\begin{proof}[Proof of \Cref{thm:cert}]
Fix $\delta\in(0,1/3)$ and let $A:=H-H_0$. If $\eps\geq3\Lambda$, comparing $\norm{H_0}_{\F}$ with $\eps/2$ decides the promise without queries; hence assume $\eps<3\Lambda$.
Take $N:=\lceil64k9^k\rceil$ independent pairs $(x_\ell,y_\ell)$. Estimate their gaps under $H$ and subtract the known $H_0$ gaps, using accuracy $\eta:=\eps/(8\cdot3^{k/2})$ and failure probability $1/(8N)$ for each estimate. If the states are not orthogonal, apply the energy gap estimation subroutine to $\ket0\ket{x_\ell}$ and $\ket1\ket{y_\ell}$ under $I\otimes H$.

For the resulting estimates $\widehat\Delta_\ell$, set $\widehat R:=(N^{-1}\sum_{\ell=1}^N\widehat\Delta_\ell^2)^{1/2}$. When all estimates are accurate, the reverse triangle inequality bounds the change in this quantity by $\eta$. Thus \Cref{lem:far-gaps,lem:close-gaps} and a union bound imply,
with probability at least $3/4$,
\begin{equation}
    \widehat R\geq\frac{7\eps}{8\cdot3^{k/2}}
    \quad\text{in the far case},\qquad
    \widehat R\leq\frac{5\eps}{8\cdot3^{k/2}}
    \quad\text{in the close case}.
\end{equation}
Output \textit{Far} if $\widehat R\geq3\eps/(4\cdot3^{k/2})$ and \textit{Close} otherwise. Independent repetition $\mathcal O(\log(1/\delta))$ times and
majority vote give success probability at least $1-\delta$. Since $k=O(1)$, the total evolution time is
$\widetilde O(N/\eta)=\widetilde O(1/\eps)$.
All gap estimates use the same accuracy and share the step $\tau=\Theta(1/[\Lambda(1+\log(\Lambda/\eta))])
=\widetilde\Theta(1/\Lambda)$.

For $0<\eps\leq\Lambda$, distinguishing $H=0$ from $H=\eps Z_1$ with $H_0=0$ requires $\Omega(1/\eps)$ total evolution time, giving the matching worst-case lower bound.
\end{proof}

\begin{proof}[Proof of \Cref{thm:cert-single}]
Fix $\delta\in(0,1/3)$ and let $A:=H-H_0$.
The case $\eps\geq3\Lambda$ is decided classically as in
\Cref{thm:cert}. Otherwise, take $N:=\lceil64kn9^k\rceil$
independent samples $(x_\ell,j_\ell)$. Use the single-qubit
energy-gap subroutine in \Cref{thm:energy-single-qubit} and subtract the known $H_0$ gaps to estimate
$\bra{x_\ell\oplus e_{j_\ell}}A\ket{x_\ell\oplus e_{j_\ell}}
-\bra{x_\ell}A\ket{x_\ell}$ to accuracy
$\eta:=\eps/(8\sqrt n\,3^{k/2})$, with failure probability $1/(8N)$ for each estimate.

Let $\widehat\Delta_\ell$ be these estimates and set
$\widehat R:=(N^{-1}\sum_{\ell=1}^N\widehat\Delta_\ell^2)^{1/2}$.
By \Cref{lem:far-gaps,lem:close-gaps} and the reverse triangle
inequality, with probability at least $3/4$,
\begin{equation}
    \widehat R\geq\frac{7\eps}{8\sqrt n\,3^{k/2}}
    \quad\text{in the far case},\qquad
    \widehat R\leq\frac{5\eps}{8\sqrt n\,3^{k/2}}
    \quad\text{in the close case}.
\end{equation}
Output \textit{Far} if $\widehat R\geq3\eps/(4\sqrt n\,3^{k/2})$ and \textit{Close} otherwise. Independent repetition $\mathcal O(\log(1/\delta))$ times and majority vote give success probability at least $1-\delta$. The total evolution time is
$\widetilde O(N/\eta)=\widetilde O(n^{3/2}/\eps)$, and all gap estimates share the step
$\tau=\Theta(1/[k\Lambda(1+\log(k\Lambda/\eta))^2])
=\widetilde\Theta(1/\Lambda)$. Each $j_\ell$ is designated as the probe qubit, so all preparations,
controls, and measurements are single-qubit operations and all unknown evolutions are forward.
\end{proof}

For fixed $k$, since every nonzero evolution time is at least the step size $\tau=\widetilde{\Theta}(1/\Lambda)$, the number of queries satisfies $\mathcal N\leq T_{\rm tot}/\tau$. Therefore, \Cref{thm:cert,thm:cert-single} use at most $\widetilde{\mathcal O}(\Lambda/\eps)$ and $\widetilde{\mathcal O}(n^{3/2}\Lambda/\eps)$ queries, respectively. For the classical postprocessing, the dominant cost comes from evaluating the known $H_0$ contribution to each sampled energy gap. Since a $k$-local Hamiltonian has $\mathcal O(n^k)$ Pauli coefficients, each such evaluation takes $\mathcal O(n^k)$ arithmetic operations. The general-control and single-qubit tests use $\widetilde{\mathcal O}(1)$ and $\widetilde{\mathcal O}(n)$ gap samples, respectively. Hence, their total classical postprocessing times are $\widetilde{\mathcal O}(n^k)$ and $\widetilde{\mathcal O}(n^{k+1})$, respectively.

\section{Discussion and future work}
We have established learning and certification algorithms that use only forward Hamiltonian evolution on a common step size. The energy-gap subroutine combines discrete controls with eigenvalue perturbation theory and extrapolation, allowing Heisenberg-limited precision dependence while keeping the step size near the optimal limit. For sparse local Hamiltonians, it leads to coefficient $\ell_2$ recovery in time $\widetilde{\mathcal O}(s/\eps)$. Product controls also suffice for $\ell_\infty$ recovery at this evolution-time scaling, although our $\ell_2$ learning and certification bounds then contain additional factors of $\sqrt n$ and $n^{3/2}$, respectively.

The general energy-gap subroutine and the single-coefficient algorithm do not require locality. The sparse Hamiltonian learning algorithm, however, uses locality to control the random product-state measurements, and the version restricted to single-qubit control also uses locality in its perturbation analysis. Extending these learning results to Hamiltonians without a locality assumption, while preserving the step size and evolution-time guarantees, is therefore a natural direction. Another question is whether the additional system-size factors in the single-qubit algorithms can be reduced. Our learning lower bound leaves a factor $\widetilde{\mathcal O}(s^{1/(2k)})$ in sparsity, so the optimal dependence on $s$ also remains to be determined.

The single-coefficient result gives a direct agnostic algorithm for estimating the local part of an arbitrary Hamiltonian. It would be useful to obtain stronger guarantees when this local approximation is itself sparse, and to determine how geometric locality can improve the resource bounds.

\section*{Statement on the Use of Artificial Intelligence}
The Hamiltonian learning and certification protocol, the idea of applying eigenvalue perturbation analysis to the protocol, and the overall proof framework, including the identification and formulation of the principal theorems and lemmas, were developed entirely by the authors. Generative artificial intelligence tools (ChatGPT~5 and 6) provided the idea of using the Diophantine bound for the single-qubit protocol, and assisted the authors in deriving and refining some of the proofs and in drafting and revising portions of the manuscript. The authors reviewed and verified all AI-assisted arguments and text and take full responsibility for the content of this work. 

\section*{Acknowledgments}
We thank John Preskill and Senrui Chen for helpful discussions. M.S. acknowledges funding from Caltech Summer Undergraduate Research Fellowship, KAIST Presidential Fellowship. Y.T. acknowledges support from the U.S. Department of Energy, Office of Science, Accelerated Research in Quantum Computing Centers, Quantum Utility through Advanced Computational Quantum Algorithms, grant no. DE-SC0025572.

\bibliographystyle{unsrt}
\bibliography{ref}

\newpage

\appendix

\section{Proofs}
In this section, $\widetilde{\mathcal O}$ suppresses polynomial factors in $k$ and logarithmic factors in $n$, $s$, $\Lambda/\eps$, and the inverse failure probability $1/\delta$. Exponential factors in $k$ are displayed explicitly.

\subsection{Proof of \texorpdfstring{\Cref{lem:RIP}}{Lemma \getrefnumber{lem:RIP}}}\label{app:RIP}
\begin{proof}
By \cite[Theorem~2.3]{BrugiapagliaDirksenJungRauhut2021}, the assumed sample bound in \eqref{eq:ortho-sample} ensures, with probability at least
$1-\delta$, that $A$ satisfies
\begin{equation}\label{eq:RIP-used}
    (1-\vartheta)\norm v_2^2\leq\norm{Av}_2^2
    \leq(1+\vartheta)\norm v_2^2,\qquad \vartheta=\frac14,
\end{equation}
for every $3s$-sparse vector $v\in\mathbb C^M$.
Fix such a matrix and write $h:=\widehat z-z$ and
$S:=\operatorname{supp}(z)$. Since $z$ and $\widehat z$ are feasible for \eqref{eq:BPDN}, we have
\begin{equation}\label{eq:measurement-error}
    \norm{Ah}_2\leq\norm{A\widehat z-y}_2+\norm{Az-y}_2\leq2\eta.
\end{equation}
Since $\widehat z$ minimizes the $\ell_1$-norm and $z_{S^c}=0$,
$\norm z_1\geq\norm{z+h}_1\geq
\norm{z_S}_1-\norm{h_S}_1+\norm{h_{S^c}}_1$. Therefore,
\begin{equation}\label{eq:l1-cone}
    \norm{h_{S^c}}_1\leq\norm{h_S}_1.
\end{equation}

Partition $S^c$ into successive groups $S_1,S_2,\ldots$ of $s$
coordinates, except possibly the last, in decreasing order of
$|h_a|$, and set $T:=S\cup S_1$.
Then $|T|\leq2s$ and
$\norm{h_{S_j}}_2\leq\norm{h_{S_{j-1}}}_1/\sqrt s$ for $j\geq2$.
Consequently,
\begin{equation}\label{eq:tail-bound}
    \sum_{j\geq2}\norm{h_{S_j}}_2
    \leq\frac{\norm{h_{S^c}}_1}{\sqrt s}
    \leq\frac{\norm{h_S}_1}{\sqrt s}
    \leq\norm{h_T}_2.
\end{equation}

For vectors $u,v$ with disjoint supports whose union
$R:=\operatorname{supp}(u)\cup\operatorname{supp}(v)$ has size at
most $3s$, \eqref{eq:RIP-used} gives
$\norm{A_R^\dagger A_R-I}_{\op}\leq\vartheta$. Since
$\langle u,v\rangle=0$, this implies
\begin{equation}\label{eq:restricted-inner-product}
    |\langle Au,Av\rangle|
    =|\langle u_R,(A_R^\dagger A_R-I)v_R\rangle|
    \leq\vartheta\norm u_2\norm v_2.
\end{equation}
Apply this to $h_T$ and each $h_{S_j}$, using
$h=h_T+\sum_{j\geq2}h_{S_j}$, to obtain
\begin{equation}\label{eq:head-estimate}
    \norm{Ah_T}_2^2
    \leq\norm{Ah_T}_2\norm{Ah}_2
    +\vartheta\norm{h_T}_2\sum_{j\geq2}\norm{h_{S_j}}_2.
\end{equation}
Combining \eqref{eq:RIP-used}, \eqref{eq:measurement-error}, and
\eqref{eq:tail-bound} gives
$(1-2\vartheta)\norm{h_T}_2^2
\leq2\sqrt{1+\vartheta}\eta\norm{h_T}_2$.
If $h_T=0$, then \eqref{eq:tail-bound} gives $h=0$.
Otherwise, divide by $\norm{h_T}_2$ and substitute
$\vartheta=1/4$ to obtain $\norm{h_T}_2\leq2\sqrt5\eta$.
Finally, using the triangle inequality and \eqref{eq:tail-bound} gives
\begin{equation}
    \norm{\widehat z-z}_2 = \norm{h}_2\leq\norm{h_T}_2
    +\sum_{j\geq2}\norm{h_{S_j}}_2
    \leq2\norm{h_T}_2\leq4\sqrt5\eta.
\end{equation}
Thus setting $C_{\rm cs}=4\sqrt5$ concludes the proof.
\end{proof}

\subsection{Proof of \texorpdfstring{\Cref{lem:analytic}}{Lemma \getrefnumber{lem:analytic}}}\label{app:analytic}
Recall the definition of $V$ in \eqref{eq:V}, $F_+,F_-$ in \eqref{eq:Fpm} and $\ket{e_0}=\ket{0^n},\ket{e_1}=\ket{10^{n-1}}$. This section proves the following lemma (equivalent to \Cref{lem:analytic}).
\begin{lemma}
The matrix $F_+(z)$ has an analytic eigenvector $\ket{\phi_{a,+}(z)}$ and corresponding analytic eigenvalue $\lambda_{a,+}(z)=e^{-i\theta_{a,+}(z)}$ for $a=0,1$. For $|z|<c/\Lambda$, the eigenvectors and eigenvalues satisfy
\begin{equation}
    \norm{\ket{\phi_{a,+}(z)}-\ket{e_a}} \le C\Lambda |z|,\quad |\lambda_{a,+}(z)-\lambda_{a,+}(0)| \le C\Lambda |z|,
\end{equation}
where $c,C$ are universal constants. The same argument holds for $F_{-}(z)$. 
Choose the phases analytically with $\theta_{a,\pm}(0)=\pm a\pi/2$.

With the value at $t=0$ defined by continuity, we have
\begin{equation}\label{eq:app-delta-a}
     \delta_a(t):=\frac{\theta_{a,+}(t)+\theta_{a,-}(t)}{2t} = \bra{e_a}K\ket{e_a}+\sum_{j=1}^\infty b_{a,j}t^{2j},\quad |b_{a,j}| \le C\Lambda(C\Lambda)^{2j}.
\end{equation}
\end{lemma}
\begin{proof}
Set $\nu_0=1$, $\nu_1=-i$, $\theta_0^{(0)}=0$, and
$\theta_1^{(0)}=\pi/2$. Then $V\ket{e_a}=\nu_a\ket{e_a}$ and
$\nu_a=e^{-i\theta_a^{(0)}}$. Each $\nu_a$ is simple and separated from the other eigenvalues of $V$ by at least $\sqrt2$.
Since $V$ is unitary and $\norm K_{\op}\leq\Lambda$,
\begin{equation}\label{eq:analytic-perturbation}
    \norm{F_+(z)-V}_{\op}=\norm{e^{-izK}-I}_{\op}
    \leq e^{\Lambda|z|}-1\leq2\Lambda|z|
\end{equation}
for $\Lambda|z|\leq1$. Choose $c\leq1/24$ sufficiently small.
Then \Cref{lem:quantitative-analytic-perturbation}, applied with
$A(z)=F_+(z)$ and initial eigenpair $(\nu_a,\ket{e_a})$, gives
analytic eigenpairs on $|z|<c/\Lambda$ satisfying
\begin{equation}\label{eq:Fplus-eigenpair-bound}
    \norm{\ket{\phi_{a,+}(z)}-\ket{e_a}}\leq C\Lambda|z|,
    \qquad |\lambda_{a,+}(z)-\nu_a|\leq C\Lambda|z|.
\end{equation}

Since $|\lambda_{a,+}(z)/\nu_a-1|\leq1/2$, we may define
\begin{equation}\label{eq:analytic-phase-definition}
    \theta_{a,+}(z):=\theta_a^{(0)}
    +i\operatorname{Log}\!\left(\frac{\lambda_{a,+}(z)}{\nu_a}\right),
\end{equation}
where $\operatorname{Log}(1)=0$. This gives
\begin{align}
    \lambda_{a,+}(z)&=e^{-i\theta_{a,+}(z)},
        \label{eq:lambda-theta}\\
    |\theta_{a,+}(z)-\theta_a^{(0)}|&\leq C\Lambda|z|.
        \label{eq:analytic-phase-bound}
\end{align}
By \eqref{eq:abstract-eigenvalue-derivative},
$\lambda_{a,+}'(0)=\bra{e_a}F_+'(0)\ket{e_a}
=-i\nu_a\bra{e_a}K\ket{e_a}$.
Differentiating \eqref{eq:lambda-theta} at zero thus yields
\begin{equation}\label{eq:phase-derivative}
    \theta_{a,+}'(0)=\bra{e_a}K\ket{e_a}.
\end{equation}

The same argument applies to $F_-(z)=V^\dagger e^{-izK}$. Since $F_-(z)=V^\dagger F_+(-z)^{-1}V$, the analytic eigenvalues satisfy $\lambda_{a,-}(z)=\lambda_{a,+}(-z)^{-1}$ and
$\theta_{a,-}(z)=-\theta_{a,+}(-z)$. Consequently,
\begin{equation}\label{eq:delta-even}
    \delta_a(t)=\frac{\theta_{a,+}(t)+\theta_{a,-}(t)}{2t}
    =\frac{\theta_{a,+}(t)-\theta_{a,+}(-t)}{2t}.
\end{equation}
The numerator vanishes at zero, so $\delta_a$ extends analytically
to $t=0$ and is even. For real $t$, unitarity of $F_+(t)$ also
makes $\theta_{a,+}(t)$ real. Writing
$\theta_{a,+}(z)=\theta_a^{(0)}+\sum_{\ell\geq1}c_{a,\ell}z^\ell$
in \eqref{eq:delta-even} and applying \eqref{eq:phase-derivative} gives $b_{a,0}=\bra{e_a}K\ket{e_a}$ and $b_{a,j}=c_{a,2j+1}$, proving the expansion in \eqref{eq:app-delta-a}.

It remains to bound $|b_{a,j}|$. Choose $R:=c/(2\Lambda)$. By \eqref{eq:analytic-phase-bound}, we have $ \sup_{|z|=R} |\theta_{a,+}(z)-\theta_a^{(0)}| \leq C\Lambda R$. The Cauchy integral formula gives
\begin{equation}
    c_{a,\ell}=
    \frac{1}{2\pi i}
    \oint_{|z|=R}
    \frac{
        \theta_{a,+}(z)-\theta_a^{(0)}
    }{
        z^{\ell+1}
    }
    \,dz.
\end{equation}
Therefore, $|c_{a,\ell}|\leq C\Lambda R^{1-\ell}\leq C\Lambda(C\Lambda)^{\ell-1}$. Finally, $|b_{a,j}|=|c_{a,2j+1}|\leq C\Lambda(C\Lambda)^{2j}$ as required.
\end{proof}

\subsection{Proof of \texorpdfstring{\Cref{lem:estimate-fL}}{Lemma \getrefnumber{lem:estimate-fL}}}\label{app:estimate-fL}
Throughout this subsection, let
\begin{equation}
    d:=\frac{1}{2kL},
    \qquad
    \Delta:=kL\Lambda.
    \label{eq:estimate-fL-scales}
\end{equation}
For \(a\in\{0,1\}\) and \(R\geq0\), define the Hamming-ball
subspace
\begin{equation}
    \mathcal B_{a,R}
    :=
    \operatorname{span}
    \left\{
        \ket{u}:
        \operatorname{Ham}(u,e_a)\leq R
    \right\}.
    \label{eq:Hamming-ball-subspace}
\end{equation}

Since $W$ is a product of single-qubit unitaries and $H$ is $k$-local, the conjugated Hamiltonian $K=W^\dagger HW$ is also $k$-local.  Consequently,
\begin{equation}
    K^p\mathcal B_{a,R}
    \subseteq
    \mathcal B_{a,R+kp}
    \label{eq:K-Hamming-propagation}
\end{equation}
for every integer \(p\geq0\). Indeed, each local term in \(K\) acts nontrivially on at most \(k\) qubits, so a product of \(p\) such terms can change at most \(kp\) additional computational-basis bits. We first establish the locality of the Taylor coefficients of the auxiliary eigenvectors.

\begin{lemma}[Taylor locality of the auxiliary branches]
\label{lem:single-taylor-locality}
Write the expansion of $\ket{\widehat\phi_{a,s}(z)}$, $s\in\{+,-\}$, in \Cref{lem:analytic-single} as
\begin{equation}
    \ket{\widehat\phi_{a,s}(z)}
    =
    \sum_{r=0}^{\infty}
    z^r\ket{\phi_{a,r,s}},\quad e^{-izK}\ket{\widehat\phi_{a,s}(z)}=\sum_{r=0}^{\infty}
    z^r\ket{\chi_{a,r,s}}.
    \label{eq:auxiliary-eigenvector-expansion}
\end{equation}
For \(a\in\{0,1\}\), \(s\in\{+,-\}\) and every \(r\geq0\), 
\begin{equation}
    \ket{\phi_{a,r,s}}\in\mathcal B_{a,kr},\quad \ket{\chi_{a,r,s}}\in\mathcal B_{a,kr}.
    \label{eq:eigenvector-Taylor-locality}
\end{equation}
\end{lemma}

\begin{proof}
We prove the claim for $s=+$; the other sign is identical.
Write $\widehat\lambda_{a,+}(z)=\sum_{r\geq0}\lambda_{a,r,+}z^r$,
where $\lambda_{a,0,+}=\nu_{e_a}$ and
$\ket{\phi_{a,0,+}}=\ket{e_a}$. The eigenvalue equation is
\begin{equation}\label{eq:auxiliary-eigenvalue-equation}
    \widehat V_{a,+}e^{-izK}\ket{\widehat\phi_{a,+}(z)}
    =\widehat\lambda_{a,+}(z)\ket{\widehat\phi_{a,+}(z)}.
\end{equation}
Comparing the coefficient of \(z^r\) in
\eqref{eq:auxiliary-eigenvalue-equation} gives
\begin{equation}
    \left(\widehat V_{a,+}-\nu_{e_a}I\right)
    \ket{\phi_{a,r,+}}=
    \sum_{q=1}^{r}
    \lambda_{a,q,+}
    \ket{\phi_{a,r-q,+}}-\widehat V_{a,+}\sum_{p=1}^{r}\frac{(-iK)^p}{p!}\ket{\phi_{a,r-p,+}}.
    \label{eq:auxiliary-Taylor-recursion}
\end{equation}

We prove $\ket{\phi_{a,r,+}}\in\mathcal B_{a,kr}$ by induction on \(r\).  It holds for \(r=0\).  Suppose that it holds for all orders
smaller than \(r\).  Then, for every \(q\geq1\), we have $\ket{\phi_{a,r-q,+}}\in\mathcal B_{a,k(r-q)}\subseteq\mathcal B_{a,kr}$. Also, by \eqref{eq:K-Hamming-propagation}, $K^p\ket{\phi_{a,r-p,+}}\in\mathcal B_{a,k(r-p)+kp}=\mathcal B_{a,kr}$. Since \(\widehat V_{a,+}\) is diagonal in the computational basis,
it preserves each Hamming-ball subspace.  Hence the entire right-hand side of \eqref{eq:auxiliary-Taylor-recursion} belongs to
\(\mathcal B_{a,kr}\). By \Cref{lem:single-local-gap}, \(\nu_{e_a}\) is a nondegenerate
eigenvalue of \(\widehat V_{a,+}\).  Thus
\(\widehat V_{a,+}-\nu_{e_a}I\) has a nonzero diagonal entry on every computational-basis vector other than \(\ket{e_a}\). Taking the \(\ket{u}\)-coordinate of \eqref{eq:auxiliary-Taylor-recursion} for \(\operatorname{Ham}(u,e_a)>kr\) therefore gives $\braket{u}{\phi_{a,r,+}}=0$. This proves $\ket{\phi_{a,r,+}}\in\mathcal B_{a,kr}$.

Finally, expanding the left-hand side of
\begin{equation}
    e^{-izK}\ket{\widehat\phi_{a,s}(z)}=\sum_{r=0}^{\infty}
    z^r\ket{\chi_{a,r,s}}
\end{equation}
and comparing the degree-$r$ coefficient gives
\begin{equation}
    \ket{\chi_{a,r,+}}
    =
    \sum_{p=0}^{r}
    \frac{(-iK)^p}{p!}
    \ket{\phi_{a,r-p,+}}.
    \label{eq:chi-coefficient}
\end{equation}
Every term in \eqref{eq:chi-coefficient} belongs to
\(\mathcal B_{a,kr}\) by \eqref{eq:K-Hamming-propagation} and \eqref{eq:eigenvector-Taylor-locality}. This proves
$\ket{\chi_{a,r,s}}\in\mathcal B_{a,kr}$.
\end{proof}

We next compare the physical and auxiliary evolutions.  For
notational convenience, set
\begin{equation}
    V_+:=V,
    \qquad
    V_-:=V^\dagger,
    \qquad
    F_s(z):=V_se^{-izK},
    \quad
    s\in\{+,-\}.
\end{equation}

\begin{lemma}
\label{lem:single-physical-defect}
There exist universal constants \(c,C>0\) such that, for every
\(a\in\{0,1\}\), \(s\in\{+,-\}\), and real \(t\) satisfying $|t|\leq c/\Delta$, we have
\begin{align}
    &\norm{
        F_s(t)\ket{\widehat\phi_{a,s}(t)}
        -
        \widehat\lambda_{a,s}(t)
        \ket{\widehat\phi_{a,s}(t)}
    }\leq
    C(C\Delta|t|)^{L+1}.
    \label{eq:physical-quasimode-defect}
\end{align}
Here \(F_s(t)\) is the physical controlled evolution defined in \eqref{eq:Fpm}, with \(V\) given by \eqref{eq:single-product-control}, while \(\ket{\widehat\phi_{a,s}(t)}\) and \(\widehat\lambda_{a,s}(t)\) form the auxiliary eigenpair introduced in \Cref{lem:analytic-single}.
\end{lemma}

\begin{proof}
Define the vector-valued analytic function
\begin{align}
    \ket{\operatorname{err}_{a,s}(z)}
    &:=
    F_s(z)\ket{\widehat\phi_{a,s}(z)}
    -
    \widehat\lambda_{a,s}(z)
    \ket{\widehat\phi_{a,s}(z)}.
    \label{eq:def-physical-error} \\
    &=\left(V_s-\widehat V_{a,s}\right)
    e^{-izK}\ket{\widehat\phi_{a,s}(z)}.
    \label{eq:physical-error-difference}
\end{align}
The second equation \eqref{eq:physical-error-difference} holds because \(\ket{\widehat\phi_{a,s}(z)}\) is an exact eigenvector of the auxiliary matrix \(\widehat{\mathcal F}_{a,s}(z)
=\widehat V_{a,s}e^{-izK}\) for $s\in\{+,-\}$. By \Cref{lem:single-taylor-locality}, the coefficient of \(z^r\) in $e^{-izK}\ket{\widehat\phi_{a,s}(z)}$ belongs to \(\mathcal B_{a,kr}\).  For every \(0\leq r\leq L\), $\mathcal B_{a,kr}\subseteq\mathcal B_{a,kL}$. By the definition of $\widehat V_{a,s}$ in \eqref{eq:single-auxiliary-V},
\begin{equation}
    V_s\ket{u}
    =
    \widehat V_{a,s}\ket{u}
    \qquad
    \text{whenever }
    \operatorname{Ham}(u,e_a)\leq kL.
    \label{eq:controls-agree-on-ball}
\end{equation}
It follows from
\eqref{eq:physical-error-difference} that the Taylor coefficients of
\(\ket{\operatorname{err}_{a,s}(z)}\) of orders
\(0,1,\ldots,L\) all vanish.  Thus
\begin{equation}
    \ket{\operatorname{err}_{a,s}(z)}
    =
    \sum_{r=L+1}^{\infty}
    z^r\ket{\operatorname{err}_{a,r,s}}.
    \label{eq:error-Taylor-start}
\end{equation}

By \Cref{lem:analytic-single}, $\ket{\widehat\phi_{a,s}(z)}$ is analytic on a disk of radius \(c_0/\Delta\), for sufficiently small $c_0$.  Set $R:=c_0/(2\Delta)$. On the circle \(|z|=R\), the eigenvector perturbation bound in \eqref{eq:perturbation} gives
\begin{equation}\label{eq:phi-bound}
    \norm{\ket{\widehat\phi_{a,s}(z)}} \le 1+\norm{\ket{\widehat\phi_{a,s}(z)}-\ket{e_a}}
    \leq
    1+C\Delta |z| \le 1+C\Delta R =1+\frac{c_0C}{2}.
\end{equation}
Moreover,
\begin{equation}\label{eq:K-V-bound}
    \norm{e^{-izK}}_{\op}
    \leq
    e^{\Lambda|z|}
    \leq
    e^{c_0/(2kL)}, \quad
    \norm{
        V_s-\widehat V_{a,s}
    }_{\op}
    \leq2
\end{equation}
Hence, \eqref{eq:phi-bound} and \eqref{eq:K-V-bound} are bounded by a constant. Applying the bounds to \eqref{eq:physical-error-difference} gives 
\begin{equation}
    \sup_{|z|=R}
    \norm{
        \ket{\operatorname{err}_{a,s}(z)}
    }
    \leq C_1,
    \label{eq:error-bound-on-circle}
\end{equation}
for some universal constant $C_1$. The Cauchy estimate for the Taylor coefficients of the vector-valued analytic function in
\eqref{eq:error-Taylor-start} gives $\norm{\ket{\operatorname{err}_{a,r,s}}}\leq C_1R^{-r}$. Therefore, whenever \(|t|\leq R/2\),
\begin{equation}
    \norm{\ket{\operatorname{err}_{a,s}(t)}}\leq
    C_1\sum_{r=L+1}^{\infty} \left(\frac{|t|}{R}\right)^r \leq 2C_1\left(\frac{|t|}{R}\right)^{L+1}\leq
    2C_1(\frac{2}{c_0}\Delta|t|)^{L+1}.
\end{equation}
Writing the universal constant \(C = \max\{2C_1,2/c_0\}\) proves \eqref{eq:physical-quasimode-defect}.
\end{proof}

We now analyze repeated applications of $F_s(t)$ to $\ket{e_a}$.

\begin{lemma}
\label{lem:single-repeated-branch}
Under the assumptions of
\Cref{lem:single-physical-defect}, for every integer \(m\geq0\),
\begin{align}
    \norm{
        F_s(t)^m\ket{e_a}
        -
        e^{-im\widehat\theta_{a,s}(t)}
        \ket{e_a}
    }\leq
    C\Delta|t|
    +
    Cm(C\Delta|t|)^{L+1}.
    \label{eq:single-repeated-branch-bound}
\end{align}
Here \(F_s(t)\) is the physical controlled evolution defined in \eqref{eq:Fpm}, with \(V\) given by \eqref{eq:single-product-control}, while \(\ket{\widehat\phi_{a,s}(t)}\) and \(\widehat\lambda_{a,s}(t)=e^{-i\widehat\theta_{a,s}(t)}\) form the auxiliary eigenpair introduced in \Cref{lem:analytic-single}.
\end{lemma}

\begin{proof}
Write $\ket{\widehat\phi}:=\ket{\widehat\phi_{a,s}(t)}$,
$\widehat\lambda:=\widehat\lambda_{a,s}(t)$, and
$\ket\varepsilon:=F_s(t)\ket{\widehat\phi}
-\widehat\lambda\ket{\widehat\phi}$. Telescoping gives
\begin{equation}\label{eq:branch-telescoping}
    F_s(t)^m\ket{\widehat\phi}-\widehat\lambda^m\ket{\widehat\phi}
    =\sum_{q=0}^{m-1}\widehat\lambda^{m-1-q}F_s(t)^q\ket\varepsilon.
\end{equation}
For real \(t\), both \(F_s(t)\) and
\(\widehat{\mathcal F}_{a,s}(t)\) are unitary.  Hence $|\widehat\lambda|=1$ and
\begin{equation}
    \norm{
        F_s(t)^m\ket{\widehat\phi}
        -
        \widehat\lambda^m\ket{\widehat\phi}
    }
    \leq
    m\norm{\ket{\varepsilon}}.
    \label{eq:telescoping-norm-bound}
\end{equation}
Using
\(\widehat\lambda=e^{-i\widehat\theta_{a,s}(t)}\), we obtain
\begin{align}
    &\norm{
        F_s(t)^m\ket{e_a}
        -
        \widehat\lambda^m\ket{e_a}
    }
    \nonumber\\
    &\quad\leq
    \norm{
        F_s(t)^m
        \left(
            \ket{e_a}-\ket{\widehat\phi}
        \right)
    }+
    \norm{
        F_s(t)^m\ket{\widehat\phi}
        -
        \widehat\lambda^m\ket{\widehat\phi}
    }
    +
    \norm{
        \widehat\lambda^m
        \left(
            \ket{\widehat\phi}-\ket{e_a}
        \right)
    }
    \nonumber\\
    &\quad\leq
    2\norm{
        \ket{\widehat\phi}-\ket{e_a}
    }
    +
    m\norm{\ket{\varepsilon}}.
    \label{eq:eigenvector-to-basis-vector}
\end{align}
Substituting $\norm{\ket{\widehat\phi}-\ket{e_a}}\leq C\Delta|t|$ by \Cref{lem:analytic-single} and $\norm{\ket{\varepsilon}}\leq C(C\Delta|t|)^{L+1}$ by \Cref{lem:single-physical-defect} into \eqref{eq:eigenvector-to-basis-vector} proves \eqref{eq:single-repeated-branch-bound}.
\end{proof}

Now we prove \Cref{lem:estimate-fL}, restated below.
\begin{lemma}\label{lem:app-estimate-fL}
    Let $\Delta_L=kL\Lambda$. There exist universal constants $c,C$ such that, for $0<t\le c/\Delta_L$ and every integer $m\geq0$, we have
    \begin{equation}\label{eq:app-circuit-approx-L}
        \norm{F_{-}(t)^mF_{+}(t)^m\ket{+}\ket{0^{n-1}} - \ket{\psi_L}} \le Ct\Delta_L + Cm(Ct\Delta_L)^{L+1},
    \end{equation}
    where $\ket{\psi_L}$ is
    \begin{equation}
        \ket{\psi_L}=\frac{1}{\sqrt2}(e^{-im(\widehat\theta_{0,+}(t)+\widehat\theta_{0,-}(t))}\ket{e_0}+e^{-im(\widehat\theta_{1,+}(t)+\widehat\theta_{1,-}(t))}\ket{e_1}).
    \end{equation}
    Here \(F_s(t)\) is the physical controlled evolution defined in \eqref{eq:Fpm}, with \(V\) given by \eqref{eq:single-product-control}, while \(\ket{\widehat\phi_{a,s}(t)}\) and \(\widehat\lambda_{a,s}(t)=e^{-i\widehat\theta_{a,s}(t)}\) form the auxiliary eigenpair introduced in \Cref{lem:analytic-single}.
\end{lemma}
\begin{proof}
Define
\begin{equation}
    q_L(t,m)
    :=
    C\Delta t
    +
    Cm(C\Delta t)^{L+1}.
    \label{eq:def-qL}
\end{equation}
By \Cref{lem:single-repeated-branch}, for every
\(a\in\{0,1\}\),
\begin{align}
    \norm{
        F_+(t)^m\ket{e_a}
        -
        e^{-im\widehat\theta_{a,+}(t)}
        \ket{e_a}
    }
    &\leq
    q_L(t,m),
    \label{eq:plus-repeated-bound}\\
    \norm{
        F_-(t)^m\ket{e_a}
        -
        e^{-im\widehat\theta_{a,-}(t)}
        \ket{e_a}
    }
    &\leq
    q_L(t,m).
    \label{eq:minus-repeated-bound}
\end{align}
Using the unitarity of \(F_-(t)^m\), we have
\begin{align}
    &\norm{
        F_-(t)^mF_+(t)^m\ket{e_a}
        -
        e^{-im(
            \widehat\theta_{a,+}(t)
            +
            \widehat\theta_{a,-}(t)
        )}
        \ket{e_a}
    }
    \nonumber\\
    &\quad\leq
    \norm{
        F_-(t)^m
        \left[
            F_+(t)^m\ket{e_a}
            -
            e^{-im\widehat\theta_{a,+}(t)}
            \ket{e_a}
        \right]
    }
    \nonumber\\
    &\qquad+
    \norm{
        e^{-im\widehat\theta_{a,+}(t)}
        \left[
            F_-(t)^m\ket{e_a}
            -
            e^{-im\widehat\theta_{a,-}(t)}
            \ket{e_a}
        \right]
    }
    \nonumber\\
    &\quad\leq
    2q_L(t,m).
    \label{eq:composite-branch-bound}
\end{align}
Finally, using $\ket{+}\ket{0^{n-1}}=\frac{\ket{e_0}+\ket{e_1}}{\sqrt2}$ and applying \eqref{eq:composite-branch-bound} to \(a=0,1\) with the triangle inequality yields
\begin{align}
    &\norm{
        F_-(t)^mF_+(t)^m
        \ket{+}\ket{0^{n-1}}
        -
        \ket{\psi_L}
    }\leq
    \frac{1}{\sqrt2}
    \sum_{a=0}^{1}
    2q_L(t,m)\leq
    C\Delta t
    +
    Cm(C\Delta t)^{L+1}.
\end{align}
This proves the lemma.
\end{proof}

\subsection{Proof of \texorpdfstring{\Cref{thm:near-upper}}{Theorem \getrefnumber{thm:near-upper}}}\label{app:near-upper}
\begin{theorem}\label{app-thm:near-upper}
Let $H=\sum_{P\in\Pnk}h_PP$ be traceless, $s$-sparse, $k$-local, and satisfy $\norm{H}_{\op}\leq\Lambda$. For every $\eps\in(0,\Lambda)$ and $\delta\in(0,1/3)$, one can output $\widehat h$ such that
\begin{equation}
\norm{\widehat h-h}_2\leq \eps
\end{equation}
with probability at least $1-\delta$. The total evolution time is 
\begin{equation}
    \widetilde{\mathcal{O}}(\frac{3^{3k/2}s}{\eps}\log\frac{n}{\delta})    
\end{equation}
and the step size is 
\begin{equation}
    \Theta(\frac{1}{\Lambda
    (k+\log(\Lambda/\eps))}).
\end{equation}
\end{theorem}

Choose independently and uniformly $\sigma\in\{X,Y,Z\}^n$ and $x,r\in\{0,1\}^n$. Define the gap function as
\begin{equation}\label{eq:app-gap-def}
G(\sigma,x,r)
:=\frac12\left(
\bra{x}_{\sigma}H\ket{x}_{\sigma}-\bra{x\oplus r}_{\sigma}H\ket{x\oplus r}_{\sigma}
\right).
\end{equation}
We can derive that, for every $\sigma,x,r$,
\begin{equation}\label{eq:app-linear-measurement}
G(\sigma,x,r)=\sum_{P\in\Pnk}z_P\phi_P(\sigma,x,r),\qquad
z_P:=\frac{h_P}{\sqrt{2}\,3^{w_P/2}},
\end{equation}
where
\begin{equation}
    \phi_P(\sigma,x,r)
:=
\sqrt{2}\,3^{w_P/2}
\mathbf 1\{P\preceq\sigma\}
(-1)^{x\cdot S_P}
\mathbf 1\{r\cdot S_P=1\}.
\end{equation}
Recall the definitions of $P\preceq\sigma$, $S_P$ in \Cref{subsec:notation}.

\begin{lemma}[Bounded orthonormal system]\label{lem:bos}
The functions $\{\phi_P:P\in\Pnk\}$ satisfy
\begin{equation}\label{eq:bos}
\E[\phi_P\phi_Q]=\delta_{P,Q},
\qquad
\norm{\phi_P}_{\infty}\leq \sqrt{2}\,3^{k/2}.
\end{equation}
\end{lemma}

\begin{proof}
We first prove orthonormality.  Suppose $P\neq Q$. If $S_P\neq S_Q$, then
\begin{equation}
    \E_x
    \left[
        (-1)^{x\cdot S_P}
        (-1)^{x\cdot S_Q}
    \right]
    =
    \E_x
    \left[
        (-1)^{x\cdot(S_P\triangle S_Q)}
    \right]
    =
    0,
\end{equation}
because $S_P\triangle S_Q$ is nonempty.  Since $x$ is independent of $\sigma$ and $r$, it follows that $\E[\phi_P\phi_Q]=0$.  If $S_P=S_Q$ but $P\neq Q$, then there exists $i\in S_P$ such that $P_i\neq Q_i$. Hence no Pauli basis $\sigma$ can satisfy both $P\preceq\sigma$ and $Q\preceq\sigma$. Therefore, $\mathbf 1\{P\preceq\sigma\}\mathbf 1\{Q\preceq\sigma\}=0$ for every $\sigma$, and again $\E[\phi_P\phi_Q]=0$. It remains to compute $\E[\phi_P^2]$.  Since every $P\in\Pnk$ is nonidentity, $S_P$ is nonempty. Thus $\Pr_{\sigma}[P\preceq\sigma]=3^{-w_P}$ and $\Pr_r[r\cdot S_P=1]=\frac12$. Consequently,
\begin{align}
    \E[\phi_P^2]
    &=
    2\,3^{w_P}
    \Pr_{\sigma}[P\preceq\sigma]
    \Pr_r[r\cdot S_P=1] =
    2\,3^{w_P}\cdot 3^{-w_P}\cdot\frac12 = 1.
\end{align}
Combining the three cases gives $\E[\phi_P\phi_Q]=\delta_{P,Q}$.

Finally, each indicator has absolute value at most one, and the sign $(-1)^{x\cdot S_P}$ has absolute value one.  Hence
\begin{equation}
    |\phi_P(\sigma,x,r)|
    \leq
    \sqrt{2}\,3^{w_P/2}
    \leq
    \sqrt{2}\,3^{k/2},
\end{equation}
because $w_P\leq k$.  This proves the claim.
\end{proof}

\begin{proof}[Proof of Theorem~\ref{app-thm:near-upper}]
Set
\begin{equation}\label{eq:M-K}
M:=|\Pnk|=\sum_{j=1}^{k}3^j\binom nj,
\qquad
K:=\sqrt{2}\,3^{k/2}.
\end{equation}
To apply \Cref{lem:RIP}, draw
\begin{equation}\label{eq:app-sample}
    m:=
    \left\lceil
        C K^2s(1+\log(sK^2))^2
        \log\left(\frac{M}{\delta}\right)
    \right\rceil
\end{equation}
independent triples $\omega_\ell=(\sigma_\ell,x_\ell,r_\ell)$, with the universal constant $C$ large enough to apply \Cref{lem:RIP} with failure probability $\delta/2$. For each sample, estimate $g_\ell=G(\sigma_\ell,x_\ell,r_\ell)$ using \Cref{thm:energy-gap-subroutine}, dividing the energy-gap estimate by two, obtaining $\widehat g_\ell$ with additive error at most $\eta$ and failure probability at most $\delta/(2m)$. Set $A_{\ell,P}=m^{-1/2}\phi_P(\omega_\ell)$ and $\widehat y_\ell=m^{-1/2}\widehat g_\ell$. A union bound makes all gap estimates accurate with probability at least $1-\delta/2$. Hence, with this probability, we have
\begin{equation}\label{eq:noise}
\widehat y=Az+e,
\qquad
\norm{e}_2
\leq
\left(\frac1m\sum_{\ell=1}^{m}\eta^2\right)^{1/2}
=\eta.
\end{equation}

By \Cref{lem:bos}, the functions $\phi_P$ satisfy the assumptions of \Cref{lem:RIP}. Hence, the solution of
\begin{equation}\label{eq:bpdn}
\widehat z\in
\operatorname*{arg\,min}_{v\in\mathbb R^M}\norm{v}_1
\quad\text{subject to}\quad
\norm{Av-\widehat y}_2\leq\eta.
\end{equation}
gives $\norm{\widehat z-z}_2\leq C_{\rm cs}\eta$ by \Cref{lem:RIP}. Defining $\widehat h_P:=\sqrt2\,3^{w_P/2}\widehat z_P$, we obtain
\begin{equation}\label{eq:h-error}
\norm{\widehat h-h}_2
\leq
\sqrt{2}\,3^{k/2}\norm{\widehat z-z}_2
\leq
\sqrt{2}\,3^{k/2}C_{\rm cs}\eta.
\end{equation}
Taking $\eta:=\eps/(\sqrt{2}C_{\rm cs}3^{k/2})$ proves the error guarantee, and the union bound gives success probability at least $1-\delta$.

Estimating each gap $g_{\ell}$ with \Cref{thm:energy-gap-subroutine} costs $\widetilde O(1/\eta)=\widetilde O(3^{k/2}/\eps)$ evolution time. There are $m$ samples, hence
\begin{equation}\label{eq:app-total-time}
T_{\rm tot}
=
\mathcal{\widetilde{O}}\left(\frac{m3^{k/2}}{\eps}\right)=\mathcal{\widetilde{O}}\left(\frac{3^{3k/2}s}{\eps}\log\frac{n}{\delta}\right).
\end{equation}
All gap estimates use the same accuracy $\eta$, so the step size is
\begin{equation}\label{eq:step-size}
\tau=
\Theta\left(
\frac{1}{\Lambda(1+\log(\Lambda/\eta))}
\right)
=
\Theta\left(
\frac{1}{\Lambda(1+\log(3^{k/2}\Lambda/\eps))}
\right)
=
\Theta\left(
\frac{1}{\Lambda(k+\log(\Lambda/\eps))}
\right),
\end{equation}
which is deduced from \eqref{eq:energy-gap-step-size}.
\end{proof}

Finally, we explain how to construct the unitary $W$ employed in the energy gap subroutine. If $r=0$, then $G(\sigma,x,r)=0$ exactly. We record this value without employing the subroutine and retain the corresponding row in the sampling matrix. For $r\neq0$, choose an invertible binary matrix $A_r$ with $A_re_1=r$, and let $C_r$ be the corresponding CNOT/SWAP circuit. Let $B_{\sigma,x}=\bigotimes_i B_i$ satisfy $B_{\sigma,x}\ket u=\ket{x\oplus u}_\sigma$. Then
\begin{equation}\label{eq:W}
W_{\sigma,x,r}:=B_{\sigma,x}C_r,
\qquad
W_{\sigma,x,r}\ket{0^n}=\ket{x}_\sigma,
\qquad
W_{\sigma,x,r}\ket{e_1}=\ket{x\oplus r}_\sigma.
\end{equation}
Thus every random gap has the exact form required by the energy gap estimation subroutine.

\subsection{Proof of \texorpdfstring{\Cref{thm:single-qubit-learning}}{Theorem \getrefnumber{thm:single-qubit-learning}}}\label{app:single-qubit-learning}
\begin{theorem}\label{app-thm:single-qubit-learning}
Let $H=\sum_{P\in\Pnk}h_PP$ be traceless, $s$-sparse, $k$-local, and satisfy $\norm{H}_{\op}\leq\Lambda$. For every $\eps\in(0,\Lambda)$ and $\delta\in(0,1/3)$, for a universal constant $c>1$, one can output $\widehat h$ such that $\norm{\widehat h-h}_\infty\leq\eps$, using
\begin{equation}\label{eq:tot-single-infty}
    \widetilde{\mathcal{O}}(\frac{c^{k}s}{\eps}\log\frac{n}{\delta})  
\end{equation}
total evolution time with probability at least $1-\delta$. The algorithm can also output 
$\widehat h$ such that $\norm{\widehat h-h}_2\leq \eps$, using 
\begin{equation}\label{eq:tot-single-2}
    \widetilde{\mathcal{O}}(\frac{c^{k}s\sqrt{n}}{\eps}\log\frac{n}{\delta})  
\end{equation}
total evolution time with probability at least $1-\delta$. The step sizes are 
\begin{equation}
    \tau_\infty=\Theta\left(\frac{1}{k\Lambda(k+\log(\Lambda/\eps))^2}\right),
    \qquad
    \tau_2=\Theta\left(\frac{1}{k\Lambda(k+\log(\Lambda\sqrt n/\eps))^2}\right),
\end{equation}
respectively. The algorithm only uses single-qubit gates and measurements with forward evolution $e^{-iH\tau}$.
\end{theorem}

Choose independently and uniformly
$\sigma\in\{X,Y,Z\}^n$ and $x\in\{0,1\}^n$.
For every $j\in[n]$, define
\begin{equation}\label{eq:app-single-qubit-gap}
G_j(\sigma,x)
:=
\frac12\left(
\bra{x}_{\sigma}H\ket{x}_{\sigma}
-
\bra{x\oplus e_j}_{\sigma}
H
\ket{x\oplus e_j}_{\sigma}
\right).
\end{equation}
The two states $\ket{x}_\sigma, \ket{x\oplus e_j}_\sigma$ in \eqref{eq:app-single-qubit-gap} differ only on qubit $j$.
Hence, there is a product unitary $W$ satisfying $W\ket{v}=\ket{x\oplus v}_\sigma$. We apply \Cref{thm:energy-single-qubit} with qubit $j$ as the
probe, replacing the control by
$V_j=\exp[-i\gamma(N_j+2\sum_{\ell\neq j}N_\ell)]$
and measuring qubit $j$. This only relabels the qubits in the proof and requires no physical SWAP gate. Therefore, we only use single-qubit control.

Define 
\begin{equation}
    h_P^{(j)}
    :=
    \begin{cases}
        h_P,
        & j\in\operatorname{supp}(P),\\
        0,
        & j\notin\operatorname{supp}(P),
    \end{cases}
    \qquad
    d_j:=\norm{h^{(j)}}_0,
\end{equation}
where $d_j$ is the number of nonzero entries of $h^{(j)}$. A direct calculation gives
\begin{equation}
    G_j(\sigma,x)
    =
    \sum_{P\in\Pnk}
    z_P^{(j)}\phi_{j,P}(\sigma,x),
    \qquad
    z_P^{(j)}:=3^{-w_P/2}h_P^{(j)},
\end{equation}
where
\begin{equation}\label{eq:app-single-qubit-functions}
    \phi_{j,P}(\sigma,x)
    :=
    3^{w_P/2}
    \mathbf 1\{P\preceq\sigma\}
    (-1)^{\sum_{i\in\supp(P)}x_i}.
\end{equation}
Following the proof of \Cref{lem:bos}, for every fixed $j$,
the functions
$\{\phi_{j,P}:P\in\Pnk,\ j\in\supp(P)\}$ satisfy
\begin{equation}\label{eq:app-single-qubit-orthonormality}
    \E_{\sigma,x}
    [\phi_{j,P}(\sigma,x)\phi_{j,Q}(\sigma,x)]
    =
    \delta_{P,Q},
    \qquad
    \norm{\phi_{j,P}}_\infty
    \leq
    3^{k/2}.
\end{equation}

We now apply the same sparse-recovery procedure as in
\Cref{app:near-upper}, restricted to Pauli strings containing $j$.
Compared with \eqref{eq:M-K}, take $K=3^{k/2}$ and replace $s$ by $\max\{1,d_j\}$. Thus, the analogue of \eqref{eq:app-sample}
requires
\begin{equation}
    m_j
    =
    \widetilde{\mathcal O}\left(
        3^k\max\{1,d_j\}\log\frac{n}{\delta}
    \right)
\end{equation}
samples.

For these samples, we form the normalized observations exactly as
in \eqref{eq:noise} and solve the same optimization problem as in
\eqref{eq:bpdn}, using only the columns with $j\in\supp(P)$ and setting all other entries to zero. By \Cref{lem:RIP}, the recovered vector satisfies
$\norm{\widehat z^{(j)}-z^{(j)}}_2\leq C_{\rm cs}\eta$.
Since $h_P^{(j)}=3^{w_P/2}z_P^{(j)}$, the analogue of
\eqref{eq:h-error} is
\begin{equation}
    \norm{\widehat h^{(j)}-h^{(j)}}_2
    \leq
    3^{k/2}C_{\rm cs}\eta.
\end{equation}
Therefore, choosing
$\eta:=\eps/(C_{\rm cs}3^{k/2})$ gives
\begin{equation}\label{eq:h-estimate-l2}
    \norm{h^{(j)}-\widehat h^{(j)}}_\infty
    \leq
    \norm{h^{(j)}-\widehat h^{(j)}}_2
    \leq
    \eps.
\end{equation}

The remaining difference from \Cref{app:near-upper} is that each
gap is estimated using \Cref{thm:energy-single-qubit} instead of \Cref{thm:energy-gap-subroutine}. Since the required gap accuracy is
$\eta=\Theta(\eps/3^{k/2})$, one gap estimate costs
$\widetilde{\mathcal O}(3^{k/2}/\eps)$ evolution time. Multiplying by the number of samples above gives
\begin{equation}
    T_{\rm tot}
    =
    \widetilde{\mathcal O}\left(
        \frac{3^{3k/2}\max\{1,d_j\}}{\eps}
        \log\frac{n}{\delta}
    \right).
\end{equation}
Finally, substituting the same accuracy $\eta$ into the step-size guarantee of \Cref{thm:energy-single-qubit}, in place of the general-control step size used in \eqref{eq:step-size}, gives
\begin{equation}
    \tau
    =
    \Theta\left(
        \frac{1}{
            k\Lambda
            \left(
                1+\log(k3^{k/2}\Lambda/\eps)
            \right)^2
        }
    \right).
\end{equation}

However, the local sparsity $d_j$ is unknown, so we cannot determine in advance how many samples are needed to guarantee \eqref{eq:h-estimate-l2}. Consequently, for a given sparsity estimate, we do not know whether the recovered vector $\widehat h^{(j)}$ actually satisfies \eqref{eq:h-estimate-l2}. We therefore introduce a verification procedure that checks the accuracy of each candidate before accepting it. As shown in \eqref{eq:app-verification-tot}, this verification step increases the total evolution time by only an additional factor of $c^k$ for a universal constant $c$. The following fixed-coordinate version of \Cref{thm:cert-single} provides the required verification.

\begin{corollary}\label{col:cert}
    Let $H=\sum_{P\in\Pnk}h_PP$ be traceless, $k$-local, and satisfy $\norm{H}_{\op}\leq\Lambda$. Given $k$-local Hamiltonian $H_0=\sum_{P\in\Pnk}\widehat{h}_PP$ and $j\in[n]$, for every $\eps\in(0,\Lambda)$ and $\delta\in(0,1/3)$, promised that one of the following conditions holds, one determines whether
    \begin{equation}\label{eq:col-cert-cond}
        \textit{Far}: \|h^{(j)}-\hat{h}^{(j)}\|_2 \ge \eps \quad \text{or} \quad \textit{Close}: \|h^{(j)}-\hat{h}^{(j)}\|_2 \le \eps/c_1^k
    \end{equation}
    with probability at least $1-\delta$, for universal constants $c_1,c_2>1$. The total evolution time is $\widetilde{\mathcal O}(c_2^k\eps^{-1}\log(1/\delta))$, and the step size is 
    \begin{equation}
        \Theta\left(\frac{1}{k\Lambda(k+\log(\Lambda/\eps))^2}\right).
    \end{equation}
    The algorithm only uses single-qubit gates and measurements with forward evolution $e^{-iH\tau}$. No guarantee is imposed on the output when neither promise holds.
\end{corollary}
\begin{proof}
Fix $j$ and set $A_j:=\sum_{P:j\in\supp(P)}(h_P-\widehat h_P)P$. The gap under $H-H_0$ when qubit $j$ is flipped equals $-2\bra{x}A_j\ket{x}$. Its second moment lies between $4\cdot3^{-k}\norm{A_j}_{\F}^2$ and $4\norm{A_j}_{\F}^2/3$, and its fourth moment is at most $16\norm{A_j}_{\F}^4$, by the calculations in \Cref{lem:far-gaps}. Thus the proof of \Cref{thm:cert-single} applies with $j$ fixed, $N=\lceil64k9^k\rceil$, and gap accuracy $\eps/(8\cdot3^{k/2})$. It gives $c_1=12$ and total evolution time $\widetilde{\mathcal O}(3^{5k/2}/\eps)$; a sufficiently large universal $c_2$ suffices. Independent repetition and majority vote give failure probability $\delta$ with an additional factor $\mathcal O(\log(1/\delta))$. Only gaps of $H$ are estimated, and the known $H_0$ gaps are subtracted classically.
\end{proof}

We use \Cref{col:cert} to determine when the sparsity guess is large enough. Starting from $\widehat d_j=1$, we consider the guesses $\widehat d_j=1,2,4,\ldots$. For each guess, we run sparse recovery with target error $\eps/c_1^k$ and verify the resulting estimate $\widehat h^{(j)}$ with far threshold $\eps$, using measurement data independent of those used for recovery. If the verifier outputs \textit{Close}, we accept the estimate; otherwise, we double $\widehat d_j$ and continue. If recovery is infeasible, we also proceed to the next guess. The final guess is capped at the number of candidate Pauli strings containing $j$; if no candidate is accepted at this guess, we report failure.

The recovery accuracy $\eps/c_1^k$ is chosen to match the \textit{Close} condition in \Cref{col:cert}. Provided the recovery and verification guarantees hold, once $\widehat d_j\geq d_j$, sparse recovery gives $\norm{h^{(j)}-\widehat h^{(j)}}_2\leq\eps/c_1^k$, so the verifier accepts. Hence the final guess is at most $2\max\{1,d_j\}$. Conversely, every accepted estimate has error below $\eps$, even if the procedure stops before the guess reaches $d_j$.

At round $\ell\geq0$ and coordinate $j$, assign failure probability at most $\delta/[16n(\ell+1)^2]$ to each recovery and verification procedure. Use independent data in each round. Conditional on earlier outcomes, the recovery guarantee applies when $\widehat d_j\geq d_j$, and the verification guarantee applies whenever either promise holds. Since $\sum_{\ell\geq0}(\ell+1)^{-2}<2$, a union bound shows that all these guarantees hold with probability at least $1-\delta/2$. The following cost bounds hold whenever they do.

Because the sparsity guesses grow geometrically, the total recovery
cost is dominated by the final guess. The verification procedure is
used only $\mathcal O(\log(2+d_j))$ times, each with total evolution
time $\widetilde{\mathcal O}(c_2^k/\eps)$. Combining the recovery and
verification costs gives
\begin{equation}\label{eq:app-verification-tot}
\begin{aligned}
    T_{\rm tot}^{(j)}
    &=
    \widetilde{\mathcal O}\left(
        \frac{
            3^{3k/2}c_1^k\max\{1,d_j\}
            +c_2^k\log(2+d_j)
        }{\eps}
        \log\frac n\delta
    \right) \\
    &=
    \widetilde{\mathcal O}\left(
        \frac{c^k\max\{1,d_j\}}{\eps}
        \log\frac n\delta
    \right)
\end{aligned}
\end{equation}
for a universal constant $c>1$.

Applying this procedure to every $j\in[n]$ and selecting one local
estimate for each Pauli coefficient gives
$\norm{h-\widehat h}_\infty<\eps$ with total evolution time
\begin{equation}
    T_{\rm tot}
    =
    \sum_{j=1}^n T_{\rm tot}^{(j)}
    =
    \widetilde{\mathcal O}\left(
        \frac{c^k(n+ks)}{\eps}
        \log\frac n\delta
    \right),
\end{equation}
where we used $\sum_jd_j\leq ks$. Replacing the local target error
by $\eps/\sqrt n$ similarly gives
$\norm{h-\widehat h}_2\leq\eps$ with total evolution time
$\widetilde{\mathcal O}
(c^k(n+ks)\sqrt n\,\eps^{-1}\log(n/\delta))$.

However, applying the preceding recovery procedure independently to all $n$ coordinates introduces the additive $n$ term in the total evolution time. To remove this dependence and match the total evolution times in \eqref{eq:tot-single-infty} and \eqref{eq:tot-single-2}, we first identify only the
coordinates that are relevant for recovering the Hamiltonian.

Let $\mathcal U:=\bigcup_{P:h_P\neq0}\supp(P)$ be the set of qubits on which at least one nonzero Hamiltonian term acts. Since $H$ is $s$-sparse and $k$-local, $|\mathcal U|\leq ks$. For a target local accuracy $r\in(0,\Lambda)$, we construct a set $J$ satisfying
\begin{equation}\label{eq:discovered-coordinates}
    \{j:\norm{h^{(j)}}_2>r\}\subseteq J\subseteq\mathcal U.
\end{equation}
Thus, $J$ contains every coordinate whose associated coefficient vector has norm larger than $r$, while containing no coordinate on
which the Hamiltonian acts trivially. Once such a set $J$ is found, we run the sparse-recovery procedure only for $j\in J$. This is sufficient for both error guarantees. Indeed, if a Pauli coefficient $h_P$ is not recovered from any coordinate in $J$, then every $j\in\supp(P)$ satisfies $\norm{h^{(j)}}_2\leq r$, and hence $|h_P|\leq r$. Therefore, the remaining coefficients can be set to zero without increasing the $\ell_\infty$ error beyond $r$.

It remains to identify $J$ without testing all $n$ coordinates separately. We do this by randomly selecting subsets $S\subseteq[n]$ and testing all coordinates in $S$ using the same Hamiltonian evolutions. When a relevant coordinate $j$ is selected without any other coordinate in $\mathcal U$, the experiment reduces to the single-coordinate verification procedure of \Cref{col:cert}. Repeating this experiment sufficiently many times allows us to identify every coordinate with $\norm{h^{(j)}}_2>r$ while keeping $J\subseteq\mathcal U$. We describe this procedure next.

\paragraph{Parallel coordinate discovery.}
We use the verification procedure in \Cref{col:cert} with $H_0=0$ and
far threshold $r$. Instead of testing the coordinates one at a time,
we randomly select several coordinates and test them using the same
Hamiltonian evolutions. Start with $J=\varnothing$ and perform
\begin{equation}\label{eq:discovery-repetition-count}
    N:=\left\lceil12ks\log\frac{ks}{\delta}\right\rceil
\end{equation}
independent experiments. In each experiment, form a random subset $S\subseteq[n]$ by including each coordinate independently with probability $1/(2ks)$. Keep $S$ fixed throughout that verification experiment.

For a fixed random subset $S$, our goal is to test all coordinates $j\in S$ using the same Hamiltonian evolution, rather than running the verification procedure separately for each $j$. Recall that the single-coordinate verifier for $j$ uses random $\sigma,x$ and compares the two product states $\ket{x}_\sigma$ and $\ket{x\oplus e_j}_\sigma$. To perform these tests in parallel, use the same random pairs $(\sigma,x)$ for all selected coordinates. For each pair, let $W$ be the product unitary satisfying $W\ket{v}=\ket{x\oplus v}_\sigma$, prepare $\ket{\psi_S}:=\bigotimes_{j\in S}\ket{+}_j\bigotimes_{j\notin S}\ket{0}_j$, and use the product control
\begin{equation}\label{eq:discovery-product-control}
    V_S:=
    \exp\left[-i\gamma\left(
        \sum_{j\in S}N_j+2\sum_{j\notin S}N_j
    \right)\right],
    \qquad
    F_{\pm,S}(t):=V_S^{\pm1}W^\dagger e^{-iHt}W,
\end{equation}
where $N_j=\ketbra{1}{1}_j$ and
$\gamma=2\pi(\sqrt2-1)$.

For each required $t,m$, apply $F_{-,S}(t)^mF_{+,S}(t)^m$ to $\ket{\psi_S}$, with all evolution times and measurement counts fixed in advance. After each circuit, all qubits in $S$ are measured simultaneously in the $X$ basis, with independent preparations for the $Y$-basis measurements. Thus a single sequence of Hamiltonian evolutions produces the
measurement results for every $j\in S$. We apply the decision rule
of \Cref{col:cert} separately to each selected qubit and add $j$ to
$J$ whenever the result is \textit{Far}.

We now explain why the parallel experiment correctly tests a coordinate $j$ when $S\cap\mathcal U=\{j\}$. Recall that $\mathcal U=\bigcup_{P:h_P\neq0}\supp(P)$ contains all qubits on which $H$ acts nontrivially. Therefore, qubits outside $\mathcal U$ do not affect the Hamiltonian evolution. If $S\cap\mathcal U=\{j\}$, then among the qubits on which $H$ acts, only qubit $j$ is prepared in $\ket{+}$; every other qubit in
$\mathcal U$ is prepared in $\ket{0}$. Likewise, the control $V_S$ applies the phase coefficient $1$ to qubit $j$ and coefficient $2$ to every other qubit in $\mathcal U$. Thus, restricted to the qubits in $\mathcal U$, the preparation, control, and measurement are exactly those of the single-qubit verification procedure with probe qubit
$j$. Hence, \Cref{col:cert} applies to this measurement data, and when $\norm{h^{(j)}}_2>r$, the verifier outputs \textit{Far} with its stated success probability.

On the other hand, if $j\notin\mathcal U$, the Hamiltonian acts
trivially on qubit $j$. The factors of $V_S$ and $V_S^\dagger$ on
that qubit therefore cancel, so its measurement statistics are the
same as for the zero Hamiltonian. Consequently, each selected coordinate outside $\mathcal U$ is classified as \textit{Close} with the assigned success probability. If several coordinates in $\mathcal U$ are selected at the same time, we do not require any guarantee for their
individual outcomes; any such coordinate still belongs to
$\mathcal U$, so adding it to $J$ does not violate
$J\subseteq\mathcal U$.

\paragraph{Success probability and total evolution time.}
The parallel experiment shows that, whenever
$S\cap\mathcal U=\{j\}$, the experiment correctly applies the single-coordinate verification procedure to $j$. We now show that every coordinate with $\norm{h^{(j)}}_2>r$ satisfies this condition
in at least one of the $N$ experiments with high probability.

For a fixed $j\in\mathcal U$,
\begin{equation}
    \Pr[S\cap\mathcal U=\{j\}]
    =
    \frac{1}{2ks}
    \left(1-\frac{1}{2ks}\right)^{|\mathcal U|-1}
    \geq
    \frac{1}{4ks},
\end{equation}
where we used $|\mathcal U|\leq ks$. Hence, the probability that
$S\cap\mathcal U\neq\{j\}$ in all $N$ experiments is at most
\begin{equation}
    \left(1-\frac1{4ks}\right)^N
    \leq
    e^{-N/(4ks)}
    \leq
    \frac{\delta}{4ks},
\end{equation}
by the choice of $N$ in \eqref{eq:discovery-repetition-count}.
A union bound over at most $ks$ coordinates with
$\norm{h^{(j)}}_2>r$ shows that all such coordinates are tested
under the condition $S\cap\mathcal U=\{j\}$ at least once with
probability at least $1-\delta/4$.

Assign failure probability $\delta/(4nN)$ to each selected coordinate's test. For fixed sampled subsets, apply the two preceding guarantees only when $j\notin\mathcal U$, or when $S\cap\mathcal U=\{j\}$ and $\norm{h^{(j)}}_2>r$. There are at most $nN$ such tests, so a union bound limits their total failure probability to $\delta/4$. Combining this with the subset-selection bound proves
\eqref{eq:discovered-coordinates} with probability at least
$1-\delta/2$.

By \Cref{col:cert}, each experiment requires $\widetilde{\mathcal O}(c_2^k r^{-1}\log(nN/\delta))$ evolution time. All qubits in $S$ are measured after the same evolutions, so there is no additional factor of $|S|$. Multiplying by $N=\mathcal O(ks\log(ks/\delta))$ gives
\begin{equation}\label{eq:parallel-discovery-time}
    T_{\rm disc}
    =
    \widetilde{\mathcal O}\left(
        \frac{c_2^kks}{r}
        \log\frac{n}{\delta}
    \right).
\end{equation}

\paragraph{Recovery on the discovered coordinates.}
The preceding procedure constructs a set $J$ satisfying \eqref{eq:discovered-coordinates} with probability at least $1-\delta/2$. We now recover $h^{(j)}$ only for $j\in J$, using the
single-coordinate recovery procedure above with local error at most
$r$. We allocate the remaining failure probability $\delta/2$ to
these recovery and verification procedures. When \eqref{eq:discovered-coordinates} holds, every $j\in J$ has $d_j\geq1$, and $\sum_{j\in J}d_j\leq ks$. Combining the recovery cost in
\eqref{eq:app-verification-tot} with the discovery cost in
\eqref{eq:parallel-discovery-time} gives
\begin{equation}\label{eq:single-learning-after-discovery}
    T_{\rm tot}
    =
    \widetilde{\mathcal O}\left(
        \frac{c^ks}{r}\log\frac{n}{\delta}
    \right)
\end{equation}
for a universal constant $c>1$, where polynomial factors in $k$
are suppressed.

We next combine the local estimates into a single estimate of the
Hamiltonian coefficients. For every Pauli $P$ with
$J\cap\supp(P)\neq\varnothing$, choose one
$j(P)\in J\cap\supp(P)$ and set
$\widehat h_P:=\widehat h_P^{(j(P))}$.
If $J\cap\supp(P)=\varnothing$, set $\widehat h_P:=0$.
For the first case, the local recovery guarantee gives
$|\widehat h_P-h_P|\leq r$. For the second case,
\eqref{eq:discovered-coordinates} implies
$\norm{h^{(j)}}_2\leq r$ for every $j\in\supp(P)$, and hence
$|h_P|\leq r$. Therefore,
$\norm{\widehat h-h}_\infty\leq r$.

For the $\ell_2$ error, assign each recovered coefficient to the
coordinate $j(P)$ from which it is estimated, and assign each omitted
nonzero coefficient to any coordinate in its support. This gives
\begin{equation}\label{eq:single-learning-assembled-error}
    \norm{\widehat h-h}_2^2
    \leq
    \sum_{j\in J}
        \norm{\widehat h^{(j)}-h^{(j)}}_2^2
    +
    \sum_{j\notin J}
        \norm{h^{(j)}}_2^2
    \leq
    nr^2.
\end{equation}
Thus, taking $r=\eps$ proves \eqref{eq:tot-single-infty}, while
taking $r=\eps/\sqrt n$ proves \eqref{eq:tot-single-2}. Together with
the probability guarantee for constructing $J$, the overall success probability is at least $1-\delta$. We stop and report failure before exceeding the budget in \eqref{eq:single-learning-after-discovery}, with a sufficiently large fixed constant. This does not occur when the preceding guarantees hold.

Finally, all parts of the algorithm can use a common step size. Choose the extrapolation order $p$ for gap accuracy $C^{-k}r$, with $C>1$ large enough for discovery, verification, and recovery to error $r/c_1^k$. Choose $L$ by \eqref{eq:L-choice} for the smallest resulting node accuracy, not just for the gap accuracy. Both orders are $\mathcal O(k+\log(\Lambda/r))$. Keep the weights and time nodes fixed; for coarser estimates, increase the node tolerances and adjust the phase-estimation repetitions. The same bias bound and the same $L$ remain valid. Thus
\begin{equation}
    \tau=\Theta\left(\frac{1}{k\Lambda(k+\log(\Lambda/r))^2}\right)
    =\widetilde\Theta(1/\Lambda)
\end{equation}
for fixed $k$. All state preparations, controls, and measurements are single-qubit operations, and every unknown evolution time is a positive integer multiple of $\tau$.

\end{document}